\documentclass[10pt,journal,compsoc]{IEEEtran}

\ifCLASSOPTIONcompsoc
  \usepackage[nocompress]{cite}
\else
  \usepackage{cite}
\fi
\ifCLASSINFOpdf
  \usepackage[pdftex]{graphicx}
  \DeclareGraphicsExtensions{.pdf}
\else
\fi
\usepackage{algorithmic}

\usepackage{array}

\usepackage[]{subfig}

\usepackage{url}

\usepackage{dsfont}
\usepackage{amsthm}
\usepackage{amsfonts}
\usepackage{todonotes}

\usepackage{mathtools}
\usepackage{amssymb}

\usepackage[ruled,vlined, linesnumbered,noresetcount]{algorithm2e}

\usepackage{xcolor}
\usepackage[normalem]{ulem}
\DeclareMathOperator*{\argmax}{argmax}
\DeclareMathOperator*{\argmin}{argmin}

\usepackage{lineno}

\usepackage{soul}

\newtheorem{prop}{Proposition}

\begin{document}
%
\title{Generative Models for Periodicity Detection in Noisy Signals}
%
%
%
%

\author{~Ezekiel~Barnett,
        Olga~Kaiser,
		Jonathan~Masci,
		Ernst~Wit,
		Stephany~Fulda
\IEEEcompsocitemizethanks{
\IEEEcompsocthanksitem E. Barnett , Jataware Corporation.\protect\\
E-mail:{ezekbarnett}@gmail.com
\IEEEcompsocthanksitem O. Kaiser (corresponding author) and J. Masci, NNAISENSE, Lugano, Switzerland.\protect\\
E-mail:{olga, jonathan}@nnaisense.com
\IEEEcompsocthanksitem E. Wit, Institute of Computational Science of Università della svizzera italiana, Lugano, Switzerland.\protect\\
E-mail:{ewit}@usi.ch
\IEEEcompsocthanksitem S. Fulda,Sleep Medicine Unit, Neurocenter of Southern Switzerland, Civic Hospital, EOC, Lugano, Switzerland.\protect\\
E-mail:{stephany.fulda}@gmail.com

}
\thanks{XXX}}

\IEEEtitleabstractindextext{%
\begin{abstract}
We introduce a new periodicity detection algorithm for binary time series of event onsets, the Gaussian Mixture Periodicity Detection Algorithm (GMPDA). The algorithm approaches the periodicity detection problem to infer the parameters of a generative model. We specified two models - the Clock and Random Walk - which describe two different periodic phenomena and provide a generative framework. The algorithm achieved strong results on test cases for single and multiple periodicity detection and varying noise levels. The performance of GMPDA was also evaluated on real data, recorded leg movements during sleep, where GMPDA was able to identify the expected periodicities despite high noise levels. The paper's key contributions are two new models for generating periodic event behavior and the GMPDA algorithm for multiple periodicity detection, which is highly accurate under noise.%
\end{abstract}

}

\maketitle


%

\section{Introduction}\label{sec:introduction}
From heartbeats to commutes, global climatic oscillations to Facebook log-ons, periodicity - the phenomena that events happen with regular intervals - is omnipresent. Detecting periodicity in time series is often referred to as "the periodicity detection problem." In the case of \textit{event} time series - binary time series, which indicate only the occurrence of some event - the periodicity detection problem has been approached using algorithms such as Fast Fourier Transform (FFT) and Auto-correlation.
And typically has been formulated in the context of a single, stationary periodicity~\cite{welch1967use, priestley1981spectral, madsen2007time, mitsa2010temporal, box2015time}. Several issues, however, have not been sufficiently addressed in the literature such as (i) the development of generative models which appropriately describe "noise" in periodic behavior - in addition to false positives and negatives - as variance in interval length and the challenges of (ii) multiple overlapping periods, and (iii) non-stationary periodic signals.

Existing periodicity detection algorithms are often based on the FFT or the Auto-Correlation Function (ACF) and focus on single period detection. FFT maps a time series to the frequency domain, and one would typically use the inverse of the frequency with the strongest power as the predicted period. FFT is sensitive to sparse data~\cite{junier}, to noise~\cite{Glynn}, and even in the absence of noise, suffers from "spectral leakage" for low-frequencies/large periods~\cite{Vlachos}. Other approaches include the Lomb-Scargle periodogram~\cite{Glynn, Ahdesmaki} - a least-squares-based method for fitting sinusoids (to deal with noise) and unevenly sampled data - that shares the same problems with the FFT. In real applications, however, the hierarchy implied by the FFT may not be appropriate to describe the signal, especially when the periodic signals are random walks with Markov properties and the signal is non-stationary.

ACF-based methods estimate the similarity between sub-sequences of event intervals that have been extracted with a set of lags and periods are selected as the lags, which maximize the ACF. ACF-based methods have been employed for multiple periodicity detection in character series such as texts, for instance, in~\cite{Berberidis, Elfeky}. Auto-correlation detects a large number of candidate "periods" (especially integer multiples), many of which hardly differ from each other, which thereby necessitates a self-selected significance threshold for selecting "true periodicities." ACF also suffers from smaller data. In addition, these methods are typically not designed for finding multiple periodicities in event time series.

Outside of the FFT/ACF framework, E-periodicity \cite{li} is a method for single period detection based on the modulus operation, with a primary focus on periodicity detection in unevenly/under-sampled time series. E-periodicity finds the interval around which the modulus operator of the event time-stamps is minimal. Essentially, the algorithm segments the time series into all possible periodicities within some a-priori specified range. It then overlays the segments and selects the true periodicity as the periodicity that "covers" the most events.

Other methods that have been developed for single period detection include partial periodic patterns and a chi-squared test~\cite{ma}, max sub-pattern tree~\cite{sheng}, and projection~\cite{yang}. Like FFT and ACF-based methods, these methods struggle in the presence of low-frequency periodicities and low sampling rates~\cite{quan}. They are only designed to recognize a single periodic pattern in stationary signals.

Little work has been done in the area of multiple period detection in time-series event data. Most multiple-periodicity methods use a hierarchical extraction method where the frequency with the highest power (in the case of FFT) or most probable periodicity (in the case of ACF) is selected and removed iteratively. FFT is a natural choice to disentangle multiple elements of a complex function and has been used by~\cite{vandongen} in the Lomb-Scargle framework to detect multiple periods with a hierarchical extraction method, but is not designed for event data. Another approach for multiple period detection uses ACF to identify periods, and then subtracts them from the original signal using a comb filter~\cite{parthasarathy}.

An alternative approach in \cite{xu} focuses on computing a set of possible periodicities using intervals between events and selecting the set of periods as intervals above some threshold. However, this threshold is not well defined, and there is no method for dealing with noise. It also struggles with smaller data\cite{xu}.

The authors in~\cite{ghosh2017finding} present a generative model for discrete signals with a Gaussian probability density function (PDF) of the period and a Poisson process for describing the false positive noise events. The presented approach is adaptive to noise and a changing periodicity but cannot detect multiple, simultaneously overlapping periods.

A different approach in the research field of "periodic pattern-finding" finds "time slots" when events of a particular periodicity occur~\cite{quan}, i.e., locations, spaced periodically on the time series, where an event may happen. This approach is suited for data where certain events are expected to happen at certain sub-sequences of time (for instance, a student logging onto a computer every Wednesday and Thursday between 14:00 and 17:00). The algorithm aims to detect multiple-interlaced periodicities and relies on a scoring function and a heuristic algorithm to maximize the objective function to solve the NP-hard problem. The method does not account for variance in the event location. It is particularly optimized for anomaly detection tasks (when a periodic behavior is broken) and time series with a low sampling rate.

To address some of the challenges of multiple periodicity detection for noisy event time series, we propose the Gaussian Mixture Periodicity Algorithm (GMPDA). The algorithm is based on a novel generative model scheme of periodic event time series, which implicates variability in interval length through Gaussian distributed noise. Here, we compared the GMPDA to other algorithms on a large set of test cases and reported superior performance of GMPD the accuracy, sensitivity, and computational performance with outperforming results in most cases.

The rest of the paper is organized as follows. In Section~\ref{sec:Generative_Models} we introduce the generative models and discuss their inference. Section~\ref{sec:GMPDA_ALGO} presents the GMPDA algorithm. The performance of the GMPDA framework was tested in Section~\ref{sec:test_cases}. An application of GMPDA to real data can be found in Section~\ref{sec:real_application}. Section~\ref{sec:conclusion} concludes this paper.

\section{Generative Models}
\label{sec:Generative_Models}
Consider a uni-variate event time series $\left(X_t\right)_{t=1,\dots,N_T}$ where $x_t = 1$ if the event starts at time $t$ and else $x_t = 0$. In this work, we ignore the case of un-sampled/missing data. Then, the information in $X_t$ can be compressed to the set of non-zero/positive time stamps $S:=\lbrace s_i| x_{s_i} = 1\rbrace_ {i=1,\dots,N_S}$.

If the positive time stamps occur (at least partially) at regular intervals, the time series exhibits a periodic behavior, and the regular intervals correspond to periodicities or periods. We formulate the periodicity detection problem to search for the set of periodicities that explain the intervals between timestamps in $S$.

We are particularly interested in the set of \textit{prime} periodicities, that is, the minimum integer frequency that describes the intervals. For instance, for a timestamp set $S=\lbrace 12, 23, 34, 45, 56, 67\rbrace$, many intervals could be explained by a periodicity of $22$ or $33$, but $11$ would be the prime, which explains the most data and $22$ and $33$ are integer multiples of this prime period. In the following, the set of underlying prime periods in $X_t$ is denoted by $\mu^{*}=\{\mu_p^{*},p= 1,\dots,P\}$. In addition, we assume that for most real applications, the interval between two consecutive time stamps associated with a periodicity $\mu_p^{*}$ in $S$ will most probably vary from $\mu_p^{*}$  with a variance denoted by $\sigma_p^{*2}$, for instance, $s^p_{i+1}-s^p_i \in \lbrack \mu_p^{*}-\sigma_p^{*}, \mu_p^{*}-\sigma_p^{*}\rbrack$.

If the time series $X_t$ is generated by a single, stationary periodicity $\mu_1^{*}$, we can compute $\mu_1^{*}$ and thus the the prime periodicity $\mu^*$ directly from the data as:
\begin{align}
    \label{eq:mu_prime1}
    \mu^{*} &= \frac{N_T}{|S|} = \frac{\sum_{i=1}^{N_T} s_{i+1} - s_i}{|S|}.
\end{align}
Please note, the first equality in (\ref{eq:mu_prime1}) is the ratio between the length of the time series and total number of events. The second equality in (\ref{eq:mu_prime1}) describes the "average interval" between two adjacent time stamps and holds for a time series with a single, stationary periodicity without noise. The associated variance $\sigma^{*2}$ can be estimated as the square of the standard deviation.

However, estimation of $\mu^{*}$ and $\sigma^{*2}$ according to equation (\ref{eq:mu_prime1}) will not be sufficient when (i) the time series $X_t$ is generated by multiple, overlapping periodicities, (ii) the time series $X_t$ is noisy, i.e., there are false positives, (iii) when there are missing values (false negatives), or (iv) when there are different patterns of periodic behavior over time, i.e., non-stationarity.


The generative model that we present here accounts for challenges (i) and (ii): specifically, we formulate a generative model of the positive time stamps $S$ with multiple periodicities, an explicit term that incorporates noise, and a loss that enables inference of the model parameters.

Let us assume, that the set of positive time-stamps $S$ can be generated by a function $f$, as:
\begin{align}
\label{eq:general_generative}
    S = f(\mu^*, \sigma^*, \beta, \alpha, M),
\end{align}
where:
\begin{itemize}
    \item $\mu^*$ is the set of $P$ prime periodicities in the time series,
    \item $\sigma^*$ is the set $P$ variances of the periodic intervals,
    \item $\beta$ is the rate of false positive noise in a Bernoulli sense,
    \item $\alpha_{p}$ is the starting point of periodicity $p$,
    \item $M$ is the generative model scheme.
\end{itemize}
The generative model scheme $M$ is characterized by the priors for the distribution of the intervals, its mean values $\mu^*$, and the variances $\sigma^{*2}$. We follow the generative approach in equation (\ref{eq:general_generative}) and assume that a single event $s_i$ is generated according to one periodicity (except in the case of overlaps) or false positive noise $\beta$. Then, the set $S$ is the union of subsets $S^{p}$ of positive time stamps $s_i$ associated with periodicity $\mu_p$, or random noise $\beta$:
\begin{equation}
    \label{eq:general_generative_mm}
    S = S^{\mu_1^*}\cap S^{\mu_2^*}\dots \cap S^{\mu_P^*} \cap S^{\beta}.
\end{equation}
Further, without loss of generality, we parameterize the distribution of the intervals by the Gaussian distribution; any other distribution, for instance, a member of the exponential family would also be appropriate. In Section~\ref{subsec:CM} and Section~\ref{subsec:RWM} we formulate two different model schemes denoted in the following as the \textit{Clock Model} ($M=C$) and the \textit{Random Walk Model} ($M = RW$).
%
\subsection{Clock Model}
\label{subsec:CM}
The \textit{"Clock Model"} describes a periodic behavior governed by a fixed period $\mu_p^{*}$ with Gaussian noise, which does not incorporate information from previous positive time-stamps when computing the occurrence of the next event, i.e., for $p=1,\dots,P$ the events in $S^p$ are generated by:
\begin{align}
    \label{eq:cl_model}
    s_{i}^p = \alpha_p + (i \cdot \mu_p^{*}) + \epsilon, \,\epsilon \sim N(0, \sigma_p^{*2}).
\end{align}
The number of events associated with uniformly distributed false positive noise is given as $\beta * |S^{\mu^*}|$ in the interval $[0, N_T]$.

Note that for the Clock Model, the location of any event only depends on the location in the time series and Gaussian noise around some regular location, but does not depend on previous time steps. Accordingly, one can predict with equal accuracy any time step $s_{i + m}^p$ for $m>0$. This formulation is a generalization of generative models in much of the previous work on periodicity detection, e.g., \cite{quan} and \cite{li}.
In their formulation, one needs to find a time-slot $s_i$ as a pair of a period ($l$) and an offset $i$, denoted by $[l : i]$. This formulation is equivalent to finding a period $\mu_p^*$ and a starting point $\alpha_p$, with $\sigma^* = 0$, which might be a limitation in real applications, as this formulation does not allow for potential variability in the realization of event locations in the time series. In order to account for this, Gaussian noise $\sigma^*$ is added (the formulation with $\sigma^* = 0$ would be a special case of the Clock Model).%

However, the notion of an external pacemaker is a realistic expectation only for some systems, thus motivating the development of the Random Walk Model.
%
\subsection{Random Walk Model}
\label{subsec:RWM}
%
The Random Walk Model, exhibits the Markov property, i.e., a system where the next event's temporal location depends on the current event's temporal location and Gaussian noise. For $p=1,\dots,P$ the events in $S^p$ are:%
\begin{align}
    \label{eq:rw_model}
    s_{i+1}^p = s_{i}^p + \mu_p^{*} + \epsilon,\,\epsilon \sim N(0, (i\sigma_p)^{*2}).
\end{align}
Again, the number of events associated with uniformly distributed false positive noise is given as $\beta * |S^{\mu^*}|$ in the interval $[0, N_T]$.%

As the noise is Gaussian (and thus is identically distributed), the series of event-time stamps in $S^{\mu_p}$, for $p=1,\dots,P$, describes a random walk. And thus, the formulation in equation (\ref{eq:rw_model}) is hereafter referred to as the Random Walk Model (RWM).%

RWM has the property that, concerning some event, $s_i$, the $\sigma^*$'s add up for each subsequent time step. Therefore, the distribution's variance of the expected location of an event further in the future increases linearly with the distance from the current event. This assumption is an essential and a realistic expectation for many real-life systems, which have no pacemaker and are therefore predictable with decreasing accuracy as steps increase, i.e., given $s_i$, we can predict $s_{i+1}$ more accurately than $s_{i+10}$.%
%
\subsection{Inference}
\label{subsec:Inference}
%
Given an event time series $X_t$, a straightforward approach to extract the possible periodicities is to study the empirical histogram of all pairwise, forward order inter-event intervals. For every event, we consider not only the interval to the next event (onset to onset) but also to all subsequent events.%

The possible range of the intervals is defined by $(0,N_T)$\footnote{In real applications, the actual range is smaller as an interval needs to be observed a minimal amount of times to be significant.}. On the other hand, we can estimate a histogram of expected all forward order inter-event intervals with respect to one of the generative models defined by equations (\ref{eq:cl_model}) and (\ref{eq:rw_model}). %
This histogram is obtained by (i) estimating analytically the expected number of intervals for each $\mu_p^*\in\mu^*$, (ii) incorporating all the intervals between any pair of events associated with any two different prime periodicities $\mu_p^*$ and $\mu_q^*$, and (iii) by incorporating all the intervals due to noise (in the case the source of the noise is known this could be done analytically, otherwise an estimate is required). The comparison of the empirical histogram and the parametric expectation will define the loss function used to identify the optimal underlying periodicities. %

In the following, we define for every $\mu\in (0,N_T)$ the function of interval counts $D(\mu)$ by:%
\begin{align}
\label{eq:dmu_prob}
    D(\mu) = \sum_{i, m > 0 } \mathds{1}_{s_{i + m}^p - s_i^p = \mu}.
\end{align}
The evaluation of $D(\mu)$ for a given $X_t$ results in a histogram of all pairwise inter-event intervals.

The generative models provide a statistical model for the intervals. Thus, we can estimate the expected number of intervals for $\mu\in (0,N_T)$ in reference to a fixed periodicity $\mu_p^*$ and variance $\sigma_p^*$ as:
\begin{align}
    \mathbb E [D(\mu)]_{\mu_p^*}
    &= \sum_{i, m > 0 } \mathbb E[\mathds{1}_{\mu}]\label{eq:e_dmu_1}\\
    &= \sum_{i, m > 0 } \mathbb P[s_{i+ m}^p - s_{i}^p = \mu],\label{eq:e_dmu_2}
\end{align}
equality in equation (\ref{eq:e_dmu_1}) is due to linearity of expectation and equality in equation (\ref{eq:e_dmu_2}) is due to the fact that for a random variable $A$, $\mathbb E[\mathds{1}_{A}] = \mathbb P[A]$.
The distribution of all m-th order inter-events intervals depends on the specific generative model and can be written as
\begin{align}
    \label{eq:prob_interval_cl}
    \mathbb P[s_{i+ m}^p - s_{i}^p = \mu] = \frac{1}{\sqrt{2 \pi\sigma_p^{*2}}} \exp{[-\frac{(\mu - m \mu_p^{*})^2}{2 \sigma_p^{*2} }}],
\end{align}
for the Clock Model, and as
\begin{align}
    \label{eq:prob_interval_rw}
    \mathbb P[s_{i+ m}^p - s_{i}^p = \mu] = \frac{1}{\sqrt{2 \pi \left(m\sigma_p^*\right)^2}} \exp{[-\frac{(\mu - m \mu_p^{*})^2}{2 \left(m \sigma_p^*\right)^2} ]},
\end{align}
for the Random Walk Model. For the latter, the variance grows linearly with the number of steps between events.

Further, we assume that the starting point is zero, i.e., $\alpha_p = 0$. For a time series of length $N_T$ equation (\ref{eq:e_dmu_2}) can be rewritten into a more explicit form by writing the indicator function as a definite quantity. Assuming no missing values on $\left(0, N_T\right)$ for the generative models, we should observe $\frac{N_T}{\mu_p^*}$ first order intervals ($m = 1$) in the time series, distributed according to the Gaussian probability density function (PDF) parametrized by $\sigma^*$ and $\mu_p^*$. For the second order intervals ($m = 2$) the scaling factor would be $(\frac{N_T}{\mu_p^*} - 1)$, for $m = 3$,  $(\frac{N_T}{\mu_p^*} - 2)$ and so on. Thus, for a single periodicity $\mu_p^*$ the expected value of $D(\mu)$ can be written for the Clock Model as:
\begin{align}
    \label{eq:e_dmu_cm}
   \mathbb E [ D(\mu) ]_{\mu_p^*} =  \sum_{m = 1}^{\frac{N_T}{\mu_p^*}} \frac{const}{\sqrt{2 \pi \sigma_p^*{}^2}} \exp{ [-\frac{(\mu - m \mu_p^{*})^2}{2 \sigma_p^*{}^2} ]},
\end{align}
and for the Random Walk Model as
\begin{align}
    \label{eq:e_dmu_rw}
   \mathbb E [D(\mu)]_{\mu_p^*} = \sum_{m = 1}^{\frac{N_T}{\hat{\mu}_p}} \frac{const}{\sqrt{2 \pi (m\sigma_p^*)^2}} \exp{ [-\frac{(\mu - m\mu_p^{*})^2}{2 (m\sigma_p^*)^2} }],
\end{align}
with $const = \frac{N_T}{\mu_p^*}-(m-1)$. Equations (\ref{eq:e_dmu_rw}) and (\ref{eq:e_dmu_cm}) are therefore the expected values of the function $D(\mu)$ counting all order intervals that might be observed for a single periodicity $\mu_p^*$ for the Random Walk Model and for the Clock Model, respectively.

In the case of multiple, overlapping periods, and/or false positive noise, the set of positive time stamps $S$ would consist of multiple sets: $S^{\mu_1^*}\cap S^{\mu_2^*}\dots\cap S^{\mu_p^*}\cap S^{\beta}$. A-priori, the affiliation of events to periodicities is unknown and therefore we slightly adapt our definition of $D(\mu)$ in equation (\ref{eq:dmu_prob}) by removing the superscript $p$:
\begin{align}
\label{eq:Dhatmup}
    D(\mu) = \sum_{\forall i, m > 0} \mathds{1}_{s_{i+m} - s_i = \mu} .
\end{align}
Thus, the operator $D(\mu)$ counts now not only the intervals between events from the same periodicity set, but also between events in different sets and/or between noise events. We call the latter two "interaction intervals" and denote their contribution to $D(\mu)$ by:
\begin{align}
\label{eq:zeta1}
    \zeta\left(\mu\right) = \sum_{\forall i, m > 0} \mathds{1}_{s_{i+m} - s_i = \mu},
\end{align}
with three possible scenarios (or their combination): (i) intervals between events from different periodicity sets, .i e.,  $s_i \in S^{\mu_p^*}$ and $s_{i+m}\in S^{\mu_q^*}$, (ii) intervals between events from any periodicity set and noise, i.e., $s_i \in S^{\mu_p^*}$ and $s_{i+m} \in S^{\beta}$, (iii) intervals between events due to noise, i.e., $s_{i} \in S^{\beta}$ and $s_{i+m} \in S^{\beta}$.

The estimates in equations (\ref{eq:e_dmu_cm}) and (\ref{eq:e_dmu_rw}) do not include these interaction intervals. In the next step we discuss how to estimate $\zeta\left(\mu\right)$ and to account for the three cases explicitly.
The distribution of the interaction intervals for all the three cases can be obtained in a closed form by applying the convolution formula, which provides the distribution of the sum/difference of two interdependent discrete or continuous random variables~\cite{grinstead2012introduction}. For the Clock and the Random Walk Models we would obtain the following: For case (i) the interaction intervals between two periods, denoted as $\zeta_{pq}(\mu)$, are again Gaussian distributed with a mean $\mu_{pq} = \mu_p - \mu_q$ and a variance $\sigma^2_{pq} = \sigma^2_{p} + \sigma^2_{q}$, where $\mu_p > \mu_q$ without loss of generality. For case (ii) the interaction intervals, denoted as $\zeta_{p\beta}(\mu)$, follow a Gaussian-like distribution, adjusted for the corresponding uniform support. For case (iii) the forward interaction intervals, denoted as $\zeta_{\beta}(\mu)$, follow the right sight of a triangle distribution. In this context, we can write down $\zeta\left(\mu\right)$ as
\begin{align}
\label{eq:zeta_all}
    \zeta\left(\mu\right) = \zeta_{pq}\left(\mu\right) + \zeta_{p\beta}\left(\mu\right) + \zeta_{\beta}\left(\mu\right),
\end{align}
Please note that in real applications we do know neither the amount of noise in real data nor which events are associated to which periodicities, and therefore do not have an exact formulation of (\ref{eq:zeta_all}) and need an approximation of $\zeta_{\beta}\left(\mu\right)$. 
For this, we will assume that the events that contribute to the interaction intervals are uniformly distributed. We will see in Section~\ref{subsub:approx_zeta} that this assumption is not too restrictive.
\begin{prop}
\label{prop:zeta}
For uniformly distributed events on $[1,N_T]$, the expected number of interaction intervals on the interval $[1, N_T]$ is given by:
\begin{equation}
    \label{eq:zeta2}
    \mathbb E[ \zeta(\mu)] = z \cdot (1 - \frac{\mu}{N_T}),
\end{equation}
with a constant $z$, for every $\mu\in[1,N_T]$.
\end{prop}
\begin{proof}
Consider two noise events $s_i^{\beta}$, $s_j^{\beta}$ each with a uniform probability mass function $\mathbb P_{S}$ on the support $[1,\dots,N_T]$. The difference between the events $s_j^{\beta}-s_i^{\beta}=\mu \in [-N_T,N_T]$ is a random variable and its probability mass function can be derived by using the convolution formula for distributions~\cite{grinstead2012introduction}:
%
\begin{align}
    \label{eq:pmf_conv}
    \mathbb P[s_j^{\beta}-s_i^{\beta}=\mu] 
                  = & \sum_{s_j}\mathbb P_S[s_j-\mu] \mathbb P_S[s_j].
\end{align}
Next, as $\mathbb P_{S}$ is defined on $[1,\dots,N_T]$, the probability of $\mathbb P_{S}[s_j-\mu < 1]$ and $\mathbb P_{S}[s_j-\mu > N_T]$ is equal to zero, taking this into account we get:
\begin{align}
    \label{eq:pmf_conv_mid}
    \sum_{s_j}\mathbb P_S[s_j-\mu] \mathbb P_S[s_j]
                  = & (N_T-\mu)\frac{1}{N_T} \frac{1}{N_T}
\end{align}
%
%
%
%
Finally we get the probability mass function as a decaying function of the difference:
\begin{align}
    \label{eq:pmf_conv_fin}
    \mathbb P[s_j^{\beta}-s_i^{\beta}=\mu] = \frac{1}{N_T}(1-\frac{\mu}{N_T}).
\end{align}
In case we would be focusing on $|s_j^{\beta}-s_i^{\beta}|=\mu$ the right hand side of equation (\ref{eq:pmf_conv_fin}) has to be multiplied by $2$, due to symmetry. In the next step we want to estimate the expectation of $\zeta\left(\mu\right)$, that is:
\begin{align}
    \mathbb E[\zeta\left(\mu\right)]
    = &\mathbb E [\sum_{\forall i, m > 0} \mathds{1}_{s_{i+m} - s_i = \mu}]\label{eq:e_zeta_1}\\
    = &\sum_{\forall i, m > 0} \mathbb E [ \mathds{1}_{s_{i+m} - s_i = \mu}]\label{eq:e_zeta_2}\\
    = &\sum_{\forall i, m > 0} \mathbb P [s_j^{\beta}-s_i^{\beta}=\mu]\label{eq:e_zeta_3}.
\end{align}
The equality in equation (\ref{eq:e_zeta_2}) is due to linearity of expectation, equality in equation (\ref{eq:e_zeta_3}) is due to the fact that for a random variable $A$ the following equality holds: $\mathbb E[\mathds{1}_{A}] = \mathbb P[A]$.
%
%
Inserting equation (\ref{eq:pmf_conv_fin}) into equation (\ref{eq:e_zeta_3}) results in
\begin{align}
    \mathbb E[\zeta\left(\mu\right)] = &\sum_{\forall i, m > 0} \frac{1}{N_T}(1-\frac{\mu}{N_T})\label{eq:e_zeta_4}.
\end{align}
The number of all pairwise, forward order differences for the noise events, with $n=N_T\beta=|S^{\beta}|$ is given as:
\begin{align}
     \sum_i(n-i) = n^2-\frac{(n^2+n)}{2},
\end{align}
thus we obtain:
\begin{align}
    \mathbb E[\zeta\left(\mu\right)] = &\frac{2n^2-(n^2+n)}{2N_T}(1-\frac{\mu}{N_T})\label{eq:e_zeta_5}.
\end{align}
By setting $z = \frac{2n^2-(n^2+n)}{2N_T}$ we obtain equation (\ref{eq:zeta2}).
\end{proof}
In real applications the constant $z$ cannot be estimated as the amount of false positives, $\beta$, in unknown a-priori. An approximation for $z$, $\hat{z}$  will be inferred from the data in Section~\ref{subsub:approx_zeta}. For now, the expected number of interaction intervals is approximated via:
\begin{align}
    \label{eq:Ezeta_hat}
    \mathbb E[\hat\zeta(\mu)] = \hat{z} \cdot (1 - \frac{\mu}{N_T}).
\end{align}
In case of multiple periodicities, due to the linearity of the expectation, the expected number of intervals over multiple periods is the sum over the expectation for $D(\mu)$ for each periodicity $\mu_p^*$ present in the data, plus the expected number of the interaction intervals approximated by (\ref{eq:zeta1}):
\begin{equation}
    \label{eq:e_dmu_basic}
    \mathbb E \left[D\left(\mu\right)\right] =  \sum_{p = 1}^P\mathbb E \lbrack D(\mu)\rbrack_{\mu_p^*}+
    \mathbb E\lbrack \zeta\left(\mu\right)\rbrack.
\end{equation}
Hereinafter, the first addend on the right hand side is denoted as the deterministic parametric function $G_{M}\left(\mu;\hat{\mu}, \hat{\sigma}\right)$ for the Clock Model:
\begin{equation}
    \label{eq:P_Mfull_Clock}
    G_{C}(\mu; \hat{\mu}_p, \hat{\sigma_p}) = \sum_{p = 1}^P \sum_{m = 1}^{\frac{N_T}{\hat{\mu}_p}} \frac{c_p}{\sqrt{2 \pi \hat{\sigma_p}^2}} \exp{ [\frac{(\mu - m \hat{\mu}_p)^2}{2 \hat{\sigma_p}^2} ]},
\end{equation}
and the Random Walk Model:
\begin{equation}
    \label{eq:P_Mfull}
    G_{RW}(\mu; \hat{\mu}_p, \hat{\sigma_p}) = \sum_{p = 1}^P \sum_{m = 1}^{\frac{N_T}{\hat{\mu}_p}} \frac{c_p}{\sqrt{2 \pi (m \hat{\sigma_p})^2}} \exp{ [\frac{(\mu - m \hat{\mu}_p)^2}{2 (m \hat{\sigma_p})^2} ]},
\end{equation}
with $c_p = (\frac{N_T}{\hat{\mu}_p} - (m - 1))$.

Once we obtain estimates $\hat{\mu}_p$ and $\hat{\sigma}_p$ for the true periodicities $\mu_p^{*}$ and variances $\sigma_p^*$, and given a prior on the generating function (in our case, Random Walk or Clock), we can write a loss function for our estimates as the difference between the empirical $D(\mu)$ and the parametric $G_{M}\left(\mu;\hat{\mu}, \hat{\sigma}\right)$ functions. The loss function can be either the absolute error or a quadratic loss; since we have deterministic expectations, we focus on the  absolute error as follows:
\begin{align}
    \mathcal{L} &= \sum_{\mu = 1}^{N_T} | D(\mu) - \mathbb E [D(\mu)] |,\\
                &=  \sum_{\mu = 1}^{N_T} | D(\mu)  - ([\sum_{p = 1}^P\mathbb E [D(\mu)]_{\mu_p^*}]  + \mathbb E[ \zeta(\mu) ]),\\
                &\approx  \sum_{\mu = 1}^{N_T} |D(\mu) - G_{M}\left(\mu;\hat{\mu}, \hat{\sigma}\right) -  \mathbb E[\hat\zeta(\mu)]|\label{eq:loss_master}.
\end{align}
Finally, for either the clock model or the random walk model, the aim is to find a set of periodicities and variances that minimize the corresponding loss. A straightforward approach would consider all possible combinations of acceptable periodicities and variances where the optimal combination minimizes the loss.%

However, such an approach is computationally not feasible, and therefore the following section outlines the Gaussian Mixture Periodicity Detection Algorithm (GMPDA). %
%
%
\section{GMPDA Algorithm}
\label{sec:GMPDA_ALGO}
Given an event time series $X_t\in\mathbb R^{(1\times N_T)}$, the aim is (i) to extract an estimate $\hat{\mu}$ of the true generating periodicities $\mathbf{\mu^{*}}$, (ii) to infer $\sigma^*$, and (iv) to test the fit of the chosen generative model $M$. GMPDA provides a method to learn the parameters of the generative function of $X_t$ in an accurate and computationally efficient manner by minimizing the loss $\mathcal{L}$ defined in equation (\ref{eq:loss_master}). The GMPDA algorithm is open-source and available on \url{https://github.com/nnaisense/gmpda}. %
GMPDA is based on comparing $D(\mu)$, the empirical distribution of the intervals observed in the time series $X_t$ with parametrized estimates of its generative function $G_M(\hat{\mu}, \hat{\sigma})$ plus the contribution coming from the interaction intervals, using the loss function (\ref{eq:loss_master}). The main steps of the GMPDA algorithm for the estimation of the optimal parameters $\hat{\mu}, \hat{\sigma}$ are outlined Algorithm~\ref{algo:GMPDA_short}. 
%
%
\begin{algorithm}[h]
\SetAlgoLined
Extract event time-stamps: $S$ $\gets$ where $X_t = 1$\\
Compute intervals $D(\mu)$ from $S$, wrt. equation (\ref{eq:algo_dmu}) and subtracts $\zeta(\mu)$, estimated wrt. equation (\ref{eq:zetaz})\\
Identify candidate periods, deploying integral convolution\\
Initialize and optimize variance for candidate periods\\
Find optimal combination of periodicities, which minimize the loss defined in (\ref{eq:loss_master})\\
Update loss and variance wrt. optimal periodicities 
\caption{Main Steps of GMPDA}
\label{algo:GMPDA_short}
\end{algorithm}
%
%
In the \textit{first step}, GMPDA computes $D(\mu)$ wrt. equation (\ref{eq:algo_dmu}) and subtracts the approximated contribution from the interaction events. The approximation of the length of interaction intervals is either limited by the minimal expected periodicity or by a user-defined parameter, denoted as $noise\_range$. The estimation of the approximation is outlined in Appendix~\ref{subsub:approx_zeta}. Further, for estimation of $D(\mu)$ and the loss, the range for $\mu$ is limited by the parameter $loss\_length$, mainly due to flattening of the Gaussian distribution with increasing variance for the Random Walk Model. A detailed discussion can be found in Appendix~\ref{subsub:NacdidatePer}. %

In the \textit{second step}, GMPDA estimates a set of candidate periodicities using a heuristic approach, since computing $G_{M}(\hat{\mu}_p, \hat{\sigma})$ for all possible $\hat{\mu}_p$ is computationally expensive. The heuristic approach searches iteratively for periodicities $\hat{\mu}$ by performing in each iteration "integral convolutions" on $D(\mu)$. The convolution smoothes the function for extracting periods that explain the time series. The maximal number of candidates is controlled by parameter $max\_candidates$, and the maximal number of iterations by the parameter $max\_iterations$. The heuristic approach is described in detail in Appendix~\ref{subsub:NacdidatePer}. %

In the \textit{third step}, GMPDA performs least-squares curve-fitting to improve the initial guess for $\hat{\sigma}$, please see Section~\ref{sub:LS_sigma} for more details. Please note, this is optional and can be controlled by the parameter $curve\_fit$. The curve fitting procedure is described in Appendix~\ref{sub:LS_sigma}.%

In the \textit{final step}, GMPDA computes the function $G_M(\hat{\mu}_p, \hat{\sigma})$ for all combinations of candidate periodicities and corresponding variances and selects the set of "prime periodicities" $\hat{\mu}^*$ as those that minimizes the loss, defined and explained in Appendix~\ref{sec:loss_function}.%
%
%
\section{Performance Evaluation on Test Cases}
\label{sec:test_cases}
This section describes the evaluation of GMPDA's capacity to detect periodicities $\mu^*$ and variances $\sigma^*$ on synthetic time series generated according to Clock and Random Walk Models. The performance of GMPDA detection of periodicity $\mu^*$ was compared to those of other periodicity detection algorithms including FFT, Autocorrelation with FFT, Histogram with FFT, and Eperiodicity\footnote{The alternative algorithms were implemented in MATLAB. For all algorithms, the minimal, maximal considered period length was set to $10$ and $350$, respectively. The corresponding code is available on \url{https://github.com/nnaisense/gmpda}.}. Their specific algorithms are described below.

\textbf{GMPDA}: We used the baseline Algorithm~\ref{app:GMPDA_ALGO} with $\hat{\sigma}$ set equal to $\sigma^*$ (i.e., $\sigma^*=\log(\mu)$) and no non-linear curve fitting. Please note, in real applications sigma is unknown and if no non-linear curve fitting is deployed, we suggest to run the algorithm multiple times for a range of possible $\hat{\sigma}$, and the optimal $\hat{\sigma}$ can then be chosen with respect to the lowest loss.

\textbf{GMPDA $\sigma^*$ unknown}: The algorithm is initialized with $\hat{\sigma}=int(\log(10))$, which is the minimal possible value for $\mu^*$  in our test cases. For both GMPDA configurations, the algorithm searches for maximal $|\mu^*|+2$ periodicities.

\textbf{FFT}: This is a power spectral density estimates approach~\cite{welch1967use}. In the case of a single periodicity, the frequency with the highest spectral power is selected as the prime periodicity. In the case of multiple periodicities, the frequencies $|\mu^*|$  with the highest spectral power are selected as the true periodicities.

\textbf{Autocorrelation with FFT}: The Autocorrelation Function (ACF) estimates how similar a sequence is to its previous sequence for different lags and then uses the lag that maximizes ACF as the predicted period~\cite{box2015time}. Since all integer multiples of true periods will have the same function value, an FFT is applied to the ACF to select the frequencies with the highest spectral power as the true periodicities. In the case of multiple periodicities, the frequencies with the highest spectral power are selected as the true periodicities.

\textbf{Histogram with FFT}: An FFT applied to the histogram of all forward differences in the time series $D(\mu)$~\cite{382394}. In the case of multiple periodicities, the number of frequencies with the highest spectral power are selected as the true periodicities.

\textbf{Eperiodicity}: We implemented the method presented in \cite{li}, which computes a "discrepancy score" for each possible periodicity, i.e. the number of intervals between events which are equal to the candidate periodicity. To detect multiple periods, we select the top $|\mu^*|$ candidate periods from the discrepancy function.

There are several conceptual differences and similarities between GMPDA and the alternative algorithms: GMPDA is, like all the methods listed above, based on computing frequencies/periodicities based on the observed intervals between positive observations in the time series. For data that follows a clock model, variance in intervals can be handled in the regression framework for ACF and with spectral methods for FFT. However, for very small/large variation in intervals, parametrized by $\sigma^*$ in the Random Walk Model, these methods may struggle - particularly for multiple periodicities - because of the linear variance increase. Eperiodicity/Histogram methods will likely have performance decrease for variable intervals since they have no specific capacity to deal with variance in intervals. This will be particularly true for time series following the Random Walk Model.%

GMPDA is designed for multiple-periodicity detection, and its loss function is explicitly oriented toward finding all the periodicities present in the data: Once the set of candidate periodicities is identified, GMPDA checks all possible combinations of periodicities and selects the one with the smallest loss.%

In this context, ACF and FFT both accepted methods for hierarchical frequency detection, but do not present a "stopping criteria" for deciding how many periodicities are significant in the considered time series. it is, therefore,  possible to over or underestimate their number. In our test cases, we always selected the top $|\mu*|$ frequencies as the true periodicities and likely overestimated the accuracy of these methods since there would be no way to know this number without a priori knowledge of the generative mechanism. Also, Eperiodicity/Histogram is not explicitly modeled as a method for multiple periodicity detection and has no capacity for dealing with noise in intervals.%

In addition, we also want to discuss, conceptually, the main difference between GMPDA and the classical Gaussian Mixture/Hidden Markov approaches. All three methods aim to fit the shape of a distribution, which is empirically described by the corresponding histogram of the data. However, the main difference of GMPDA is the generative models that account for the peaks in the histogram at prime periods and their integer multiples. In this context, GMPDA combines these peaks to get a better estimate. Classical MM/HMM models do not deploy this information. Instead, if the number of mixture models $K$ is large, it will try to fit all peaks individually or average them if case $K$ is small and thus get biased results. %

In the following, we compare the performance of the above-described algorithms for a large set of generated test cases.%
%
%
\subsection{Test cases}
The performance of the GMPDA algorithm was evaluated on a wide range of test cases for the Clock Model and the Random Walk Model. The test cases cover systematic variations of the following model parameters: periodicity $\mu$, variance $\sigma$, noise $\beta$, and the number of events $n$. These generative model parameters influence the histogram of inter-event intervals, the input data for all applied algorithms.

To better understand how these model parameters influence the histogram, we show two illustrative test cases with different model parameters. Figure~\ref{fig:testcase_sf0_fb0_pc2_n30} shows a well posed test case, with two underlying periodicities and no noise, while Figure~\ref{fig:testcase_sf0_fb2_pc3_n100} displays an ill-posed test case where the signal to noise ratio is 1:2.
\begin{figure}[h]
    \centering
     \includegraphics[width=0.45\textwidth]{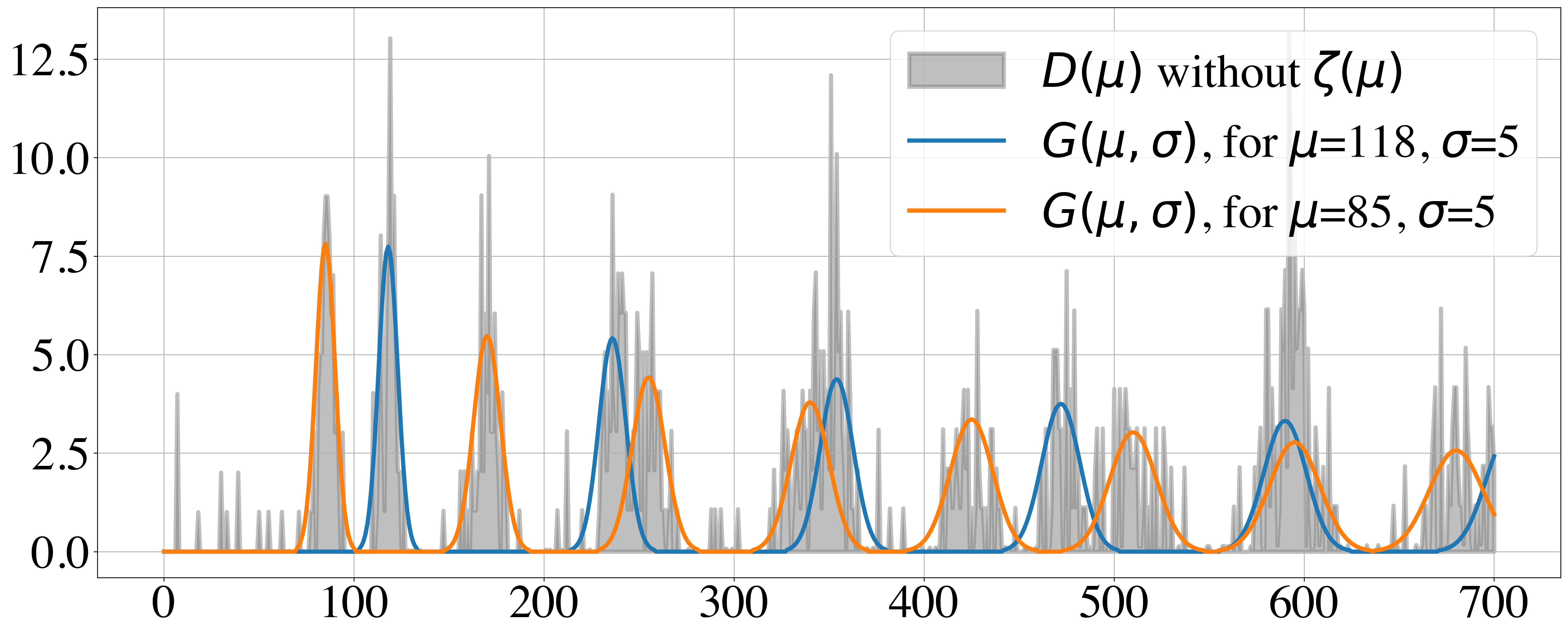}
    \caption{Histogram of the intervals $D(\mu)-\hat z$ and generative curves $G(\mu,\sigma)$ for the Random Walk Model with $n=100$ and $\beta=0$.}
    \label{fig:testcase_sf0_fb0_pc2_n30}
\end{figure}
\begin{figure}[h]
    \centering
     \includegraphics[width=0.45\textwidth]{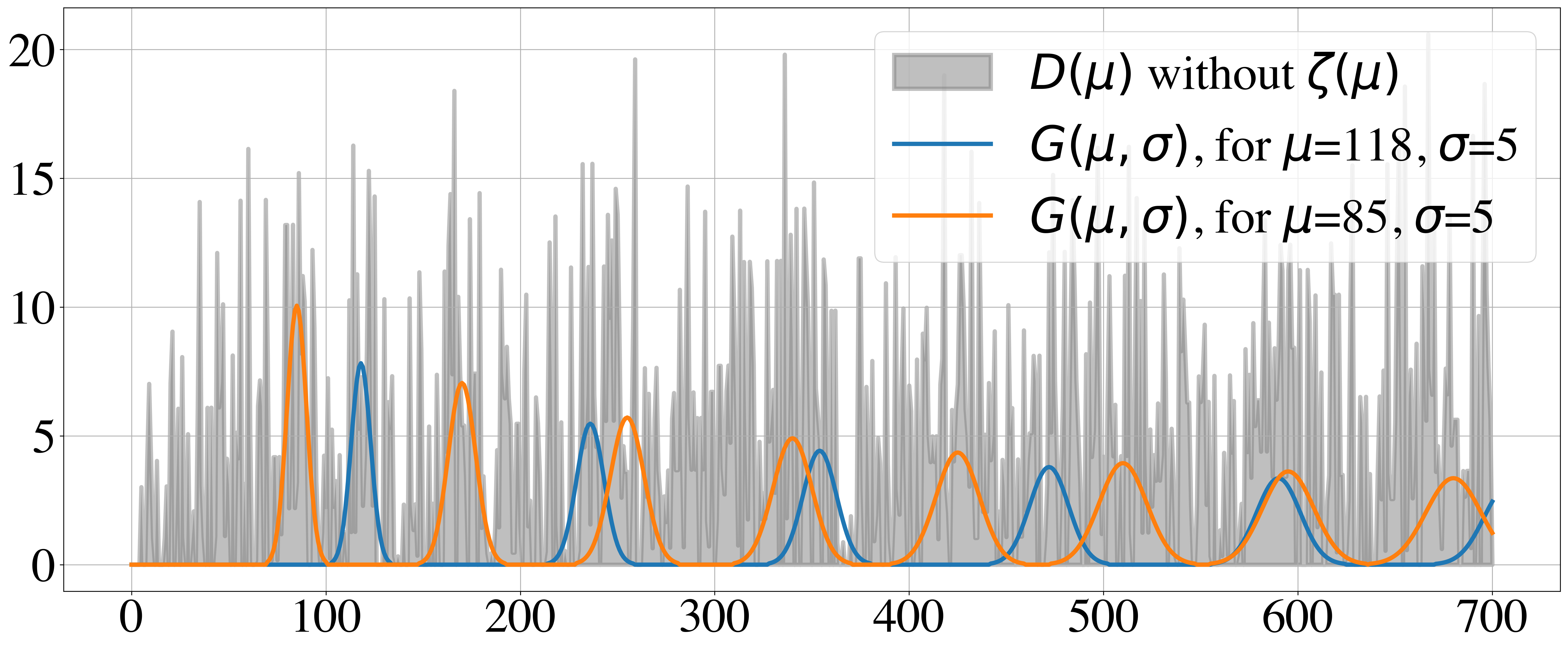}
    \caption{Histogram of the intervals $D(\mu)-\hat z$ and generative curves $G(\mu,\sigma)$ for the Random Walk Model with $n=100$ and $\beta=2$.}
    \label{fig:testcase_sf0_fb2_pc3_n100}
\end{figure}
%
The identification of underlying periodicities requires advanced analysis of the histogram. 

The following analyses examined an extensive range of test cases in order to study the limitations of the presented GMPDA algorithm and the alternatives described in Section~\ref{sec:test_cases}.
%
%
%
\subsubsection{Configurations}
The considered range for the model parameters was set as follows: 
\begin{itemize}
	\item $\sigma^*\in\lbrace 1, \log(\mu),\frac{\mu}{p}\rbrace$, with $p=16, 8, 4, 3$, 
	\item $n\in\lbrace 10, 30, 50, 100, 300, 500\rbrace$,
	\item $\beta\in\lbrace 0, 0.1, 0.5, 0.7, 1, 2, 4, 8\rbrace$.
\end{itemize}
Please note, test cases with $\sigma=\log(\mu)$ represent scenarios where no $\sigma$ optimization is required as $\sigma=\log(\mu)$ is default initialization in GMPDA. Small values for $n$ and large values for $\sigma,\beta$ were chosen to investigate the limits the periodicity detection algorithms.

For every combination of $\sigma^*$, $\beta$, and $n$ we generated $100$ event time series with randomly drawn $\mu^*\in[10,350]$. For test cases with multiple periodicities we enforced the difference between the involved periodicities to be bigger then $\log(\mu)$. Otherwise, the generative curves become indistinguishable too fast and multiple periodicity detection is getting too ill-posed.

The combination of the above settings resulted in 28800 test cases for each generative model. All algorithms were applied to identify the underlying periodicities for every generated test case.

An identified periodicity was considered to be correct if it was within $\mu^* \pm 0.5 \cdot\sigma$, where $\mu^*$ is the true periodicity and $\sigma$ the corresponding variance. For instance, for the cases $\mu^*=15,\,\sigma=2$ and $\mu^*=350,\,\sigma=44$ a guess of $\mu$ within $15\pm 1$ and $350\pm 22$, respectively, would be considered an accurate detection.

Thus for a fixed configuration of the parameters $\sigma$, $\beta$, and $n$, the performance of the algorithms is measured by accuracy, which is the averaged (across the 100 generated test cases) number of correctly identified periodicities with a value between zero and one.

In the following, we first present the results for $|\mu^*|=1$ and identify valid ranges for $n$, $\beta$ and $\sigma$. Second, within the valid range we compare the performance of GMPDA to that of the other algorithms for $|\mu^*|=1,2,3$.
%
\subsection{Performance w.r.t. $|\mu^*|=1$}
%
\subsubsection{GMPDA Performance}
In this section we focus on the performance of GMPDA with respect to $|\mu^*|=1$ in order to select the realistic limits for $\sigma,\,\beta$ and the number of events.

Figures~\ref{fig:perf_mu1_sigma_ncf} and ~\ref{fig:perf_mu1_sigma_cf} display the performance of GMPDA for fixed $\beta=0$ and $|\mu|=1$ for different values of $\sigma$ and for different number of events $n$, without and with curve-fitting, respectively.
The confidence intervals (CI) in all the following figures (if present) are estimated as $\bar x \pm 1.96$ SEM, where $\bar x$ is the mean and SEM is the Standard Error of the Mean.
\begin{figure}[h]
    \centering
    \subfloat[Random Walk Model]{{\includegraphics[width=0.23\textwidth]{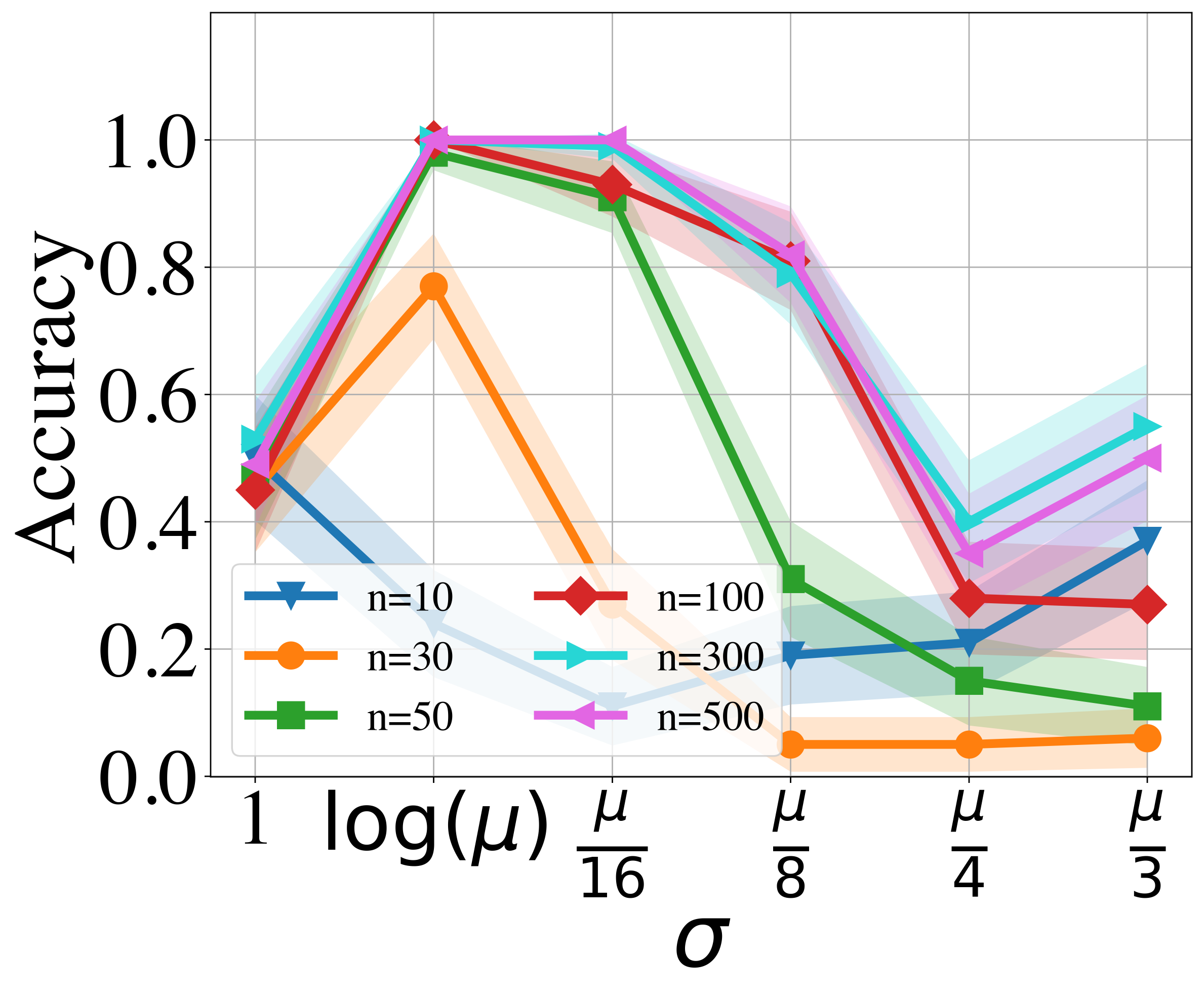}}}
    \subfloat[Clock Model]{{\includegraphics[width=0.23\textwidth]{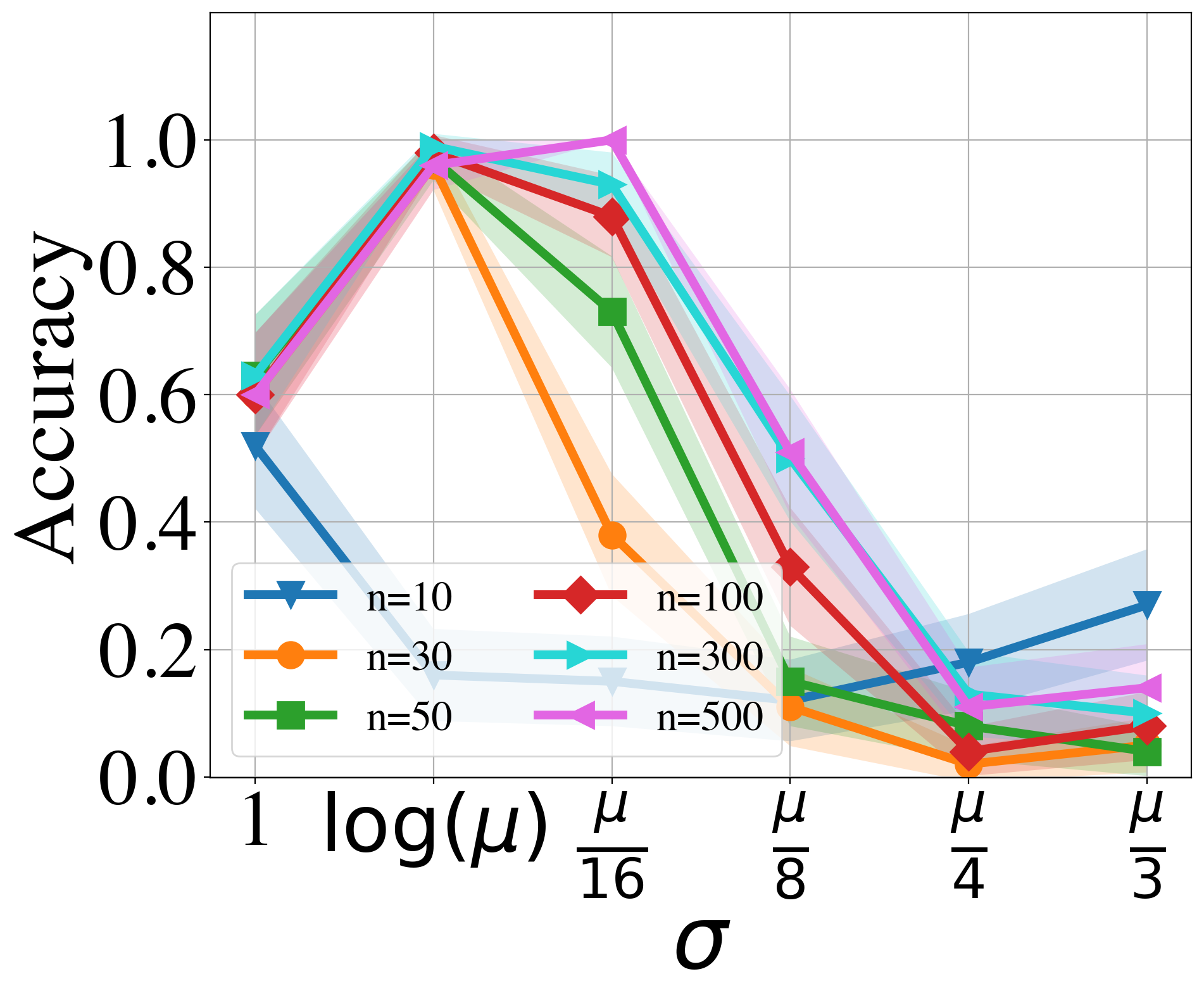}}}
    \caption{Performance of GMPDA without curve-fitting for the Random Walk Model (panel a) and for the Clock model (panel b), with $\beta=0$ and $|\mu|=1$ and varying number of events (n).}
    \label{fig:perf_mu1_sigma_ncf}
\end{figure}
\begin{figure}[h]
    \centering
    \subfloat[Random Walk Model]{{\includegraphics[width=0.23\textwidth]{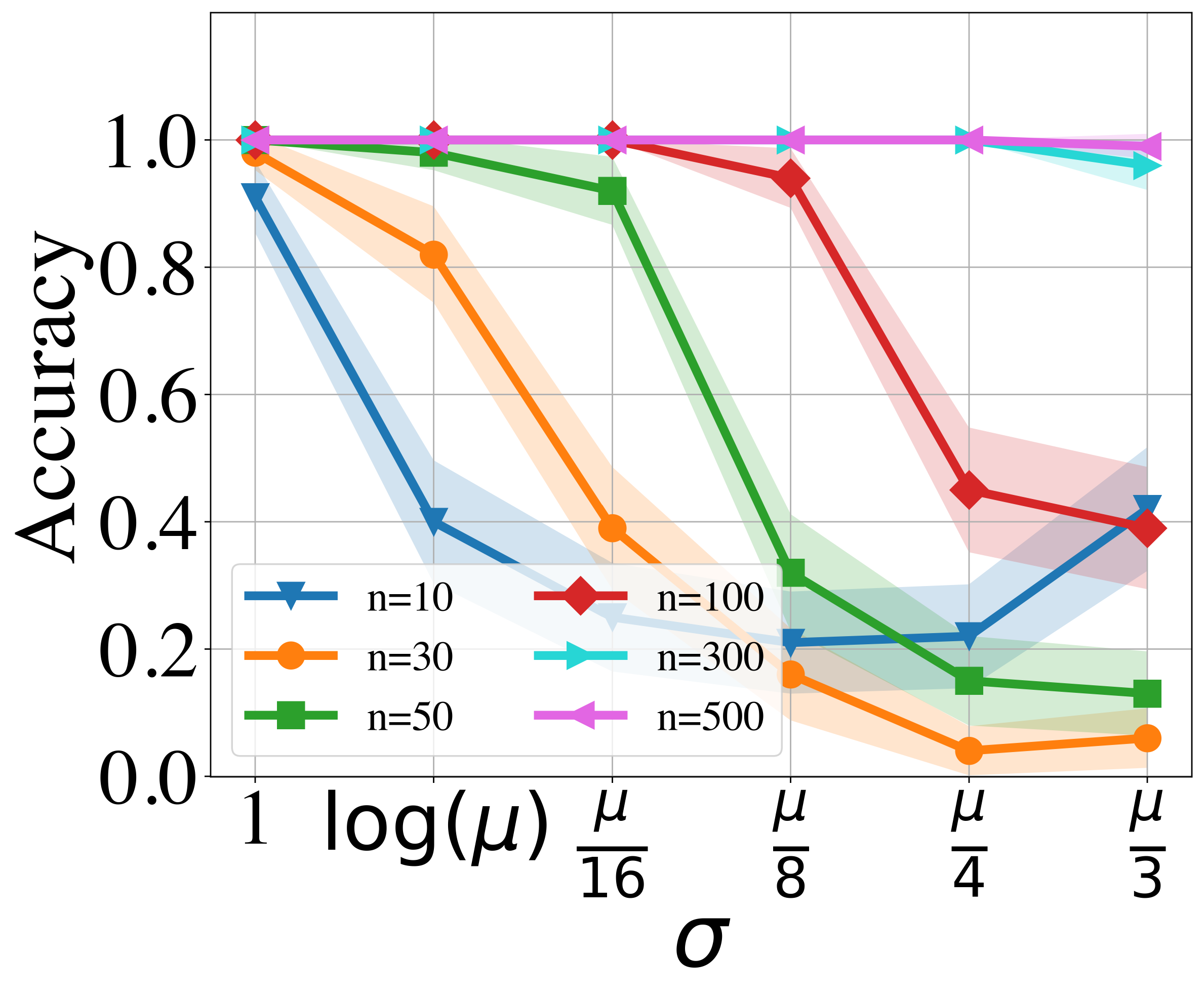}}}
    \subfloat[Clock Model]{{\includegraphics[width=0.23\textwidth]{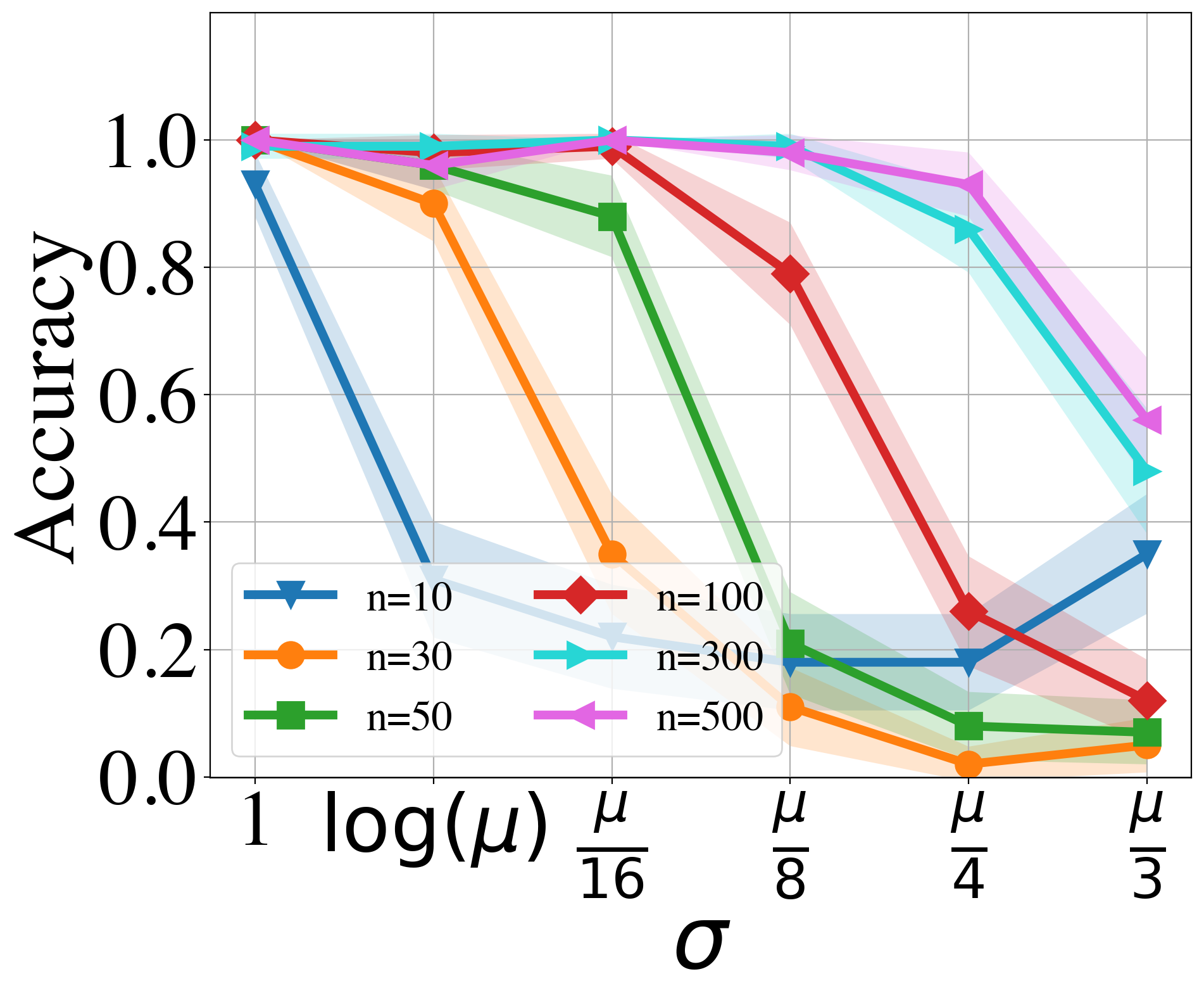}}}
    \caption{Performance of GMPDA with curve-fitting for the Random Walk Model (panel a) and for the Clock model (panel b), with $\beta=0$ and $|\mu|=1$ and varying number of events (n).}
    \label{fig:perf_mu1_sigma_cf}
\end{figure}
The results in Figures~\ref{fig:perf_mu1_sigma_ncf} and \ref{fig:perf_mu1_sigma_cf} show, as expected, that accuracy decreased with increasing $\sigma$ and decreasing number of events. Or, stated differently, with increasing variance more events were required for an accurate detection. The figures can also be used to compare the performance of GMPDA with and without curve-fitting. GMPDA without curve-fitting performed worse except in the case of $\sigma=log(\mu)$. The explanation for this behavior is as follows: In the algorithm, the default initialization value of $\sigma$ is $log(\mu)$, and therefore for this configuration, GMPDA without curve-fitting worked with a known sigma. In all the other cases, GMPDA with curve-fitting provided better results. %

Next, to compare the effect of noise, we restricted our evaluation from here on to GMPDA with curve-fitting due its better performance. Please note, the comparison between results with curve fitting and without curve fitting can be found in~\ref{app:|mu|=1}. %
Further, we focus on the case of $|\mu^*|=1$ and $\sigma$ known, which can be viewed as an \textit{ideal} scenario, as only $\mu$ needs to be estimated. For this ideal case, we compared the effect of varying noise levels across a varying number of events on detection accuracy.%
\begin{figure}[h]
    \centering
    \subfloat[Random Walk Model]{{\includegraphics[width=0.23\textwidth]{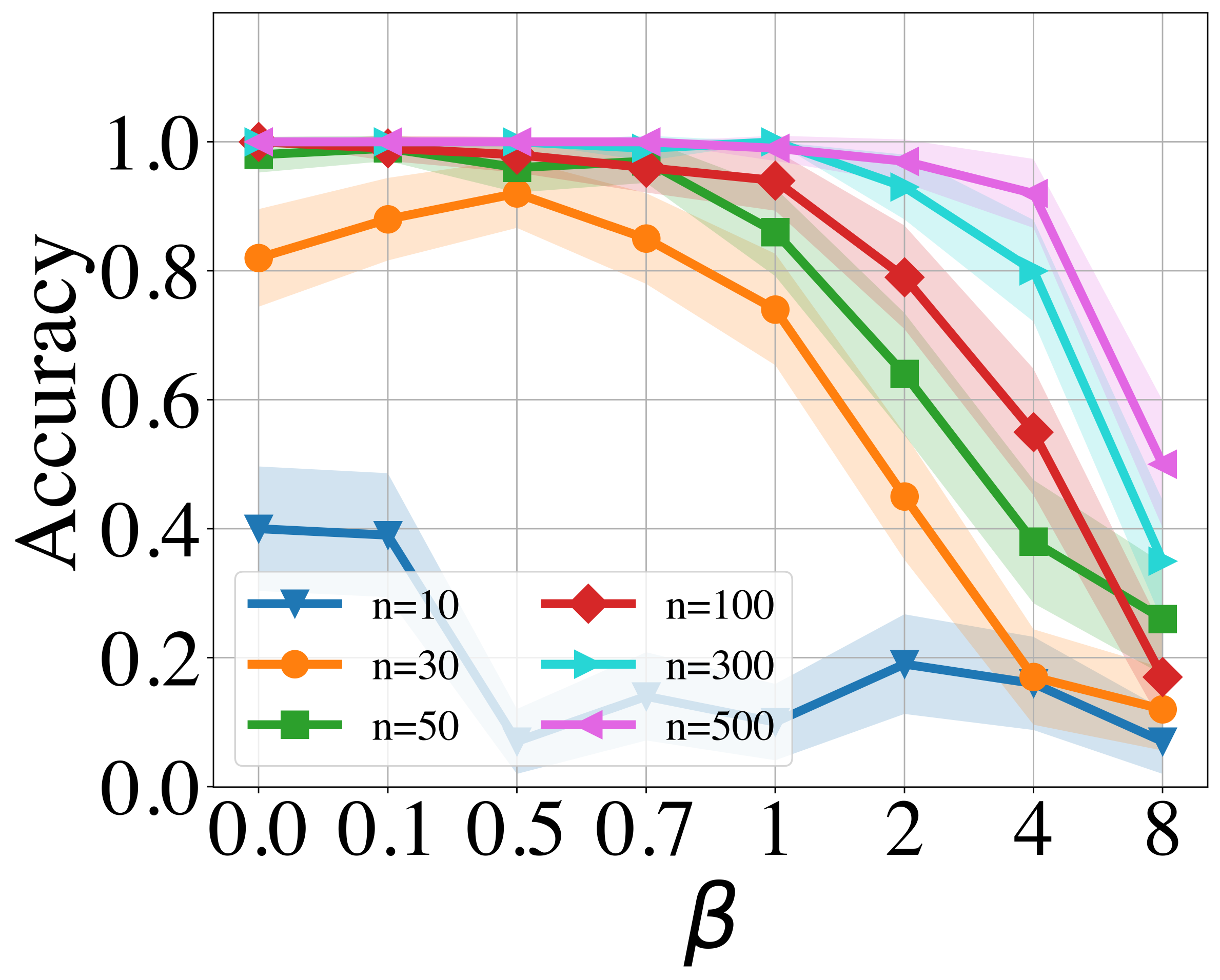}}}
    \subfloat[Clock Model]{{\includegraphics[width=0.23\textwidth]{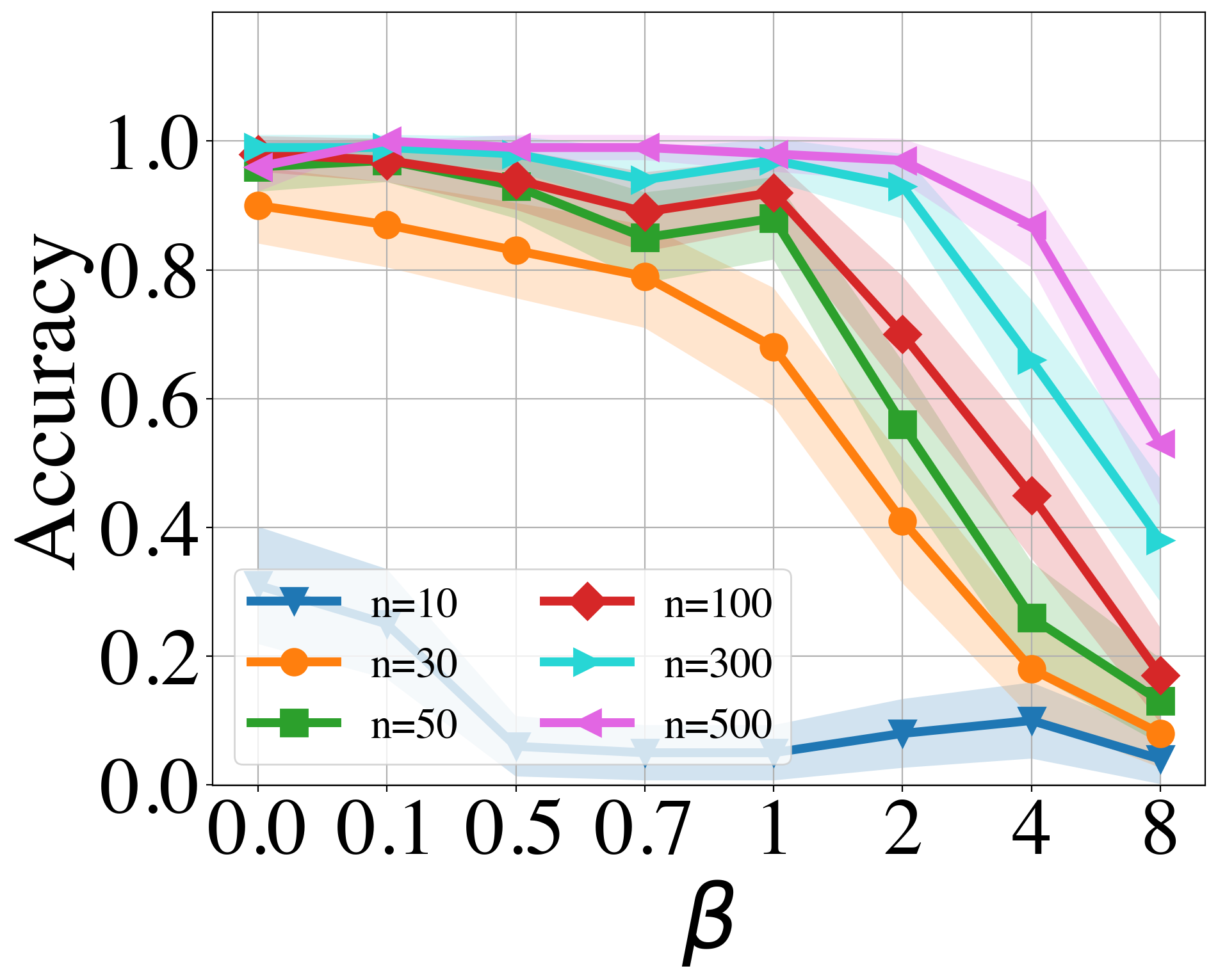}}}
    \caption{Performance of GMPDA with curve-fitting for the Random Walk Model (panel a), and for the Clock model (panel b), with $|\mu|=1$ and $\sigma=log(\mu)$ across varying levels of uniform noise beta and number of events.%
	}
    \label{fig:perf_mu1_beta}
\end{figure}
Figure~\ref{fig:perf_mu1_beta} shows the performance of the different algorithms with respect to increasing amounts of noise in the time series, for the case with $|\mu|=1$ and $\sigma=log(\mu)$, and separately for Random Walk and the Clock Models.

For the Random Walk Model, Figure~\ref{fig:perf_mu1_beta} panel (a), performance was acceptable for signals with $n \geq 300$ and noise up to $\beta = 4$, for $n\leq 300$ performance dropped below $0.75$ already for $\beta \geq 2$. In comparison, the Clock Model was substantially more sensitive to noise (Figure~\ref{fig:perf_mu1_beta}, panel b) with acceptable results only for $\beta \leq 1$.

In summary, in cases where the actual variance is unknown, GMPDA with curve-fitting outperformed GMPDA without curve-fitting. GMPDA was not suited for cases with less than 50 events. GMPDA's performance increased with the increasing number of events. GMPDA could also handle moderate to high amounts of noise, and we show in the next section how this compares to other periodicity detection algorithms.
%
\subsubsection{Comparison with alternative periodicity detection algorithms.}
\label{subsub:com_alt_method_mu1}
%
Next, we compared the GMPDA (with curve-fitting) algorithm to other periodicity detection algorithms regarding their performance in conditions with varying noise and variance. With increases in noise and variance, the histograms of the inter-event intervals analyzed by all algorithms become less informative, i.e., the peaks that indicate periodicities become less identifiable. Therefore, we investigated the sensitivity to noise and different variances used for generating the periodicities.
We first investigated the effect of varying levels of variance $\sigma$ for cases where no noise was present, i.e., $\beta=0$. The results for all algorithms and $n=100$ are shown in Figure~\ref{fig:perf_mu1_sigma}. The results for different numbers of $n$ averaged overall levels of $\beta$ can be found in Appendix~\ref{app:variance_sigma}.
For the Random Walk Model, GMPDA was very accurate up to $\sigma=\frac{\mu}{8}$. Interestingly, all other algorithms performed worse when variance was very small ($\sigma = 1$ and $\sigma = \log(\mu)$), a case where GMPDA excelled. FFT and AutoCor converged to the accuracy bound given by GMPDA for $\sigma>1$. While the accuracy of EPeriodicity and Hist had its maximum at $~0.8$. For all methods, the performance dropped for $\sigma\geq\frac{\mu}{8}$. This behavior is distinctive for the random walk model, where the variance increases with every step. Therefore the generative distributions will start to overlap, which happens faster when the variance is larger.
Performance was generally lower for the Clock Model, which was also more sensitive to increases in the variance. GMPDA was sufficiently accurate only for $\sigma=1$ and $\sigma=\log(\mu)$, with a distinct drop in performance with increased variance. For the other algorithms, except the histogram methods, performance initially increased with increasing variance up to $\sigma=\frac{\mu}{8}$ and strongly declined afterward.
\begin{figure}[h]
    \centering
    \subfloat[Random Walk Model]{{\includegraphics[width=0.23\textwidth]{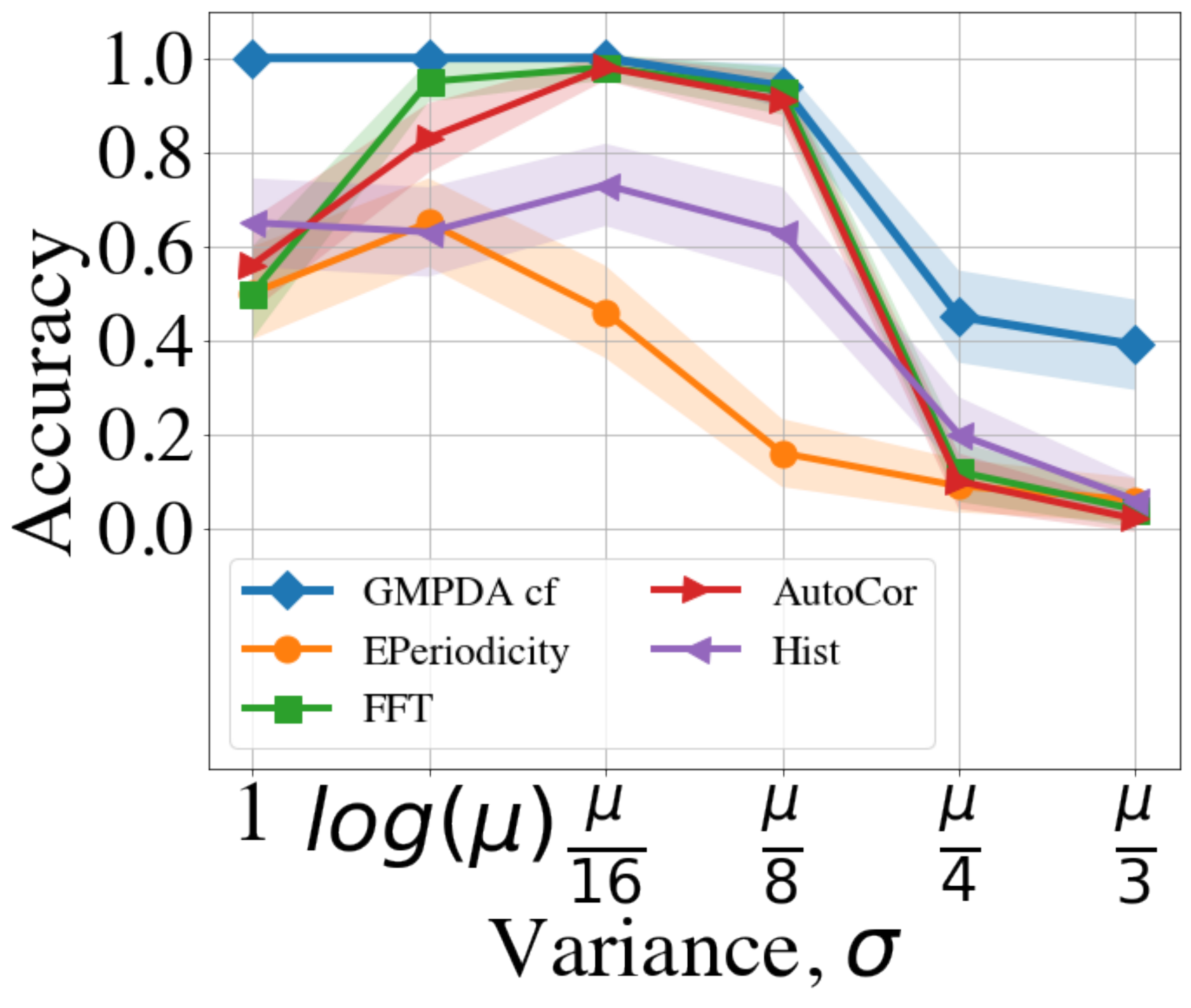}}}
    \subfloat[Clock Model]{{\includegraphics[width=0.23\textwidth]{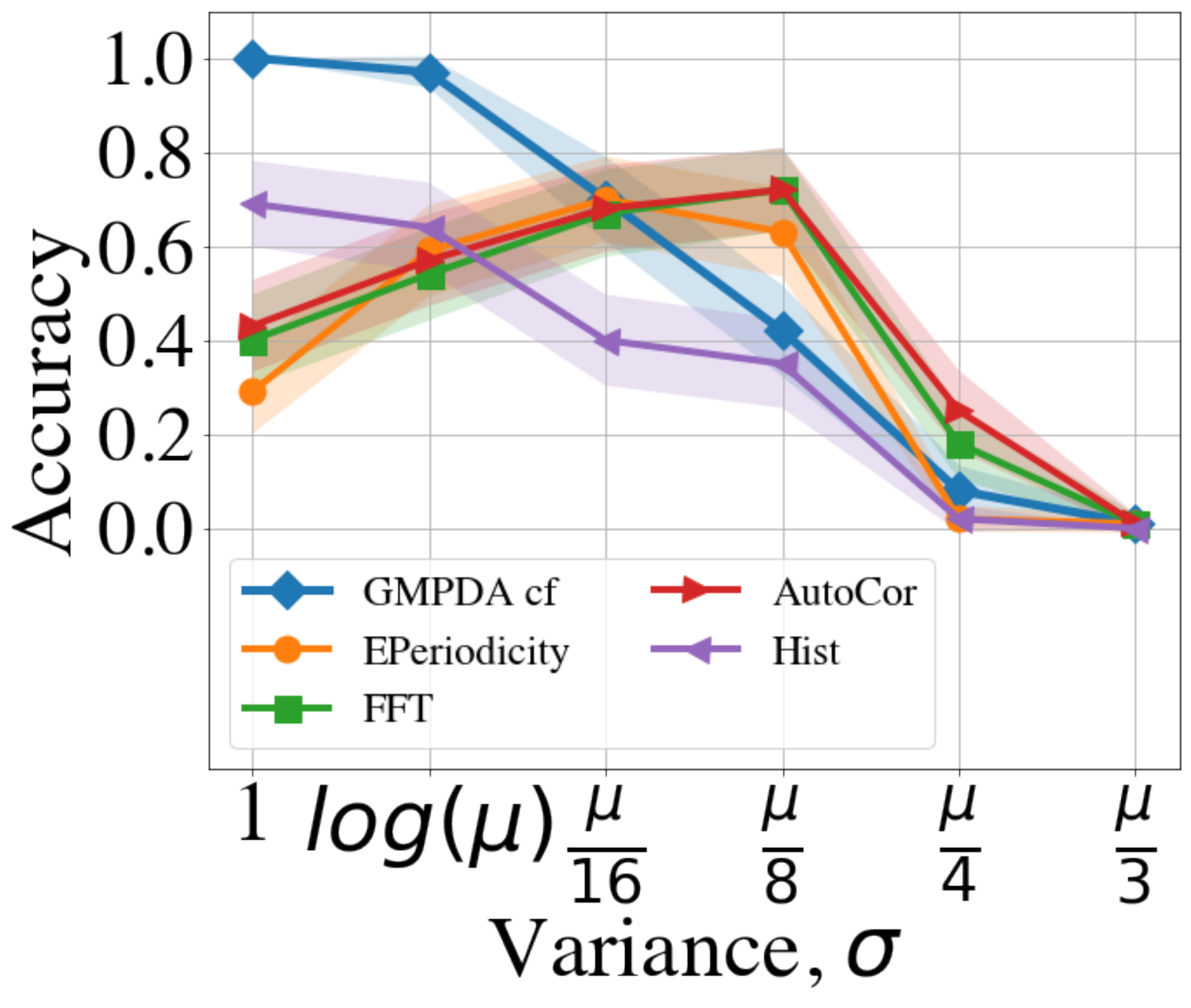}}}
    \caption{Comparison of GMPDA to alternative algorithms for the Random Walk Model (panel a), and for the Clock model (panel b). Accuracy is plotted for different levels of variance $\sigma$ for cases with one period ($|\mu|=1$), no noise ($\beta=0$) and number of events, $n=100$.}
    \label{fig:perf_mu1_sigma}
\end{figure}

Next we evaluated the performance of all methods with respect to increasing levels noise and with the results shown in Figure~\ref{fig:perf_mu1_beta_all}. For these analyses, the variance was fixed to $\sigma=log(\mu)$ and number of events to $n=100$. The plots for all numbers of $n$ can be found in Appendix~\ref{app:noise_beta}.
For the Random Walk Model, GMPDA was insensitive to noise up to $\beta=1$ with performance decreasing linearly after that. Performance of FFT and AutoCor mirrored that of GMPDA with slightly lower levels of accuracy. Of note was EPeriodicity's performance, which increased up to $\beta=1$ and declined after that, while Hist was very sensitive to all levels of noise and performed worse than all other algorithms.%

%
For the Clock Model GMPDA behaved similar, while the performance of the other methods was more sensitive to noise and accuracy was generally lower than for the Random Walk Model.

\begin{figure}[h]
    \centering
    \subfloat[Random Walk Model]{{\includegraphics[width=0.23\textwidth]{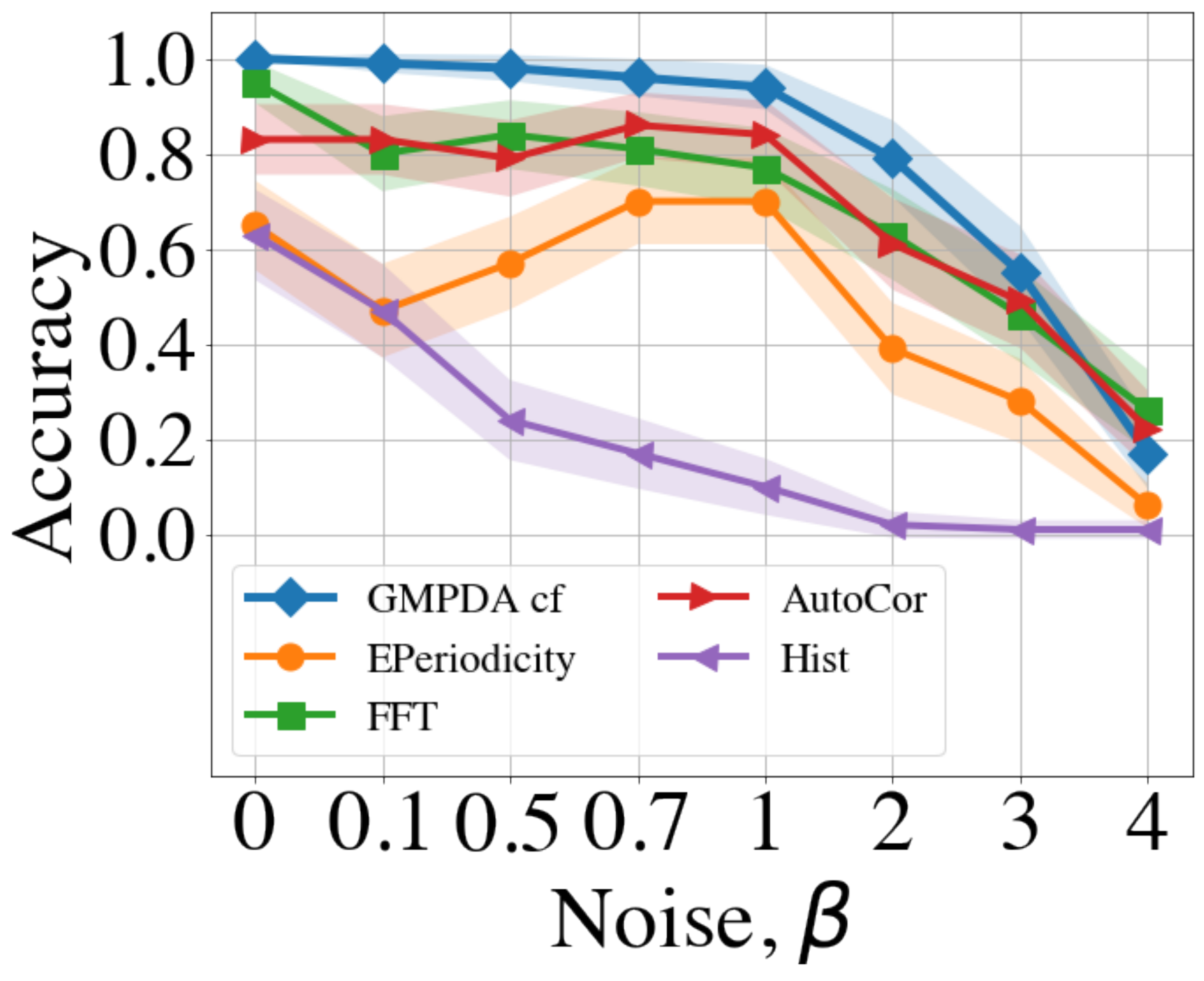}}}
    \subfloat[Clock Model]{{\includegraphics[width=0.23\textwidth]{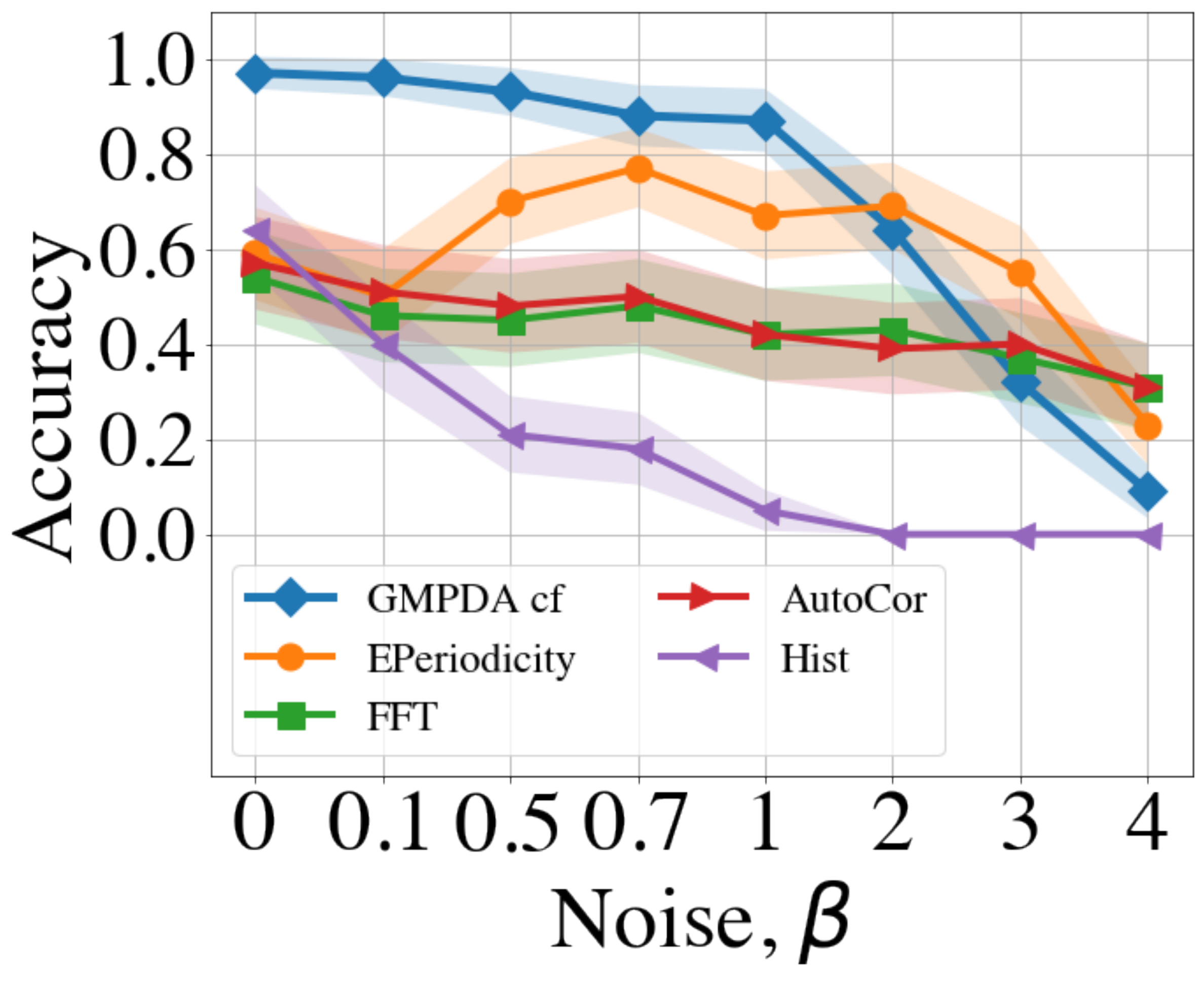}}}
    \caption{Comparison of GMPDA to alternative methods for the Random Walk Model (panel a), and for the Clock model (panel b). Accuracy is plotted against increasing levels of noise $\beta$ for cases with one period ($|\mu|=1$), known variance, i.e., $\sigma=\log(\mu)$) and number of events, $n=100$.}
    \label{fig:perf_mu1_beta_all}
\end{figure}
The presence of moderate noise (i.e., with $\beta\in[0.1, 0.7]$) did not affect performance, except for EPeriodicity, where performance increased for noise levels up to $\beta=2$.

Further, the maximal noise levels that the algorithms could handle were not higher than two $\beta\leq 2$, i.e., a signal to noise ratio of 1:2, one periodic event to two noise events.


Concluding the comparison, we averaged performance over all acceptable values of noise and variance, i.e., $\sigma=\lbrace 1, log(\mu), \frac{\mu}{16}, \frac{\mu}{8} \rbrace$ and $\beta \leq 2$). The results are shown in Figure~\ref{fig:perf_mu1_all_ave}.
\begin{figure}[h]
    \centering
    \subfloat[Random Walk Model]{{\includegraphics[width=0.23\textwidth]{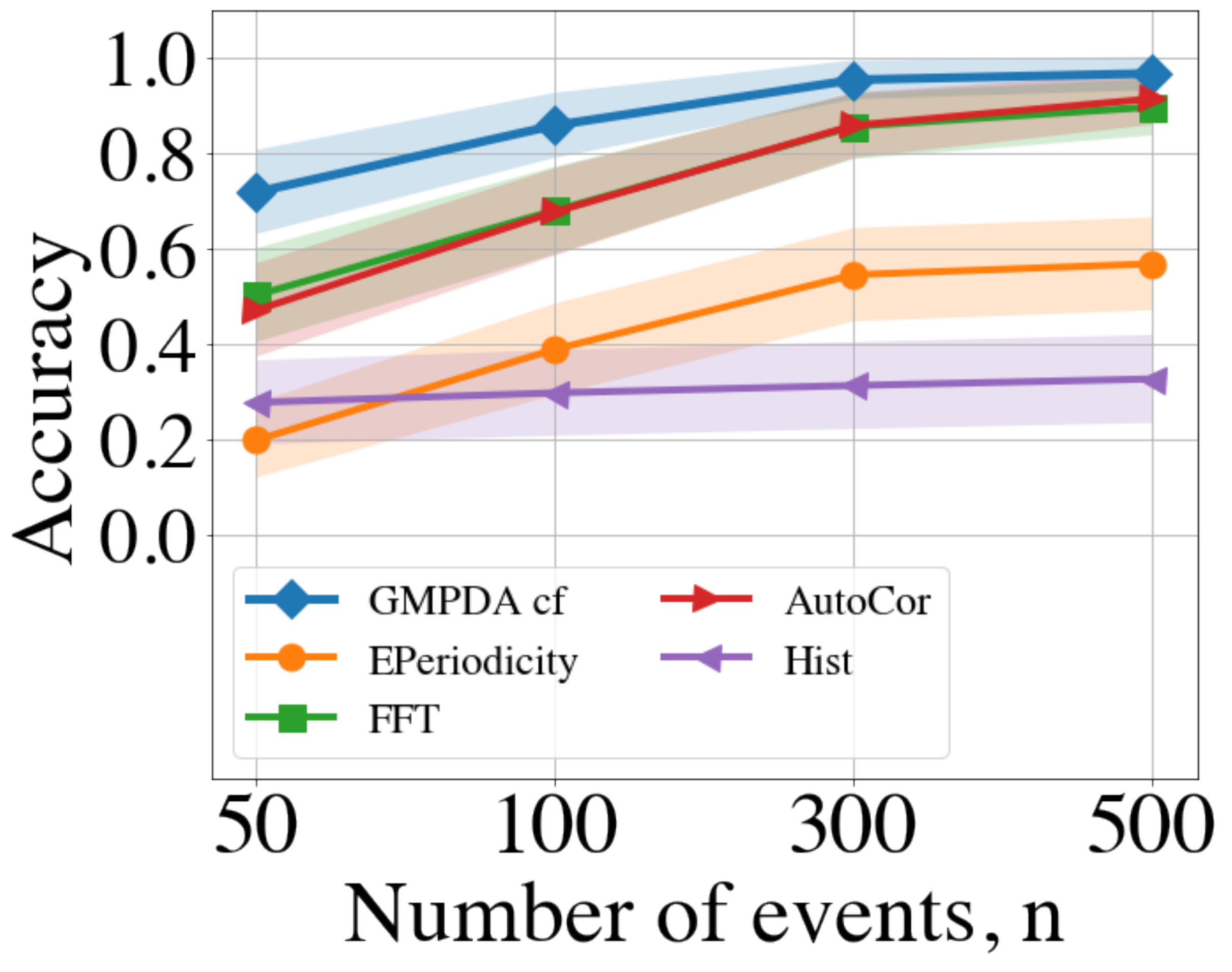}}}
    \subfloat[Clock Model]{{\includegraphics[width=0.23\textwidth]{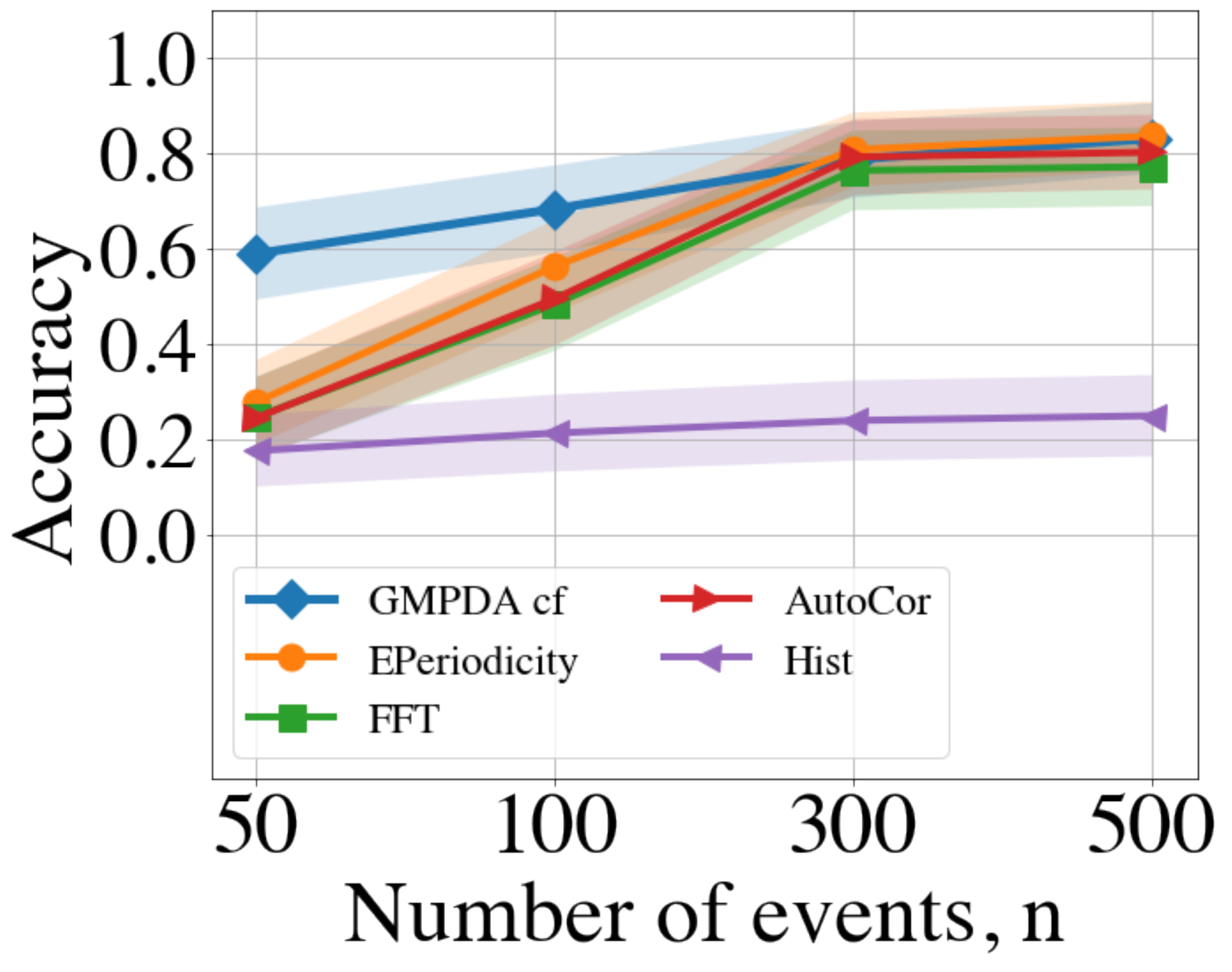}}}
    \caption{Comparison of GMPDA to alternative methods for the Random Walk Model (panel a), and for the Clock model (panel b). Accuracy is plotted against the number of events averaged over $\sigma=\lbrace 1, log(\mu), \frac{\mu}{16}, \frac{\mu}{8} \rbrace$ and $\beta \leq 2$.}
    \label{fig:perf_mu1_all_ave}
\end{figure}
Overall, detection of a single periodicity was increasingly accurate with an increasing number of events for all methods and both the Random Walk and Clock Models (see Figure 8). For both models, the periodicity detection with the Hist algorithm had very low accuracy with the maximal performance of less than $0.4$.

For the Random Walk Model, GMPDA outperformed alternative approaches with the accuracy converging to one as the number of events increased, and even for $n=30$, its performance was larger than $0.75$. FFT/Autocor achieved similar performance when the number of events was larger than $300$. In contrast, EPeriodicty's performance for the Random Walk Model was relatively poor, with a maximum at $0.6$ for $500$ number of events.

For the Clock Model, GMPDA outperformed alternatives when the number of events was smaller than $300$. For the number of events larger than $300$, the performance of all the approaches, except Hist, became equally good.
%
\subsection{Performance w.r.t. $|\mu^*|>1$}
This section compares the performance of GMPDA (with and without curve-fitting) to that of the alternative methods for multiple periodicity detection, focusing on the set of sensible simulation parameters identified in Section~\ref{subsub:com_alt_method_mu1}, i.e., $n={50, 100, 300, 500}$, $\sigma=\lbrace 1, log(\mu), \frac{\mu}{16}, \frac{\mu}{8} \rbrace$, and $\beta\leq 1$, resulting in 8000 test cases for each setting of $|\mu|=2$ and $|\mu|=3$ and for every generative model. For comparison, the performance was summarized over $n,\, \mu,\, \sigma, $ and $\beta$ and is visualized as a histogram, where the $x$-axis displays the number of correctly detected periodicities and the $y$-axis the number of test cases.

Figures~\ref{fig:mu2_hist_rw_acc} and \ref{fig:mu3_hist_rw_acc} show the results for the Random Walk Model for $|\mu|=2$ and $|\mu|=3$, respectively. Figures~\ref{fig:mu2_hist_clock_acc} and \ref{fig:mu3_hist_clock_acc} show the results for for the Clock Model for $|\mu|=2$ and $|\mu|=3$, respectively.
\begin{figure}[h]
    \centering
    \includegraphics[width=0.4\textwidth]{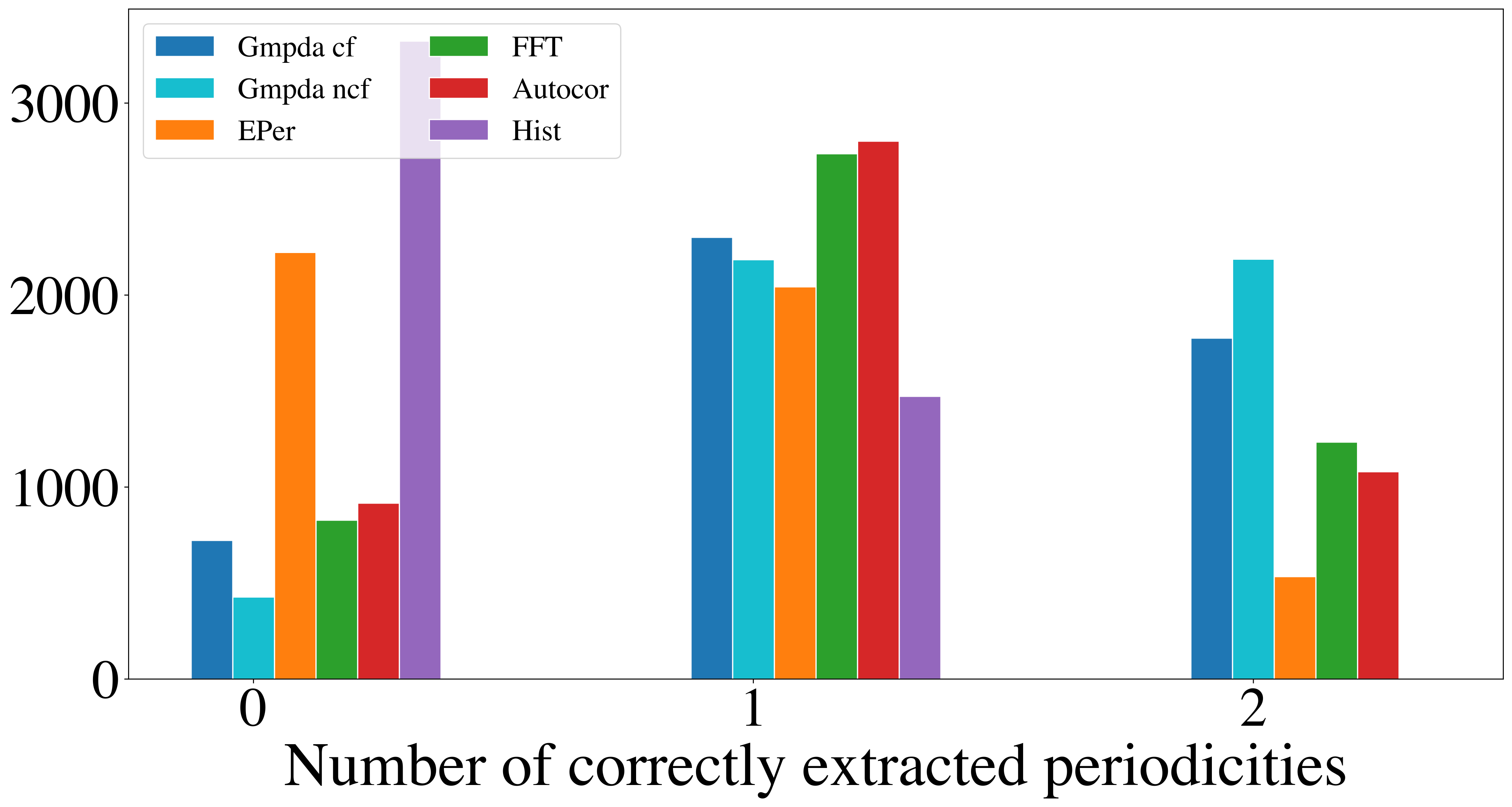}
    \caption{Summarized performance of GMPDA and alternative methods for the Random Walk Model and $|\mu|=2$.}
    \label{fig:mu2_hist_rw_acc}
\end{figure}
\begin{figure}[h]
    \centering
    \includegraphics[width=0.4\textwidth]{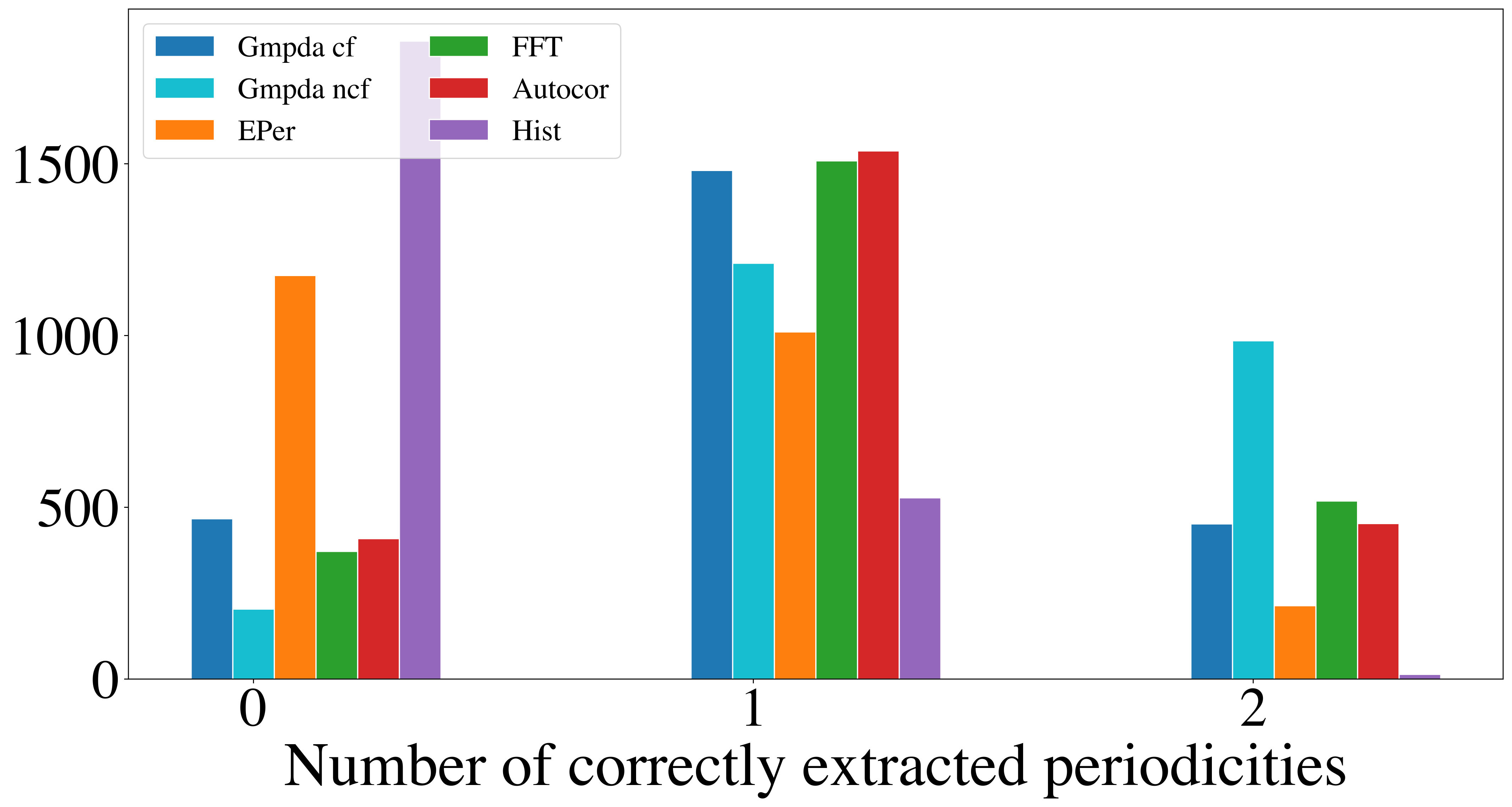}%
    \caption{Multiple periodicity detection for $|\mu|=2$; Summarized performance of GMPDA and alternative methods for the Clock Model.}
    \label{fig:mu2_hist_clock_acc}
\end{figure}
For the case with two periodicities, $|\mu|=2$, GMPDA outperformed the alternative methods, both with and without curve fitting. GMPDA without curve-fitting performed slightly better, suggesting that the currently deployed sigma optimization might require further development.
\begin{figure}[h]
    \centering
    \includegraphics[width=0.4\textwidth]{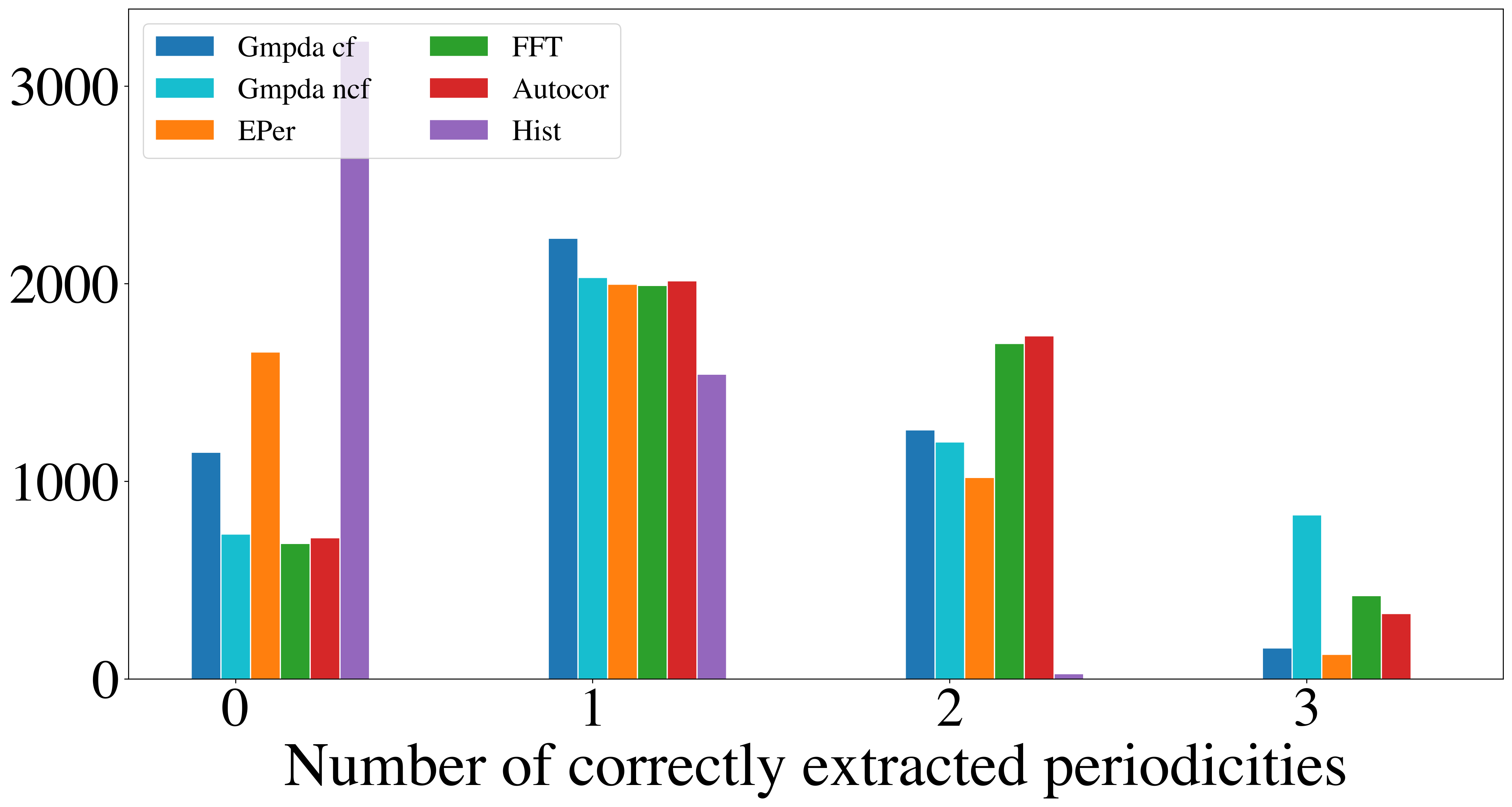}
    \caption{Multiple periodicity detection for $|\mu|=3$; Summarized performance of GMPDA and alternative methods for the Random Walk Model.}
    \label{fig:mu3_hist_rw_acc}
\end{figure}
\begin{figure}[h]
    \centering
    \includegraphics[width=0.4\textwidth]{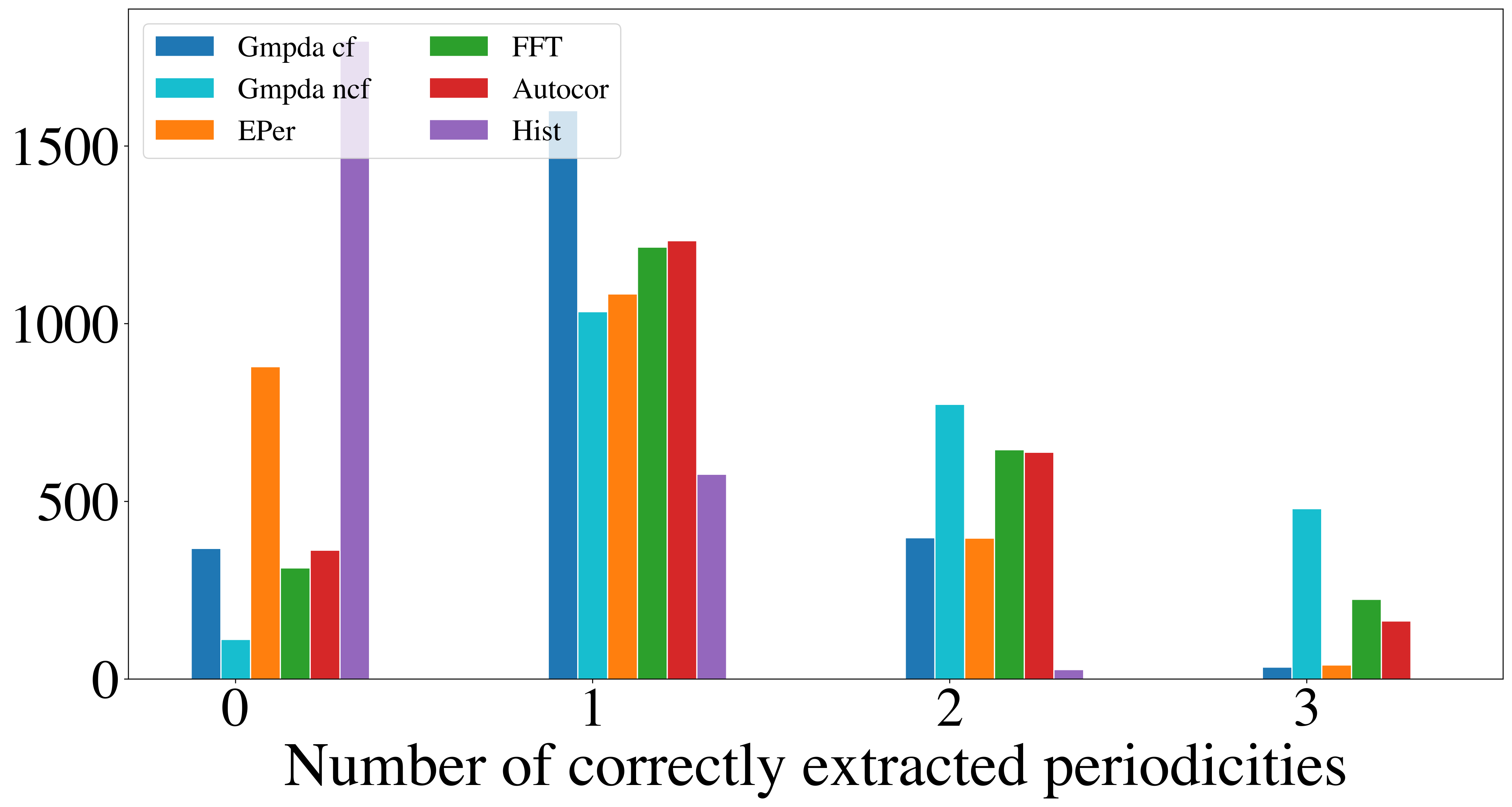}%
    \caption{Summarized performance of GMPDA and alternative methods for the Clock Model and $|\mu|=3$.}
    \label{fig:mu3_hist_clock_acc}
\end{figure}
The detection of three periodicities, $|\mu|=3$, was challenging for all methods, as shown in Figure 11. One possible explanation is that with more periodicities, there are more interaction intervals, i.e., intervals between the periodic events from different periodicities. Furthermore, at least for the Random Walk Model, the histogram is becoming less and less identifiable, as $\sigma$ grows for every subsequent step, which flat out the distribution responsible for the events and this effect is amplified when more than one periodicity is present. We conclude that GMPDA in the current version is not well suited for detecting more than two periodicities. %
%
\subsection{Computational Performance}
%
The computational performance (CPU) of the GMPDA algorithm was estimated for different experiments. For this purpose, we considered time series that were generated for every combination of the following model parameters: $|\mu^*| = [1,2,3]$, events per periodicity = $[50,100,300,500]$, $\sigma^* = [\log(\mu)]$, $\beta=[1]$. GMPDA was executed for each time series, the computational/execution time was determined via Python module \verb|timeit| with $100$ executions. For the generated test cases we tested with following GMPDA configurations (described in Appendix~\ref{app:GMPDA_ALGO}) $loss\_length=[400, 800, 1200]$ and $max\_periods=[|\mu|+2]$, while the remaining parameters were fixed at $L_{min}=5$, $max\_iterations=5$, $max\_candidates=15$, $noise\_range=5$, $loss\_tol\_change=0.01$.%

Our analysis showed that the computational performance had a strong dependence on maximal number of allowed periodicities, $max\_periods$. The CPU for both models (averaged over the number of executions, number of events $n$, and $loss\_length$) is shown in Figure~\ref{fig:cpurwm}. All other parameters had a comparatively minor influence on the performance (data not shown here). %
\begin{figure}[h]
    \centering
    \subfloat[Random Walk Model]{{\includegraphics[width=0.22\textwidth]{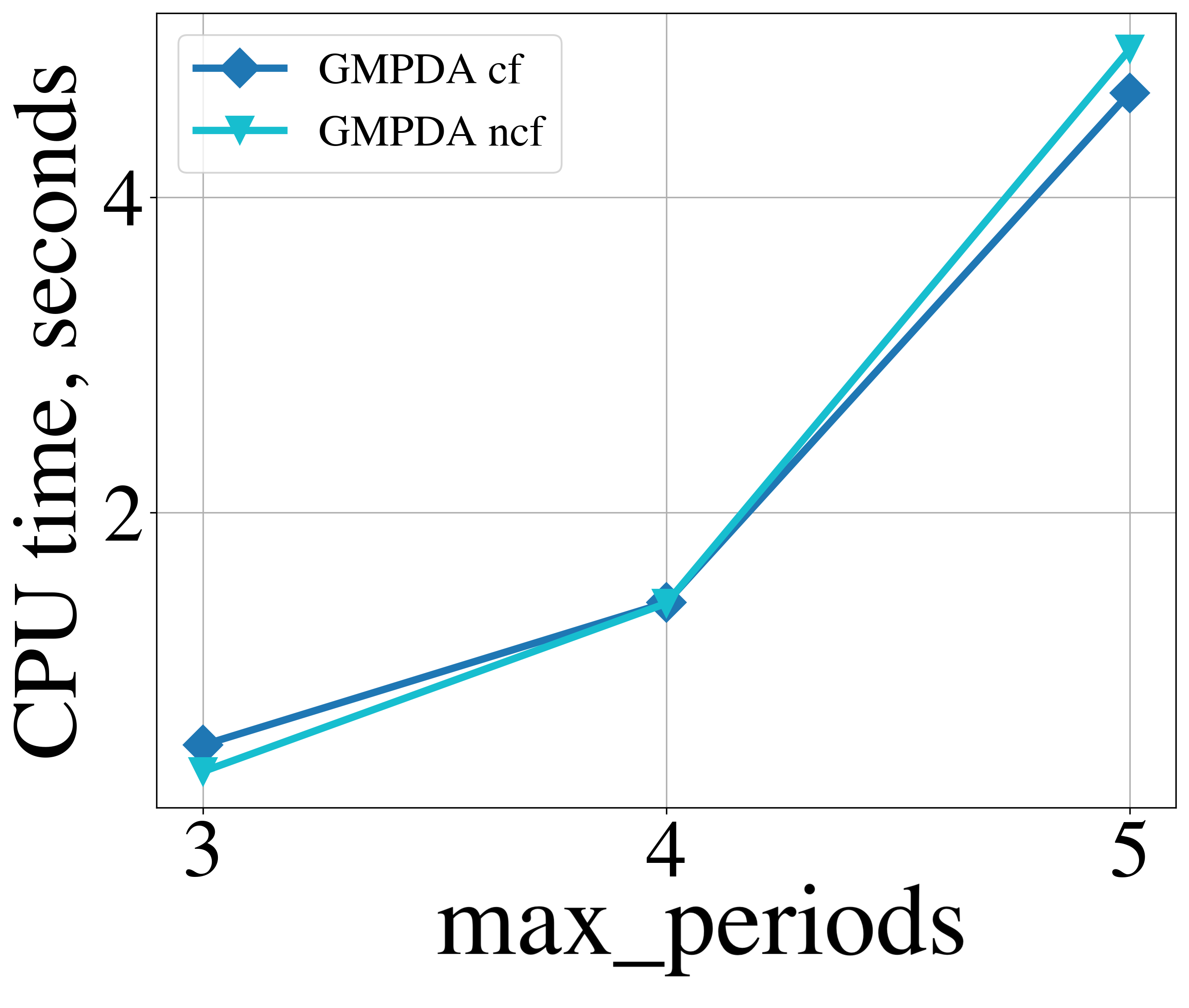}}}
    \subfloat[Clock Model]{{\includegraphics[width=0.22\textwidth]{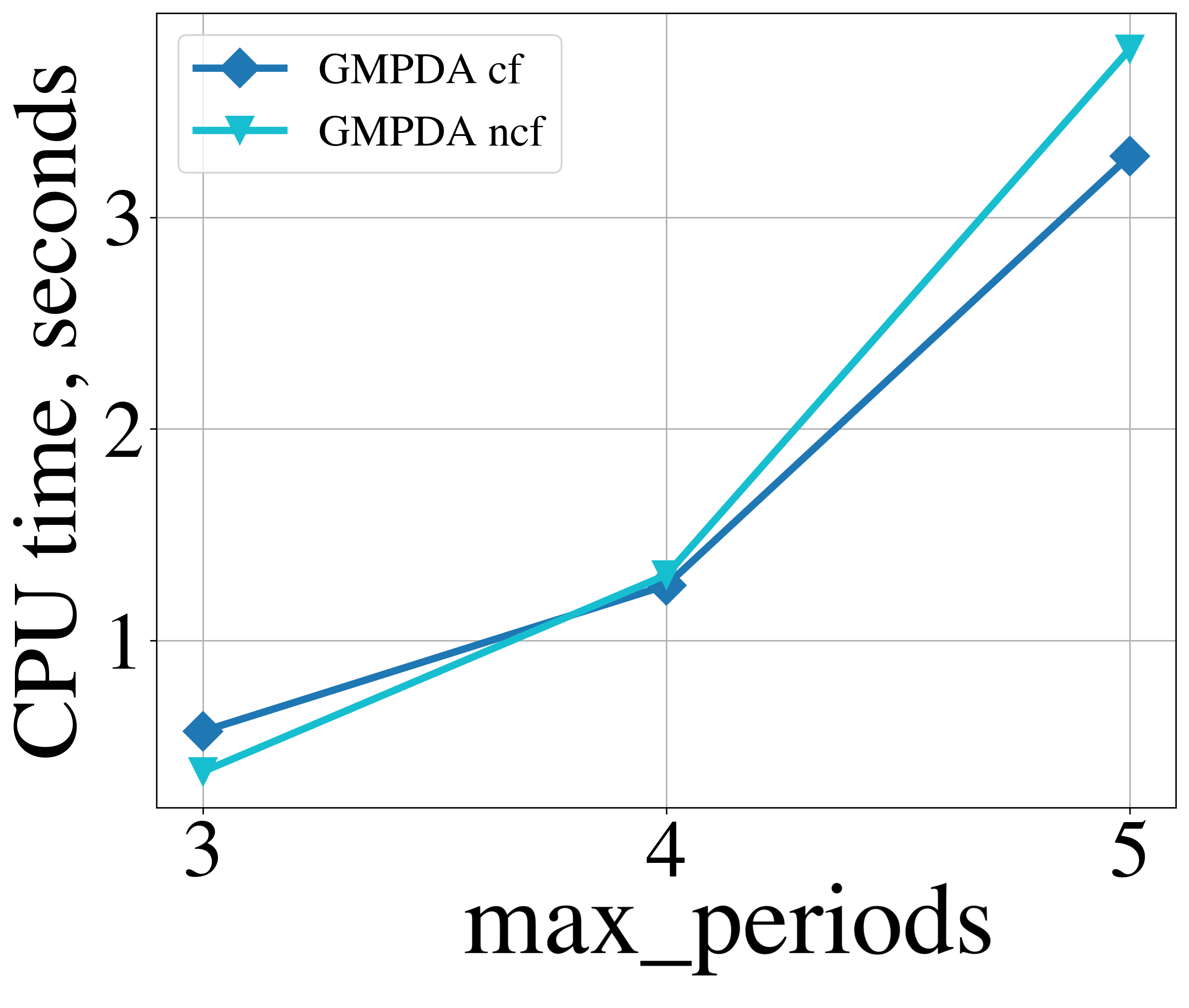}}}%
    \caption{Computational Performance averaged over 1200 executions.}
    \label{fig:cpurwm}
\end{figure}
In additional experiments not shown here, we also investigated the influence of noise, $\beta$, on the computational performance of the algorithm. The results indicated that although, on average, the CPU time increased slightly with increasing noise, $\beta$, the influence of minimal when compared to the maximal number of allowed periodicities, $max\_periods$. Finally, the maximal number of candidates periods, $max\_candidates$, will affect the CPU: a lower $max\_candidates$ resulted in faster execution time but with decreasing algorithm accuracy. %
\subsection{Summary}
We have evaluated the performance of the GMPDA algorithm for a large set of test cases, covering different configurations of the Random Walk and the Clock Models. Our main findings indicate that, first, for time series following the Random Walk Model, GMPDA outperformed alternative algorithms. Second, for time series following the Clock Model, GMPDA outperformed alternative methods in cases with a low variance of the inter-event intervals. All algorithms struggled to identify more than two periodicities.%

In addition, we analyzed the sensitivity to critical simulation parameters across the different algorithms and found that both sigma and the the number events emerged as the strongest determinants of periodicity detection accuracy. The details of the analysis can be found in Appendix~\ref{app:sensitivity_ana}.%
%
\section{Real Application}
\label{sec:real_application}
%
Finally, we also applied the GMPDA algorithm to real data and specifically to the recording of leg movements during sleep from the publicly available MrOS data set~\cite{zhang2018national, dean2016scaling},\cite{blackwell2011associations, blank2005overview, orwoll2005design}.

From 2905 available sleep recordings in community-dwelling men 67 years or older (median age 76 years), we considered all recordings with at least 4 hours of sleep, a minimum of scored events (10 leg movements, 10 arousals), and adequate signal quality based on various parameters in the MrOS database. This resulted in 2650 recordings satisfying our inclusion criteria, from which we randomly selected 100 recordings for this real application case. We have chosen to look at leg movements during sleep because it is known that in a relatively large proportion of the population (up to 23 percent~\cite{haba2016prevalence}), these leg movements tend to occur in a periodic pattern, the so-called periodic leg movements during sleep (PLMS)~\cite{ferri2016world}, with a typical intermovement interval around 20 to 40 seconds~\cite{ferri2017periodic}. We, therefore, expected to find some amount of periodicity in this data set, which - it could be argued - makes this analysis a real-life positive control.%
We applied GMPDA to raw data and preprocessed data. In the preprocessing step, the time series of leg movements of every subject was segmented into \textit{sleeping} bouts according to the following criteria: Each bout (i) contained only sleep interrupted by not more than 2 minutes of wake, (ii) lasted at least 5 minutes, and (iii) contained at least four leg movements. This resulted in 579 sleep bouts from the 100 recordings where GMPDA was applied independently to each bout. The number of events was less than 100 for $85\%$ of the bouts, and for those, the average bout length was 2572 seconds.%
\subsection{GMPDA Configurations}
The following GMPDA parameters were fixed for both data sets: $L_{min}=5$, $L_{max}=200$, $max\_iterations=5$, $max\_candidates=15$, $loss\_length=400$, $max\_periods=5$, $noise\_range=5$, $loss\_tol\_change=0.1$. We chose the tolerance value for a decrease in the loss to be 0.1. (i.e., additional periodicities are only considered if their inclusion results in a change of loss greater than this tolerance value). This value is substantially higher than in the simulated examples ($0.01$) because, in this first real-life application, we aimed to generate robust results with the expected noise in the data. In this sense, the results presented here and the periodicities identified can be seen as "low-hanging fruits."  Moreover, the detection of additional periodicities would be expected with different GMPDA parameters. For the MrOS data set, we assumed a Random Walk Model, which we applied both with and without the curve fitting of the variance parameter $\hat\sigma$. Consistent across all single records, the curve fitting approach identified periodicities with a lower loss, so that we will describe only the curve fitting results in the following. The GMPDA loss with and without curve fitting is compared in Appendix~\ref{app:MrOS_loss} Figure~\ref{fig:MROS_loss_cf_vs_ncf.pdf} and \ref{fig:MROS_loss_cf_vs_ncf_bouts.pdf} %
\subsection{Reference loss}
The GMPDA algorithm will identify the periodicity with minimal loss. However, even if minimal, this loss might still be numerically significant. In a real-life application where it can be assumed that some of the time series do not contain periodic events, there is a need to identify loss values that do not support the existence of periodicities in the data. We have chosen to address this issue by constructing a reference loss, which we derive from the minimal GMPDA loss returned for times series that only contain random noise. %

For the MrOS data set, the length of the included bouts and the number of events ranged from $300$ to $24000$ seconds and $5$ to $430$ events, respectively. In order to obtain an overall reference loss, we constructed 100 noisy bouts with uniformly distributed events for all different combinations of the number of events $[10, 30, 50, 100, 200, 400]$ and length of the bout $[500, 1000, 2000, 4000, 8000, 16000]$. Applying GMPDA to each combination, we obtained an empirical distribution of loss values for cases where the events were generated randomly and did not exhibit any clear periodic pattern. The global MrOS reference loss is set to the $0.01$ quantile of this distribution, corresponding to a value of $0.74468$, rounded to $0.75$ in the following.%

In addition, we also estimate a local reference loss for every single bout in the MrOS data set by generating 100 time series with the bout-specific length and the number of events and taking $0.01$ quantile of the resulting loss-distribution. A significant periodicity was identified when the GMPDA-loss for this bout was lower than the local reference loss. However, the significant periodicities obtained with local reference loss and global one did not differ significantly, and for simplicity, we focus on the results obtained for a global reference loss of $0.75$. %
%
%
\subsection{Results}
The distribution of the GMPDA model loss for all time series is shown in Figure~\ref{fig:mros_loss_len} and Figure~\ref{fig:mros_loss_len_bouts} for the whole night recording and the single sleep bouts, respectively. The figures suggest that GMPDA loss did not systematically change with the length of the times series.  However, the loss tended to decrease with the number of events in the time series, or more specifically - as already seen in the simulation experiments - for time series with a low number of events, the resulting loss was not distinguishable from the loss found for non-periodic time series. %
\begin{figure}[h]
    \centering
    \subfloat{\includegraphics[width=0.23\textwidth]{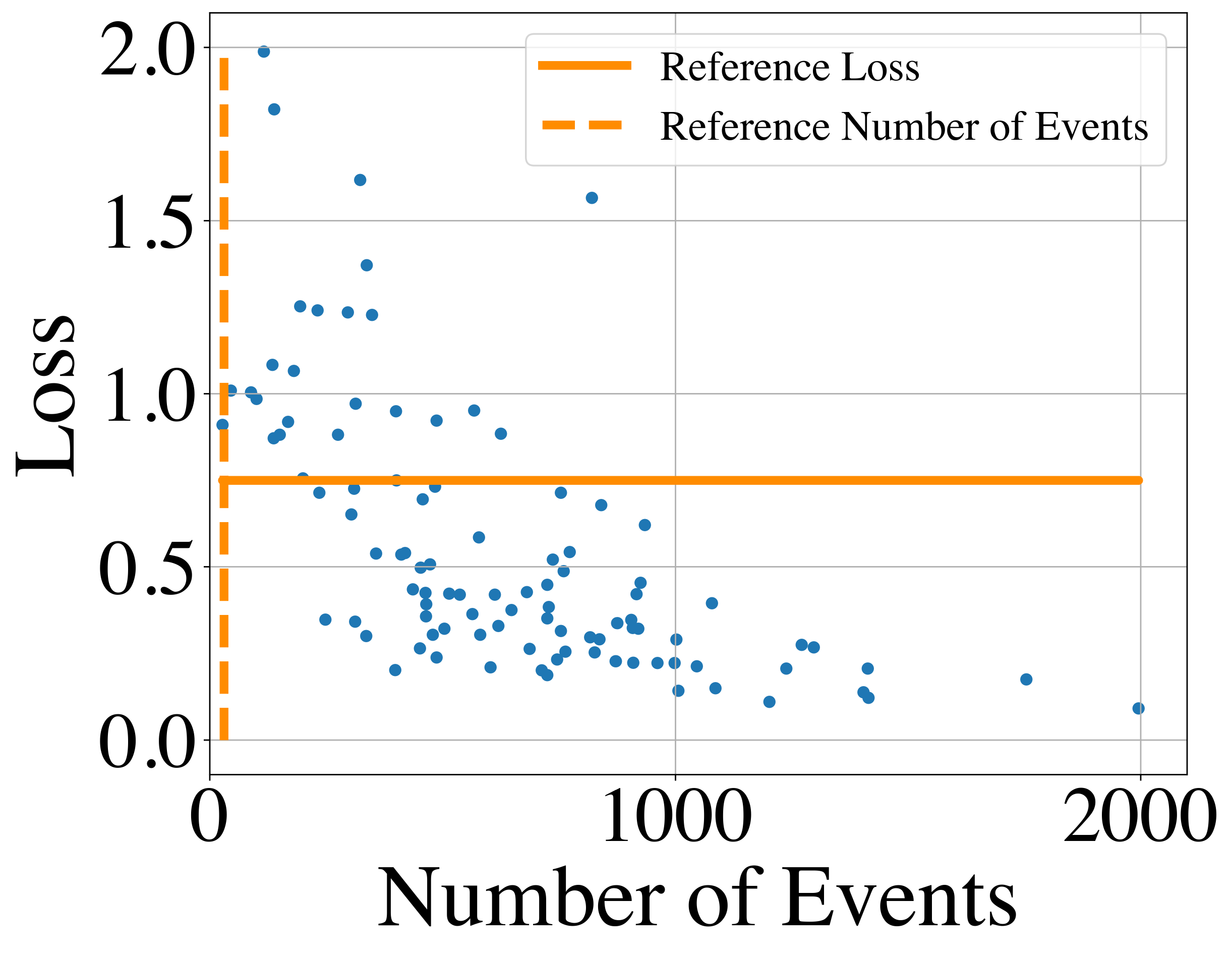}}
    \subfloat{\includegraphics[width=0.23\textwidth]{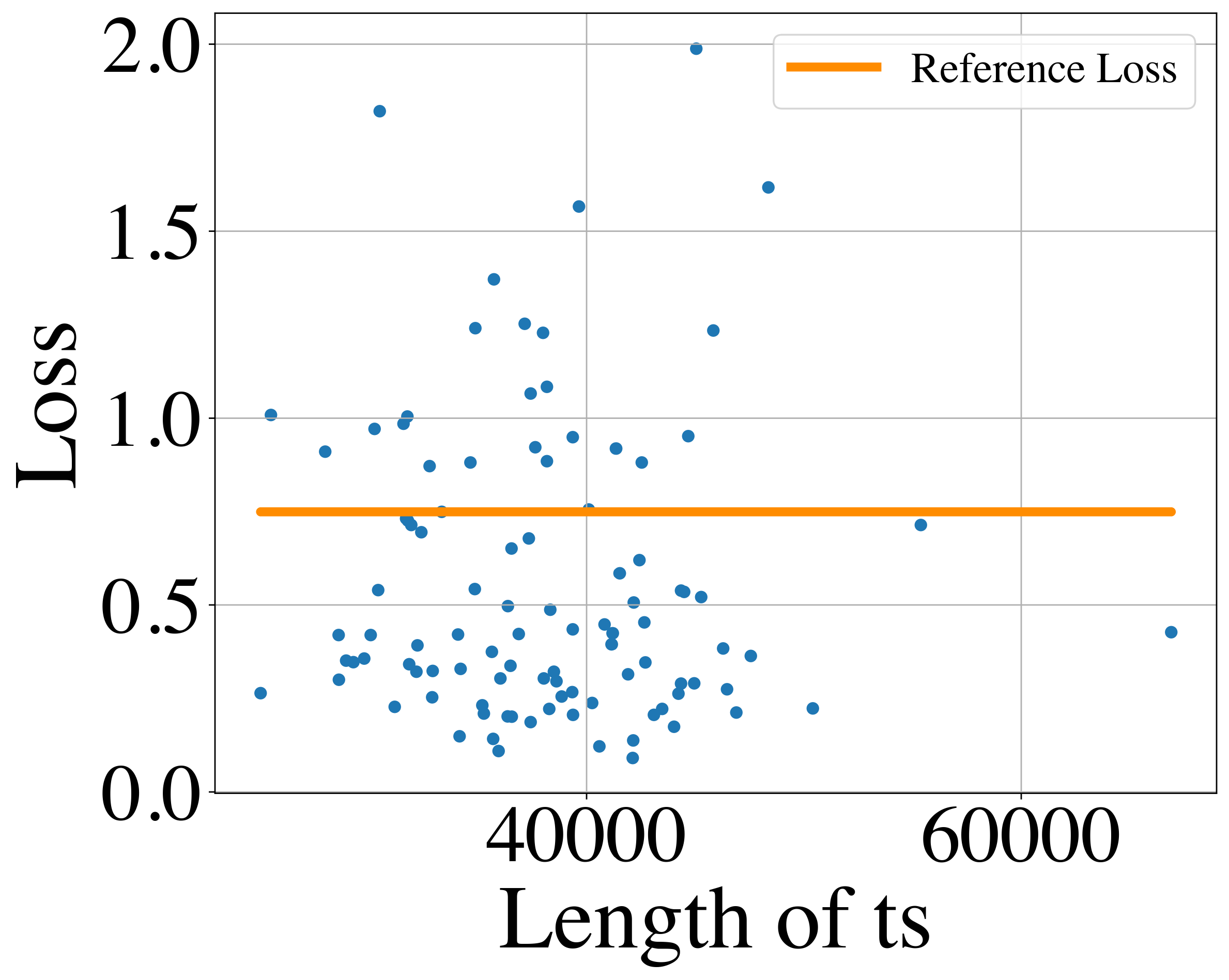}}%
    \caption{GMPDA loss for 100 whole night time series plotted against the number of events (left panel) and length of time series (in seconds, right panel)}
    \label{fig:mros_loss_len}
\end{figure}
\begin{figure}[h]
    \centering
    \subfloat{\includegraphics[width=0.23\textwidth]{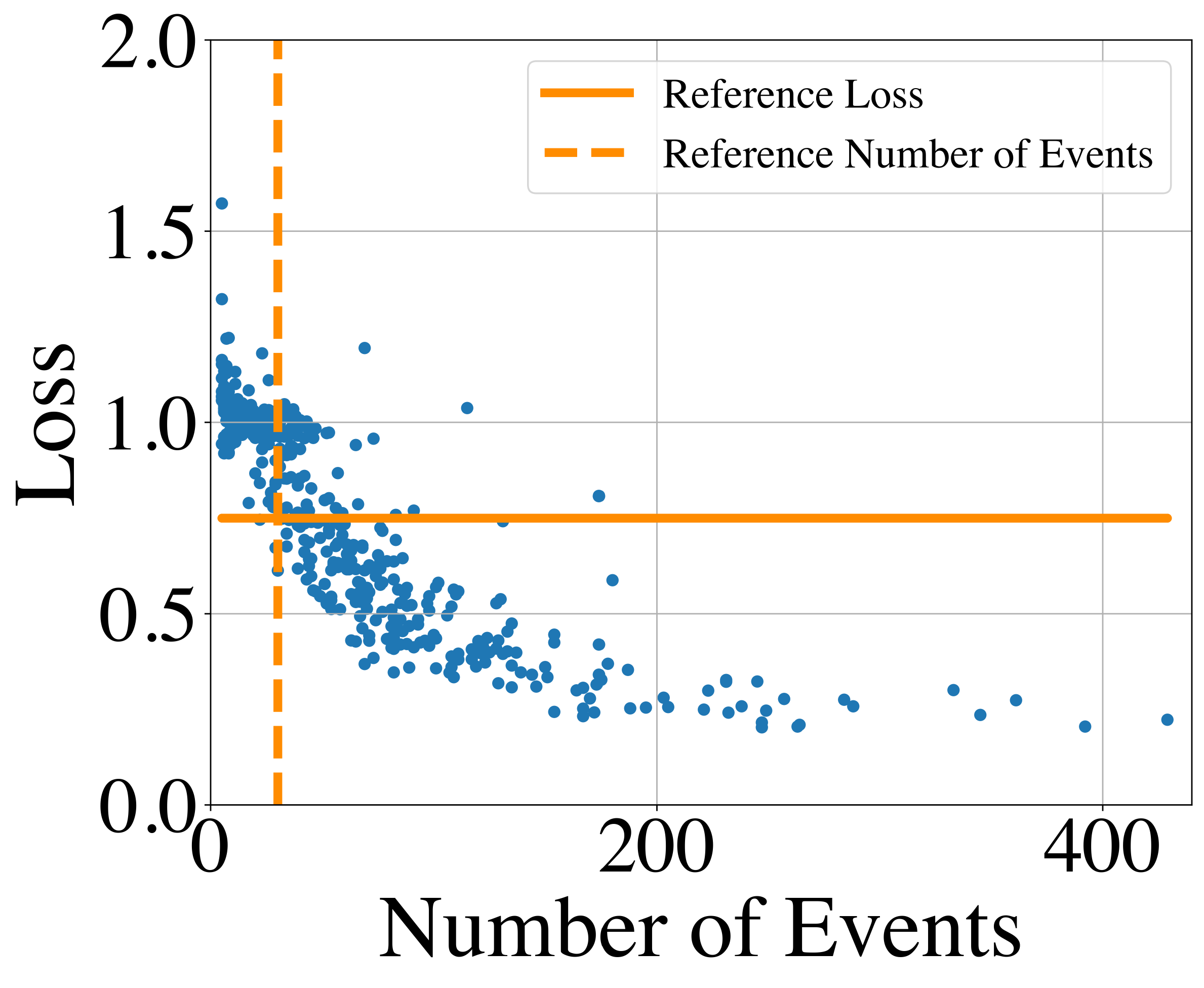}}
    \subfloat{\includegraphics[width=0.23\textwidth]{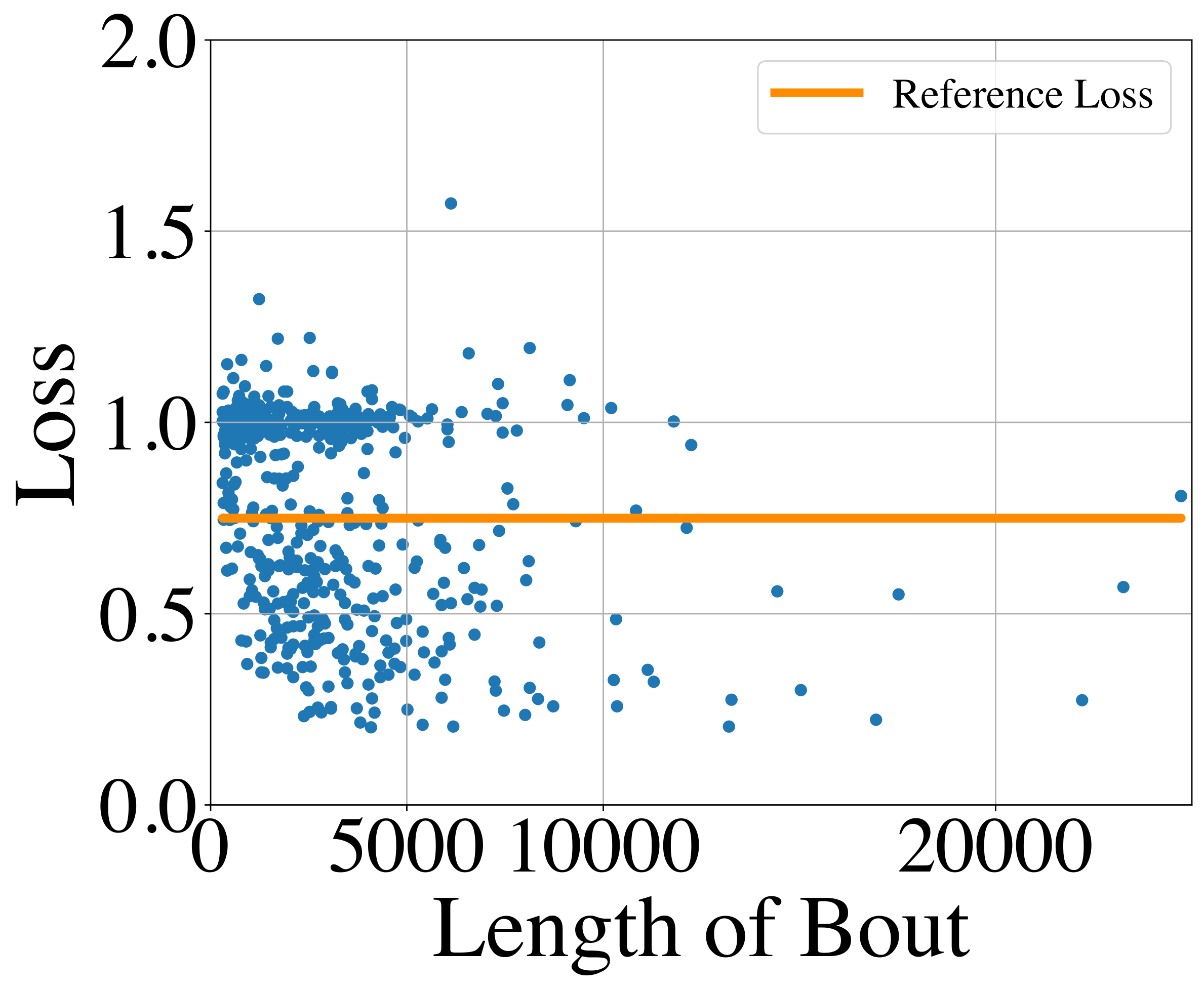}}%
    \caption{GMPDA loss for 579 sleep bouts of at least 5 minutes plotted against the number of events (left panel) and the length of the sleep bout (in seconds, right panel)}
    \label{fig:mros_loss_len_bouts}
\end{figure}
Please note that the left panel of Figure~\ref{fig:mros_loss_len_bouts}, which shows the distribution of the loss for the number of events in the MrOS data set, could also be used to suggest a minimum number of events needed for the GMPDA algorithm to detect a significant periodicity in this data set. For the records selected here, no significant periodicity was detected for any bout with less than 30 events (see reference number of events in the Figure). Other records from the same data set and new data sets are needed to determine whether this reference number constitutes an absolute threshold for bio-medical event data. %

From the $579$ sleep bouts $183$ (31.6\%) and from $100$ whole night time series $75$ had a loss below $0.75$. The corresponding histograms of the significant periodicities extracted from the signals by GMPDA are shown in Figure~\ref{fig:MrOs_hist_good_loss_bouts} and Figure~\ref{fig:MrOs_hist_good_loss}. %
\begin{figure}[h]
    \centering
    {\includegraphics[width=0.48\textwidth]{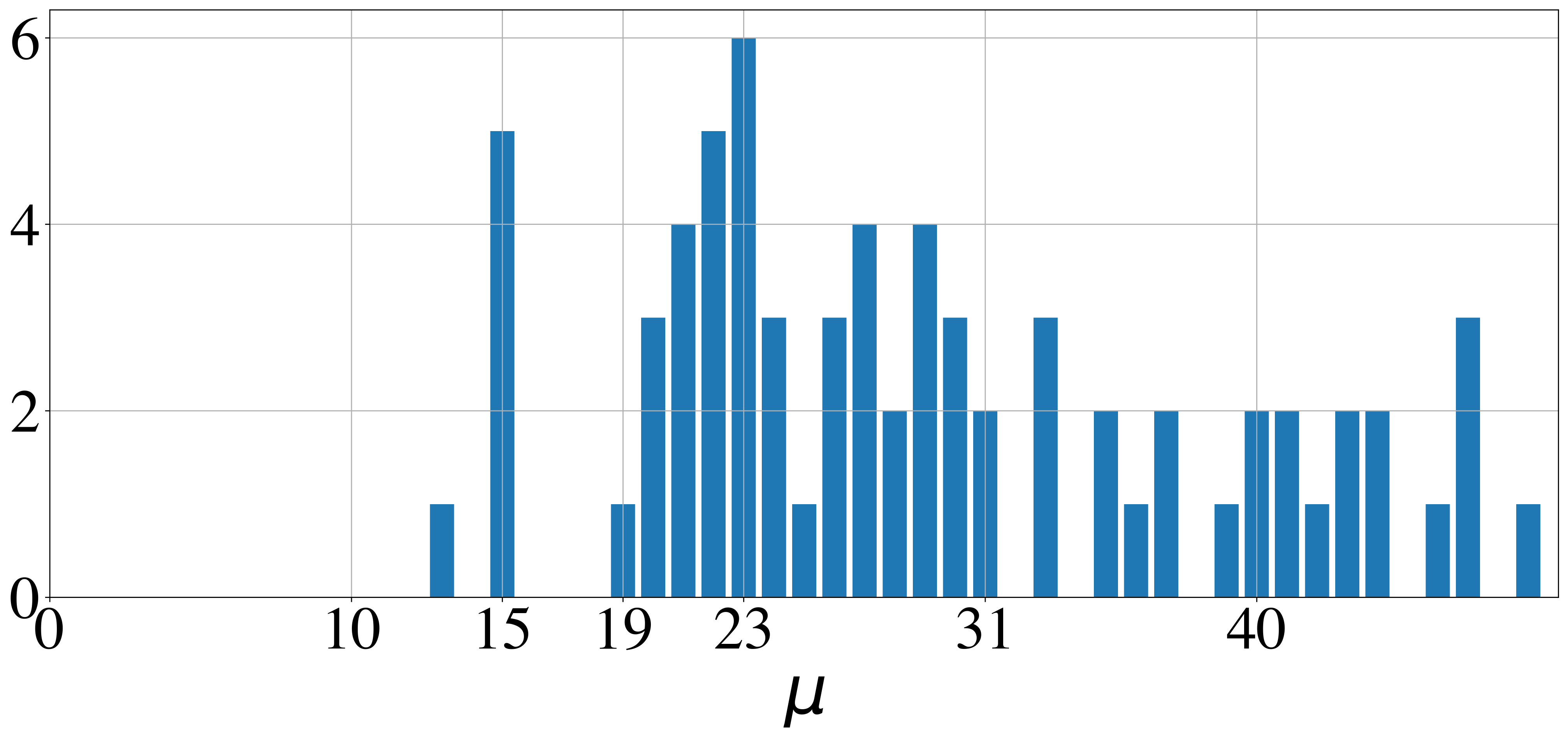}}%
    \caption{Histogram of significant periodicities identified in 100 whole night time series by GMPDA.}
    \label{fig:MrOs_hist_good_loss}
\end{figure}
\begin{figure}[h]
    \centering
    {\includegraphics[width=0.48\textwidth]{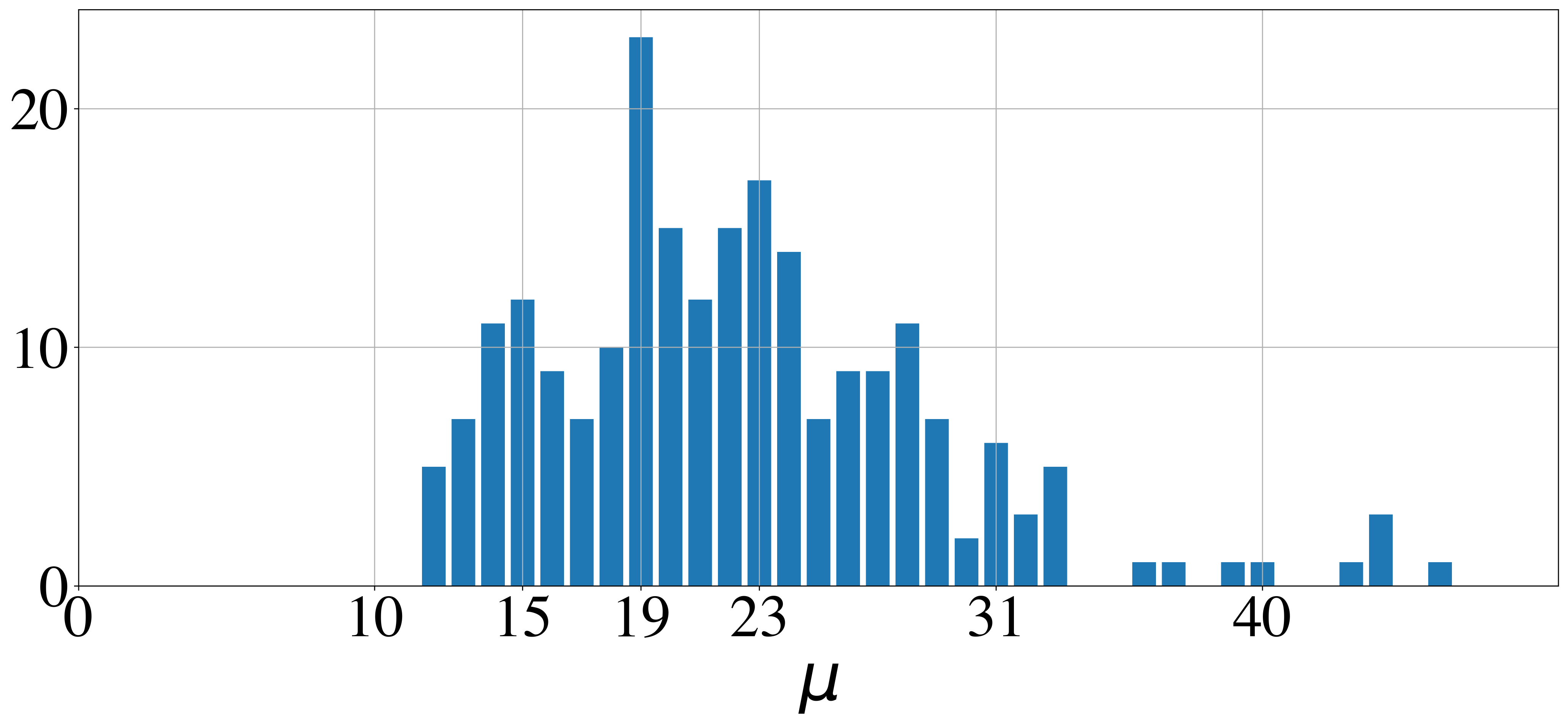}}%
    \caption{Histogram of significant periodicities identified in 579 sleep bouts, 5 minutes or longer, by GMPDA.}
    \label{fig:MrOs_hist_good_loss_bouts}
\end{figure}
In both Figures, the expected peak in periodicities is around $20$. Another minor peak (rather unexpected) is at $15$. Significant periodicities ranged from 10 to 33 seconds (except two bouts with a periodicity of 49 and 192 seconds). Periodicities around $20$, i.e., $\mu\in[17, 18, 19, 20]$, were present in 95 bouts (out of $183$), from $77$ subjects (out of $100$). Periodicities around $15$, i.e., $\mu\in[12, 13, 14]$, were present in $30$ bouts from $18$ subjects.%

Although the minimal periodicity and noise range were set to five, $\mu=12$ was the smallest periodicity identified by the algorithm for \textit{significant} bouts.%
%
\section{Conclusion}
\label{sec:conclusion}
%
%
In this paper, we developed the Gaussian Mixture Periodicity Algorithm (GMPDA) to address the overlapping periodicity detection problem for noisy data. The GMPDA algorithm is based on a new generative model scheme that accounts explicitly for a Clock Model and a Random Walk Model. The Clock Model describes periodic behavior in systems, in which variances do not change over time because a pacemaker governs the events in the system. Examples for a Clock Model are scheduled or seasonal behavior - like traffic patterns guided by working hours or migration patterns governed by seasons. In contrast, in Random Walk Models, the variances increase over time, making distant temporal predictions difficult or impossible. The Random Walk Model describes a biological behavior - like footsteps or gene expression - where events only depend on the interval to the last event. %

The main entry point for GMPDA is the empirical histogram of all forward order inter-event intervals, i.e., for every event, we consider not only the interval to the next event (onset to onset) but also to all subsequent events. This histogram contains information about the underlying prime periodicities but also about the interaction noise between events associated with different periodicities and about false positive noise. We approximate the overall noise by an explicit formulation under the assumption of the noise being uniformly distributed. The approximation accounts for all interaction intervals, which length is limited by a user-defined parameter in the GMPDA algorithm. After its subtraction, the GMPDA Algorithm hierarchically extracts the multiple overlapping periodicities by minimizing the loss, which is defined as the absolute difference between the parametrized histogram obtained by the generative scheme and the empirical histogram. %

GMPDA is implemented in a computationally efficient way and is available open-source on \url{https://github.com/nnaisense/gmpda}. %
We have demonstrated its performance on a set of test cases. Test cases included up to three overlapping periodicities with different values for the involved Gaussian noise and the different number of events. For the Random Walk Model, we can conclude that GMPDA outperformed the FFT and Autocorrelation-based approaches and the EPeriodicity algorithm in identifying true prime periodicities. For the Clock Model, GMPDA outperformed other algorithms in cases with a low variance of intervals.

GMPDA performed well in the presence of noise for a signal-to-noise ratio of 1:1, and it performed adequately up to a ratio of 1:2, with an appropriate number of events. This appropriate number of observed events depends on the signal-to-noise ratio. However, it needs to include more than 30 actual periodic events for GMPDA to identify any periodicity. %

Finally, we applied GMPDA to extract significant periods in real data, focusing on leg movements during sleep. The main results here were (i) that GMPDA was able to identify the expected periodicities around 20 seconds, (ii) that we have introduced a procedure to identify a data set dependent reference loss (of $0.75$) that could be used to distinguish significant from spurious periodicities, (c) and that our results suggest that a minimal number of events ($30$) required by GMPDA for performing periodicity detection successfully in biomedical data. %

The general nature of the generative framework and the formulation of GMPDA allow an alternative statistical parametrization for the event data. An extension, for example, would be to model events as a Poisson process, which for multiple periodic generative functions could similarly be modeled in terms of a sum of scaled probability density functions. Further, GMPDA can be extended towards periodicity extraction in non-stationary event time series. One approach could be to assume that the time series under investigation can be divided into locally stationary segments. Then, deploying a bottom-up-based segmentation strategy, we could estimate in an alternating manner the optimal switching points between the segments and the underlying prime periodicities for each stationary segment. Another approach could be to incorporate the Monte Carlo based particle approach for an adaptive periodicity detection presented in~\cite{ghosh2017finding}. This would allow GMPDA to adapt to non-stationary changes and remains for future work. %
%
%
%
%
\appendices

\section{GMPDA Algorithm}
\label{app:GMPDA_ALGO}
In the section the steps involved in the GMPDA algorithm are outlined in very detail. 
\subsection{Approximation of $|\zeta(\mu)|$}
\label{subsub:approx_zeta}
In equation (\ref{eq:zeta2}) we estimate the number of interaction intervals as a decreasing linear function $\mathbb E[\hat\zeta(\mu)] = z(1-\frac{\mu}{N_T})$. An approximation of $z$ is required, as the information about the amount of noise and the number of true periodicities is unknown a priori. Here, we propose the following approximation:
\begin{align}
    \label{eq:zetaz}
    \hat{z} = \frac{1}{z_{min}}\sum\limits_{i=1}^{z_{min}} D(\mu)(i),
\end{align}
This approximation follows the idea that first, $z$ should be close to the maximal value of $\mathbb E[\hat\zeta(\mu)]$ and that second, all non zero contributions in $D(\mu)$ for $\mu < \argmin\limits_{\mu}\mu^*$ are due to the interaction intervals, asymptotically for an increasing number of events and/or increasing number of involved prime periodicities. For the GMPDA algorithm we set default $z_{min}=L_{min}-1$. Thus, in the algorithm, the length of the interaction intervals is limited either by the minimal expected periodicity or can be adjusted by the user.

We verified this approximation empirically on a set of 50000 test cases with and without noise and a priori known two random chosen prime periodicities $\mu*_{1/2}\in[10,60]$: $\mathbb E[\hat{\zeta}(\mu)]$ with $\hat{z}$ provided on average a good linear fit, since the distribution of the mean errors between $\mathbb E[\hat{\zeta}(\mu)]$ and $\mathbb E[\zeta(\mu)]$ was centered at zero and was approximately normal (data not shown).

However, it must be stressed that the assumption of an uniform distribution of interaction intervals between the periodic and the noise events maybe unrealistic in real life data sets and there is currently no alternative available for estimating zeta from the observed data. This remains an area of improvement for the GMPDA algorithm.
%
\subsection{Candidate Period Identification}
\label{subsub:NacdidatePer}
%
The proposed algorithm hierarchically extracts a set of candidate periodicities which can explain $D(\mu)$ using an integral convolution approach. The method works by iteratively selecting periodicities which explain many of the observed intervals, and then subtracting the integer multiple intervals which can be explained by these periodicities.

The algorithm takes as input some guess $\hat{\sigma}$, a range in which to search for periodicities $\{L_{min}, L_{max}\}$, a number of hierarchical periodicity extraction steps $max\_iterations$, and a maximal number of periodicities to extract at each hierarchical iteration, $max\_candidates$.

Recall that $D(\mu)$ counts the number of times a given interval $\mu$ appears between any two events in the time series.
\begin{align}
\label{eq:algo_dmu}
    D(\mu) = \sum_{i,m >0} \mathds{1}_{s_{i+m} - s_i = \mu}.
\end{align}
One tempting method would be to select $\argmax\limits_{\mu}D(\mu)$ as the first prime period $\hat{\mu}_1$. However, for small or noisy real world data $\argmax$ may not be the prime period.

We introduce the notion of an \textit{integral convolution}, in which we integrate around a fixed $\mu$ to capture how much of the observed intervals in $D(\mu)$ are explainable by that particular "mean" periodicity, which also acts to smooth $D(\mu)$. We therefore define a function $\tau(\hat{\mu},\hat{\sigma}, D(\mu))$ which will act as a symmetric convolution kernel across $D(\mu)$, centered at the candidate means $\hat{\mu}$. This function provides a point-wise estimate of the explained data for a given $\mu$ in $D(\mu)$:
\begin{align}
    \label{eq:tau_mu}
    \tau(\hat{\mu},\hat{\sigma}, D(\mu)) = \int_{\hat{\mu} -(1.96*\hat{\sigma})}^{\hat{\mu} +(1.96*\hat{\sigma})} D(\mu) \partial \mu.
\end{align}
Because we don't want to calculate the full loss function (\ref{eq:loss_master}) at this stage due to computational expense, we approximate our loss function with the function $\tau(\hat{\mu}, \hat{\sigma}, D(\mu))$ for the Clock Model:
\begin{equation}
    \label{eq:expectationclock}
    \mathbb E(\hat{\mu}_p) = \frac{\hat{\mu}_p}{N_T} \cdot [\sum_{i=1}^{\frac{N_T}{\hat{\mu}_p}} \tau(i \cdot \hat{\mu}_p, \hat{\sigma},D(\mu))],
\end{equation}
and for the random walk:
\begin{equation}
    \label{eq:expectationrw}
    \mathbb E(\hat{\mu}_p) = \frac{\hat{\mu}_p}{N_T} \cdot [\sum_{i=1}^{\frac{N_T}{\hat{\mu}_p}} \tau(i \cdot \hat{\mu}_p, i \cdot \hat{\sigma},D(\mu))].
\end{equation}
Functions (\ref{eq:expectationclock}) and (\ref{eq:expectationrw}) approximate, for each candidate prime period, how much of the data can be explained by this periodicity in some confidence interval about $\mu_p$. If a periodicity is present and is persistent through the time series, integer multiples of the periodicity $\hat{\mu}_p$ will also frequently appear in (\ref{eq:expectationclock}) and (\ref{eq:expectationrw}), we can use this information to select the periodicity which explains the most data.

Once the first periodicity - $\hat{\mu}_1$ - has been identified, we remove intervals found in $D(\mu)$ which can be explained by $\hat{\mu}_1$, i.e. integer multiples of $\hat{\mu}_1$. We realize this by setting $D(\mu)$ to zero for $\mu \in [i\cdot \hat{\mu}_1+\hat{\sigma},i\cdot \hat{\mu}_1-\hat{\sigma}],\, i=1,\dots$, and recompute $\tau(\mu)$ from $D(\mu)$ missing these intervals.

The GMPDA algorithm performs reasonably well at identifying the true periods $\mu^*$ as $\hat{\mu}$ without the use of the loss function or adjustments to $\hat{\sigma}$. But without a measure of relative goodness of this estimates, we have no stopping criteria for finding multiple periodicities. Instead, we repeat this procedure for $max\_iterations$ iterations.

Once we have initialized (hierarchically) a set of candidate prime periods $\hat{\mu}^{init}$ using this "fast" method, we compute a better estimate of the variance and loss using methods which will be elaborated in the following sections.


\subsection{Non-linear least squares fitting for $\hat{\sigma}$}
\label{sub:LS_sigma}
We can improve our guess of the variance $\hat{\sigma}$ by formulating a non-linear least squares curve fitting optimization problem, in which our set of parameters are those of a Gaussian PDF. That is, here we consider $D(\mu)$ which can be modeled as the sum of Gaussian PDFs, and a set of candidate means $\hat{\mu}$ for those Gaussian PDFs. For a fixed set $\hat{\mu}_p$ we initialize guesses for $\hat{\sigma}_p$, for $p=1,\dots,P$, and deploy Trust Region Reflective algorithm to obtain an update for the guesses of $\hat{\sigma}_p$. It is implemented with $curve\_fit()$ from Scipy's optimization package.
%
\subsection{Selecting true parameters: Loss Function}
\label{sec:loss_function}
%
The parameter estimates $\hat{\mu}^{init}$, $M$, $\hat{\sigma}$ are assessed with respect to $D(\mu)$ - the observed intervals between events in the time series - using a loss function. This loss function describes the proportion of intervals in the data which can be explained with (i) the parametrized "generative function" $G_M(\hat{\mu}, \hat{\sigma})$ implicated by the estimates, which is asymptotically the same as the expectation of $D(\mu)$ if $\hat{\mu} = \mu^*$ and $\hat{\sigma} = \sigma^*$ and  (ii) the noise approximation. 

Computing the loss function (\ref{eq:loss_master}) is expensive because in the case of the Random Walk Model $G_{RW}$ requires computation at increasingly large intervals since the variance terms grow linearly, and thus the area covered with some density by a single Gaussian distribution grows at the same rate. Thus, we only want to compute $G_M$ for a few very probable periodicities (the set $\hat{\mu}^{init}$ computed in the fast algorithm), using the optimal variance guesses $\hat{\sigma}$, and only across a limited range of intervals specified by $loss\_length$ chosen a priori.

We also adjust the scaling factor of the generative function to account for sections of the time series which may not have any events, for instance missing values, or large intervals of the time series with no observations).The concerns the scaling factor $c_p$, for $p=1,\dots,P$, of $G_M(\hat{\mu}, \hat{\sigma})$ for Clock Model and Random Walk Model in equations (\ref{eq:P_Mfull_Clock}) and (\ref{eq:P_Mfull}), respectively. The adjusted scaling factor is:
\begin{align}
    \label{eq:new_scalingfactor}
    c_p=\frac{N_T(\hat\mu_p)}{\hat\mu_p}-(m-1),
\end{align}
where $N_T(\hat{\mu}_p)$ is the sum of intervals which are smaller than $\hat\mu_p + (\hat{\sigma}_p\cdot 2)$. This correction ensures that we only count "possible appearances" in the time series on sections which actually have events. Without this scaling factor, missing values would bias our results towards higher frequencies and the scaling factor would be far too large for lower frequencies which may appear in the time series, but with intervals of no-events.

Our final loss function will therefore be constructed using $D(\mu)$ computed from the real data, $\mathbb E[\hat{\zeta}(\mu)]$, $G_M(\hat{\mu}_{init}, \hat{\sigma})$, and one additional parameter, $loss\_length$. This parameter manages high variance at high integer multiples and decreases the computational complexity. In the Random Walk Model, for high integer multiples of a periodicity, the implied Gaussian distributions of intervals begin to have large tails and the distributions density in mean decreases. Meanwhile for the Clock Model, estimating many integer multiples is not actually necessary to compute the true periods. Therefore, the loss we compute in the algorithm is:
\begin{equation}
    \label{eq:loss_empirical}
    \hat{\mathcal{L}} = \sum_{\mu = 0}^{loss\_length} |  D(\mu) - \hat{\zeta}(\mu) - G_M(\hat{\mu}), \hat{\sigma}|,
\end{equation}
Within the algorithm, we compute $G_M$ for all combinations of the set $\hat{\mu}^{init}$ up to order $max\_combi$, and select our true set of periodicities as that which minimizes (\ref{eq:loss_empirical}).

Please note, the number of true periodicities $max\_combi$ is not known a priori, the optimal value of $max\_combi$ will minimize the loss. However, in real applications we might have situations where there are weak peaks in $D(\mu)$ around very large $\mu$ due to noise or the influence of large/slow interactions intervals. Adding these to the set of prime periodicities will decrease the loss, but will not contribute to the identification of intrinsic periodicities.
To account for this, GMPDA provides the possibility to control the magnitude of the loss decrease by a parameter $loss\_decrease\_tol$, with the loss being typically of magnitude $1$ and lower, see Section~\ref{sec:real_application}. That is, setting this tolerance parameter to a very low number, e.g., $loss\_change\_tol=0.001$, will result in including more periodicities (that might be due to noise), while a larger number, e.g., $loss\_change\_tol=0.1$, will be more conservative.
%
%
\section{Performance}
\label{app:performance}
%
\subsection{$|\mu|=1$}
\label{app:|mu|=1}
Here the performance of GMPDA with respect to noise with curve fitting and without curve fitting is presented.
\begin{figure}[h]
    \centering
    \subfloat[No curve fit]{{\includegraphics[width=0.23\textwidth]{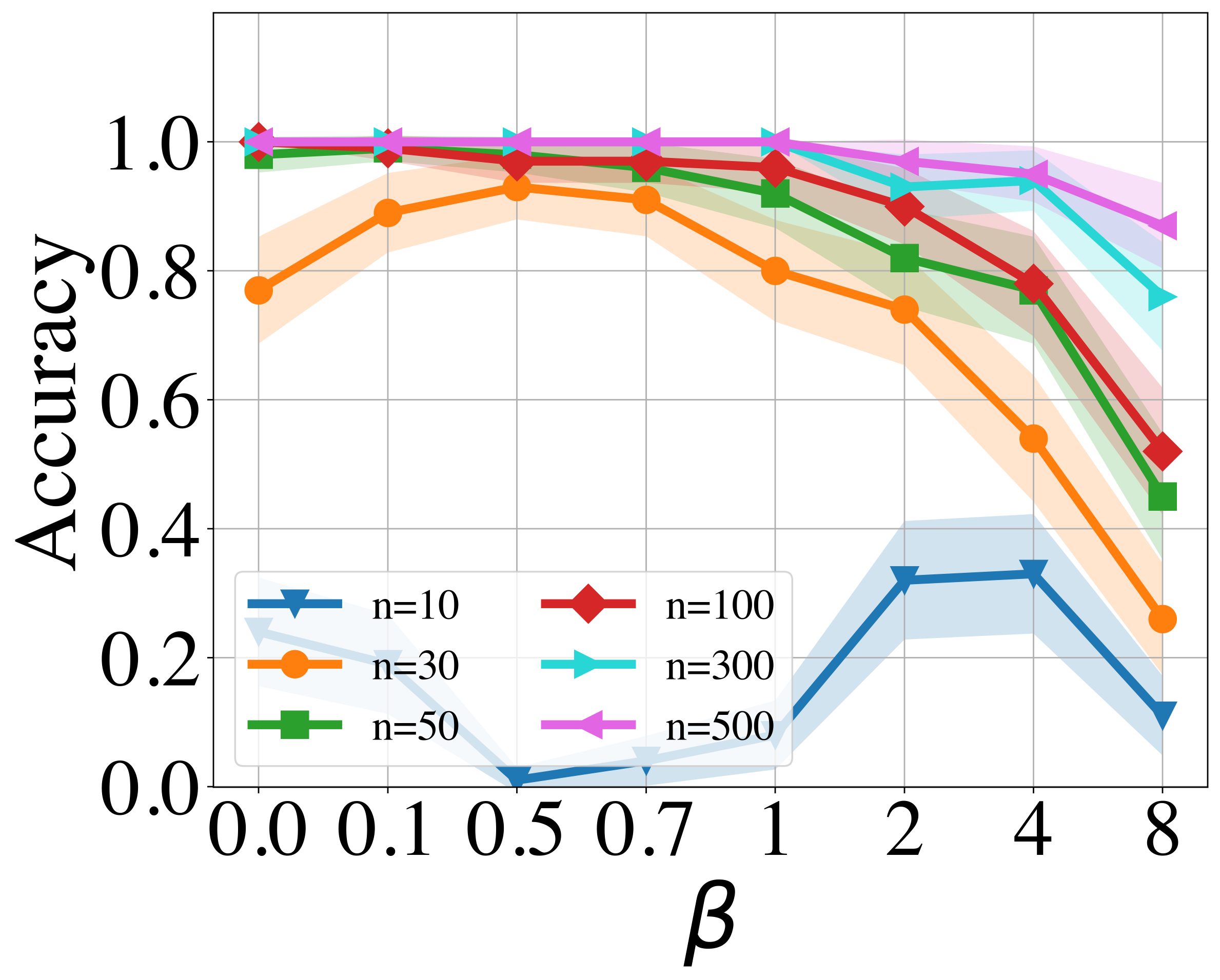}}}
    \subfloat[With curve fit]{{\includegraphics[width=0.23\textwidth]{pics/mu1_snkn_rw_acc.pdf}}}
    \caption{Random Walk Model Performance w.r.t. $\beta$, for $\sigma=log(\mu)$ and $|\mu|=1$.}
    \label{fig:perf_rw_sigma_cf_ncf_app}
\end{figure}
The Random Walk Model model exhibits a decay in performance with an increasing noise for $n>30$. For $n=30$, the performance increases  till $\beta=0.5$ as the noise, to a certain extent, is acceptable due to definition of the variance for the Random Walk Model\ref{eq:rw_model}.
\begin{figure}[h]
    \centering
    \subfloat[No curve fit]{{\includegraphics[width=0.23\textwidth]{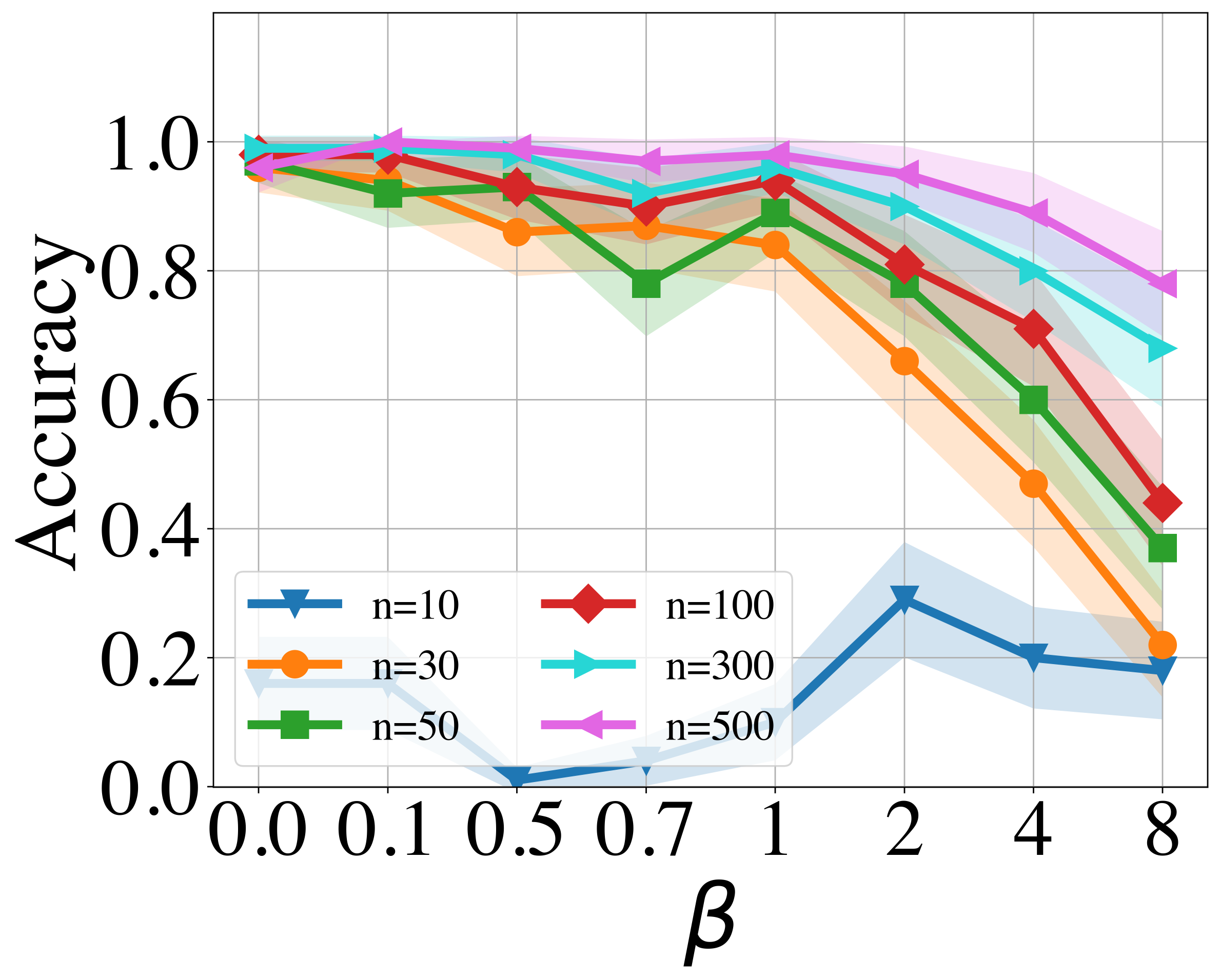}}}
    \subfloat[With curve fit]{{\includegraphics[width=0.23\textwidth]{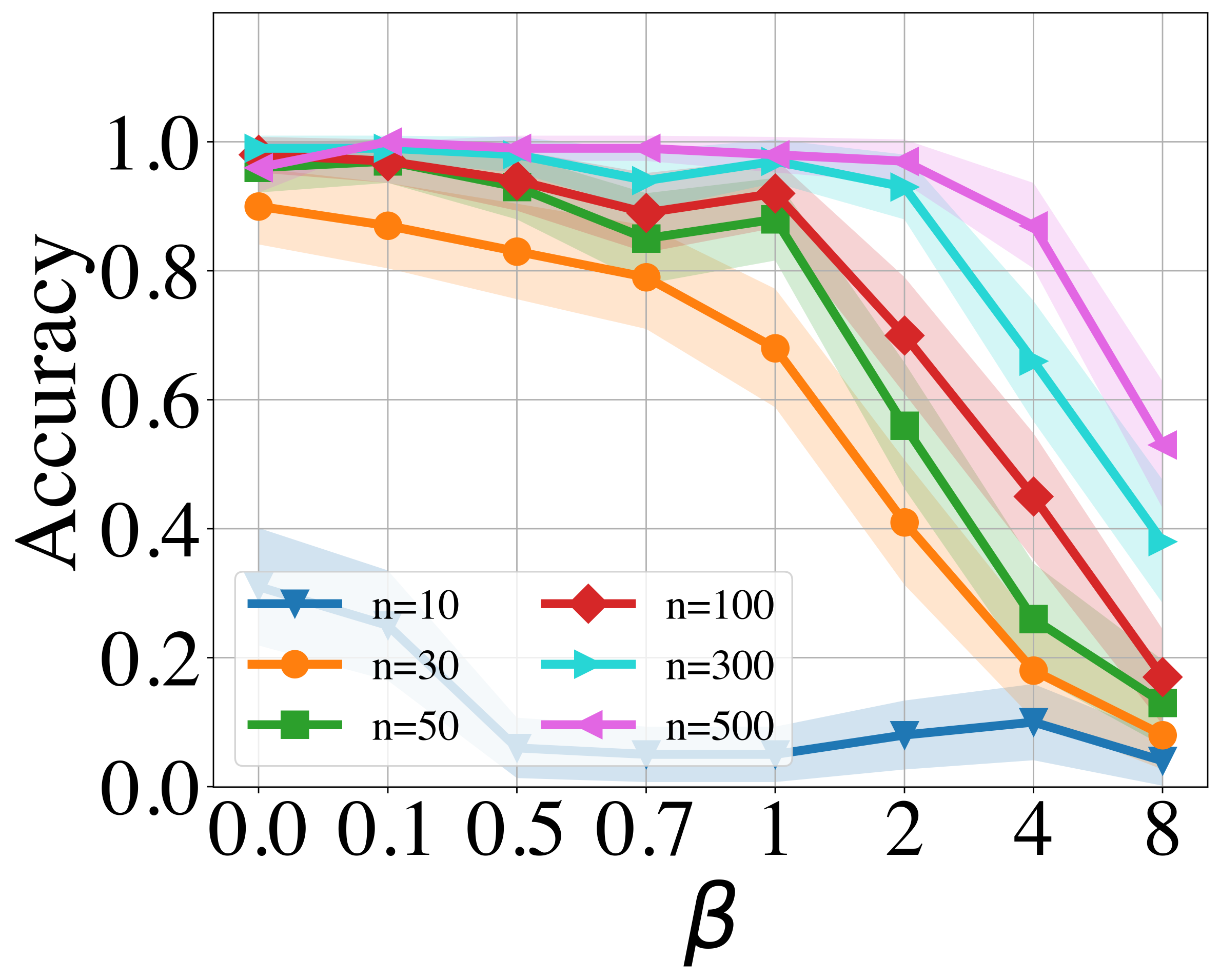}}}
    \caption{Clock Model Performance w.r.t. $\beta$, for $\sigma=log(\mu)$ and $|\mu|=1$.}
    \label{fig:perf_cl_sigma_cf_ncf_app}
\end{figure}
The Clock Model exhibits for $n>10$ a decay in performance with an increase in $\beta$.
For both models the case $n=10$ performs not sufficient, indicating that the number of events must be definitely higher then ten.
%
\subsection{Comparison to alternative Methods wrt. Noise $\beta$}
\label{app:noise_beta}
Here the performance of GMPDA and of the alternative methods is compared regarding an increase in noise.
In the following figures the accuracy of all the involved methods is plotted for $|\mu=1|$, different number of events $n$, while it was averaged over all considered values of variance $\sigma$.
\begin{figure}[h]
    \centering
    \subfloat[n=10]{{\includegraphics[width=0.23\textwidth]{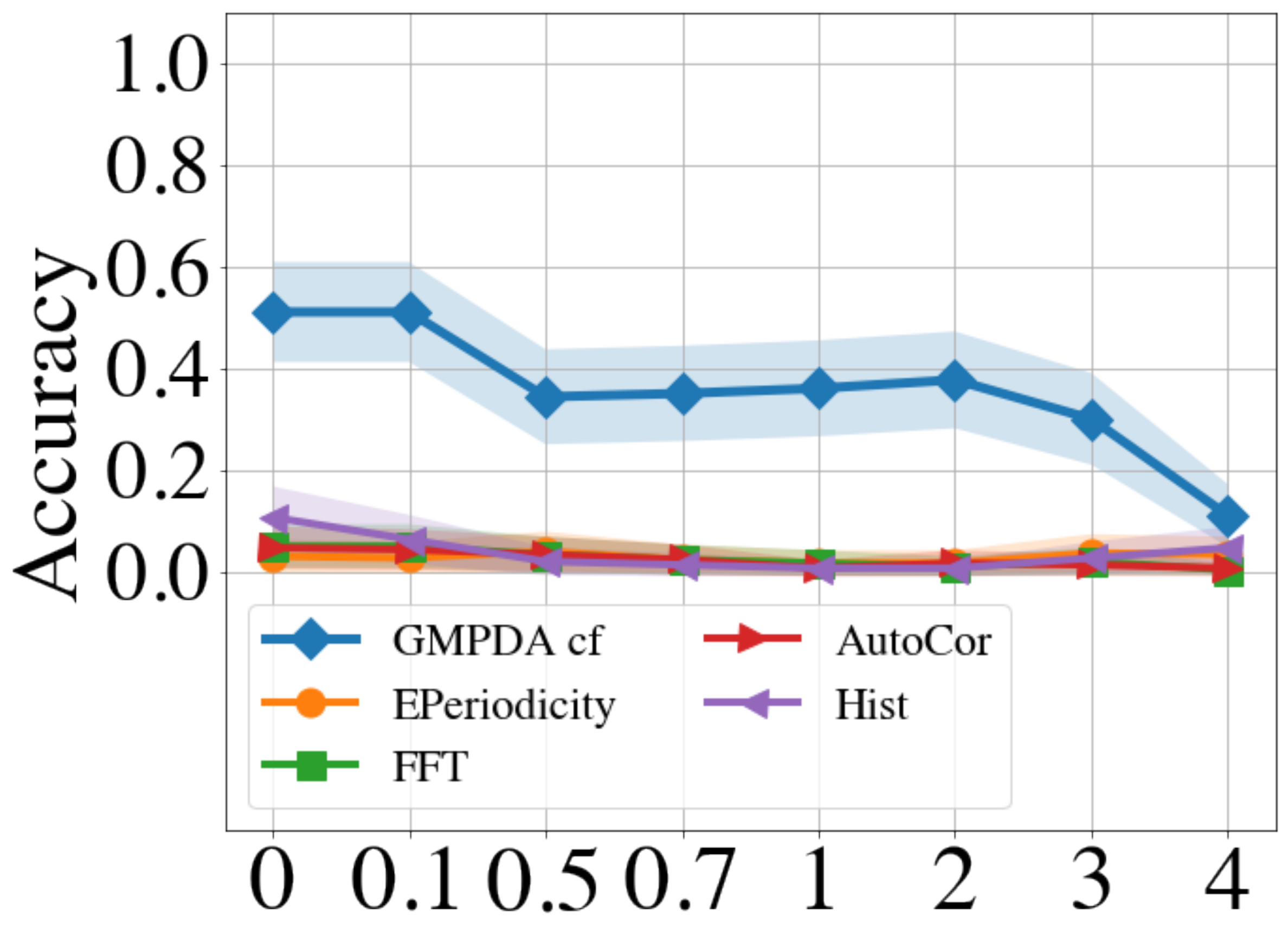}}}
    \subfloat[n=30]{{\includegraphics[width=0.23\textwidth]{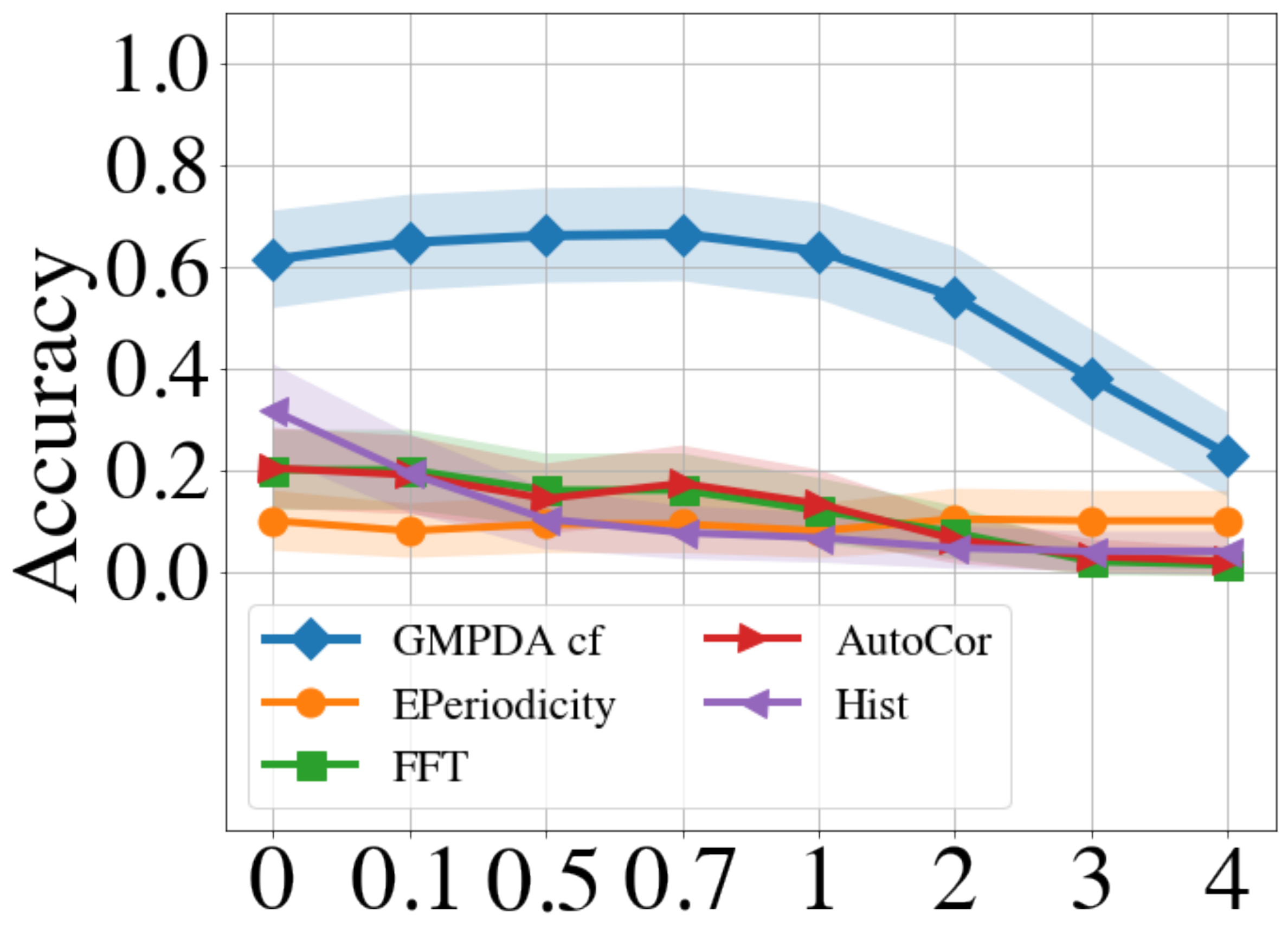}}}
    \newline
    \subfloat[n=50]{{\includegraphics[width=0.23\textwidth]{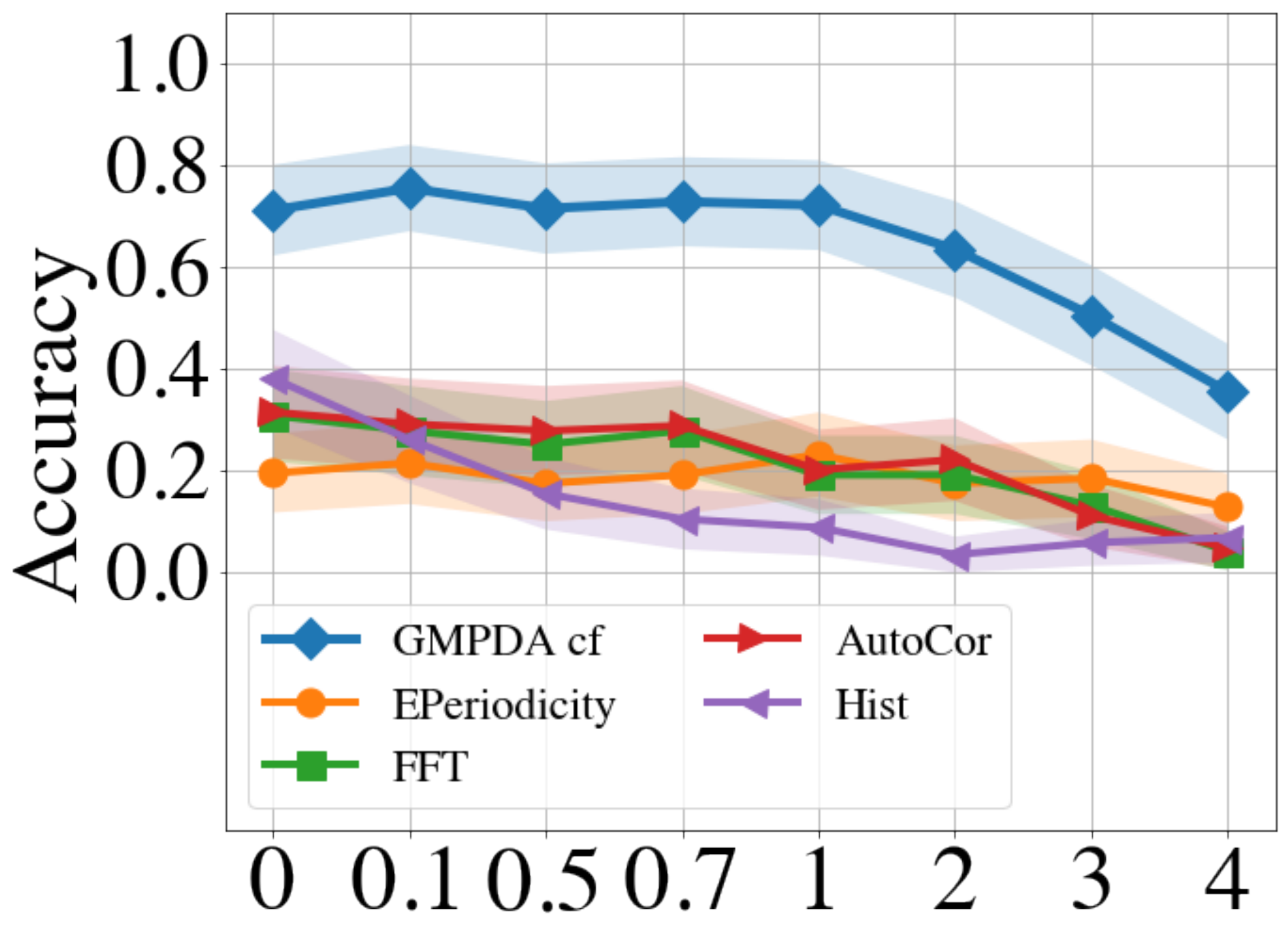}}}
    \subfloat[n=100]{{\includegraphics[width=0.23\textwidth]{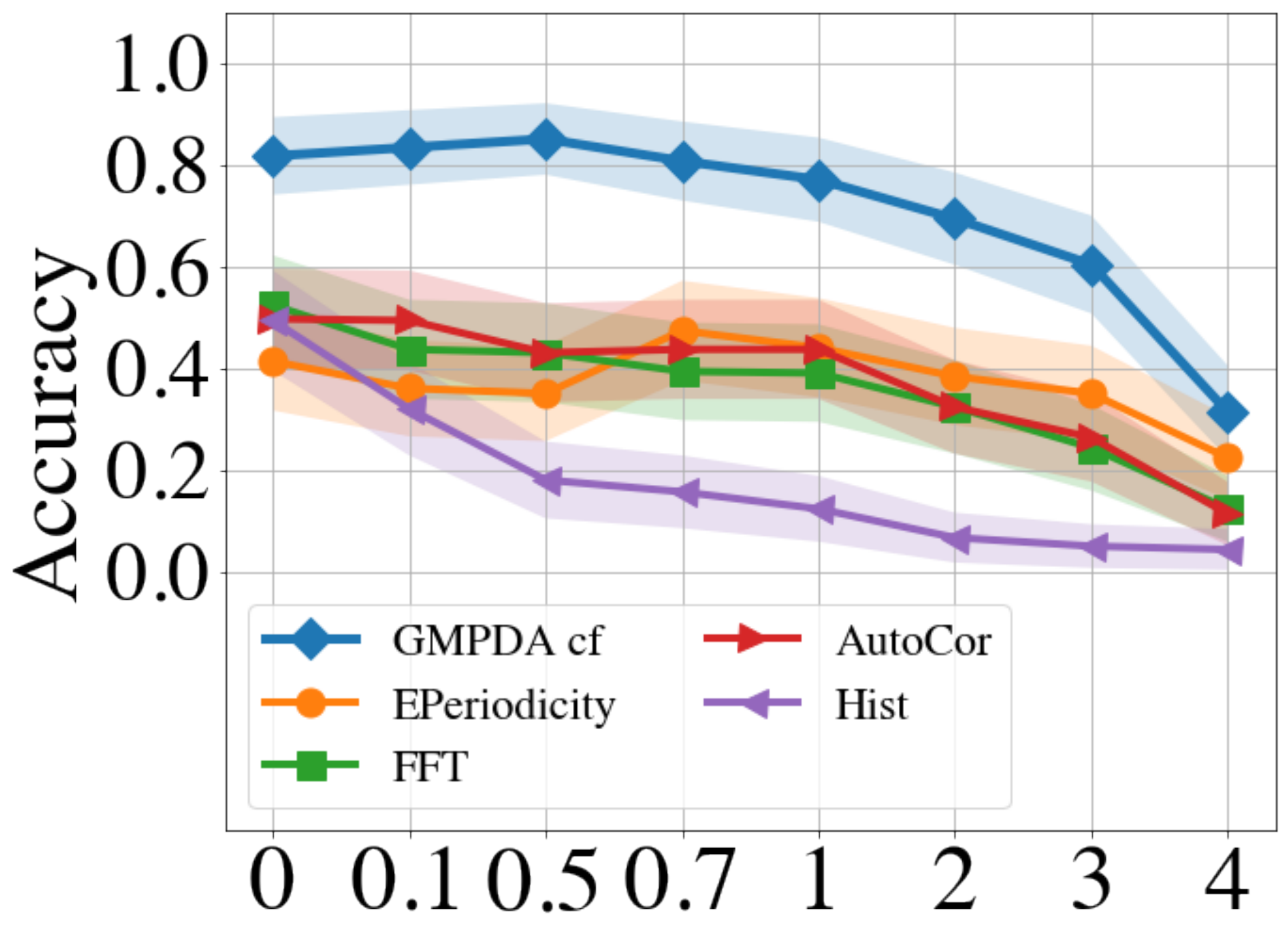}}}
    \newline
    \subfloat[n=300]{{\includegraphics[width=0.23\textwidth]{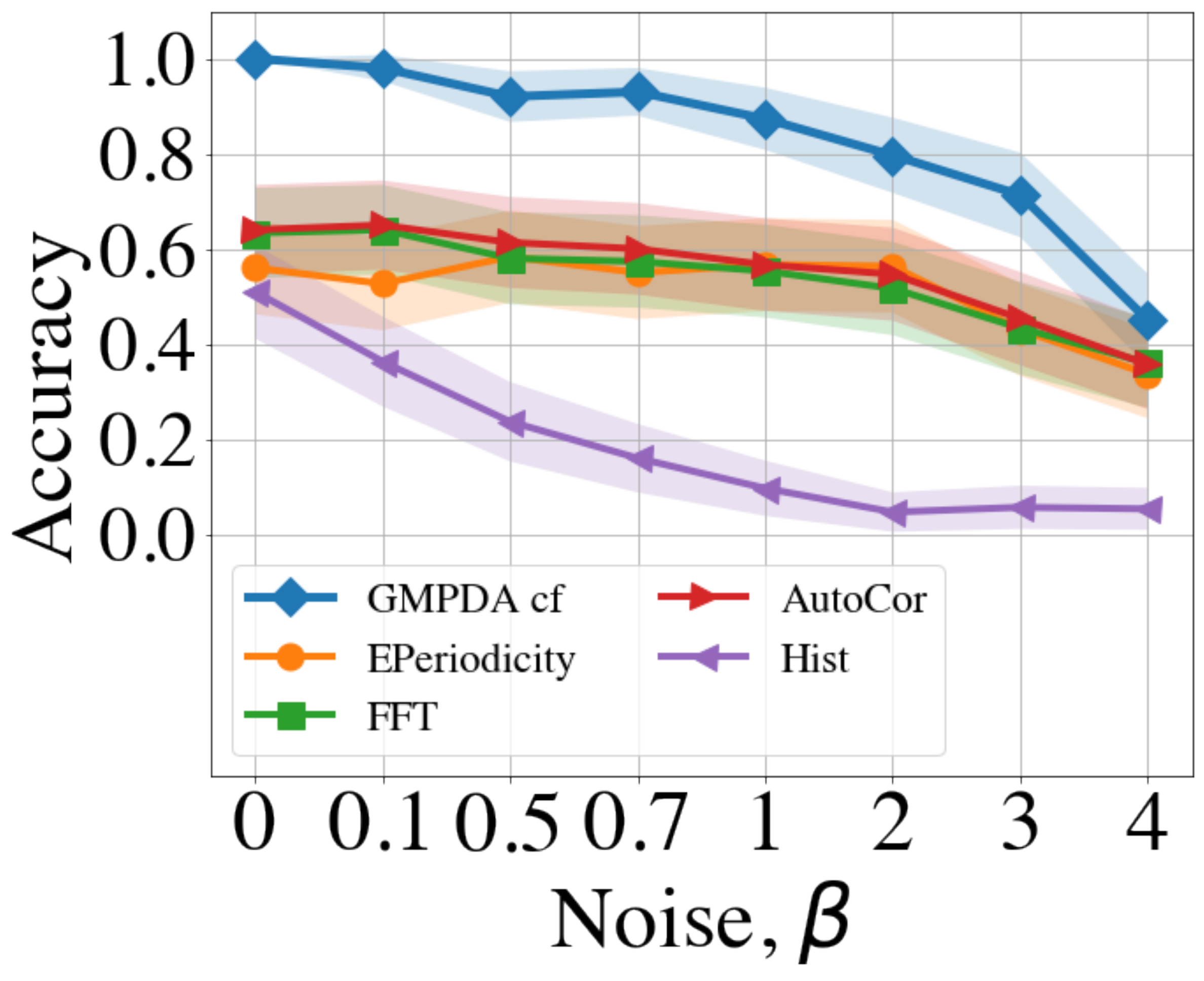}}}
    \subfloat[n=500]{{\includegraphics[width=0.23\textwidth]{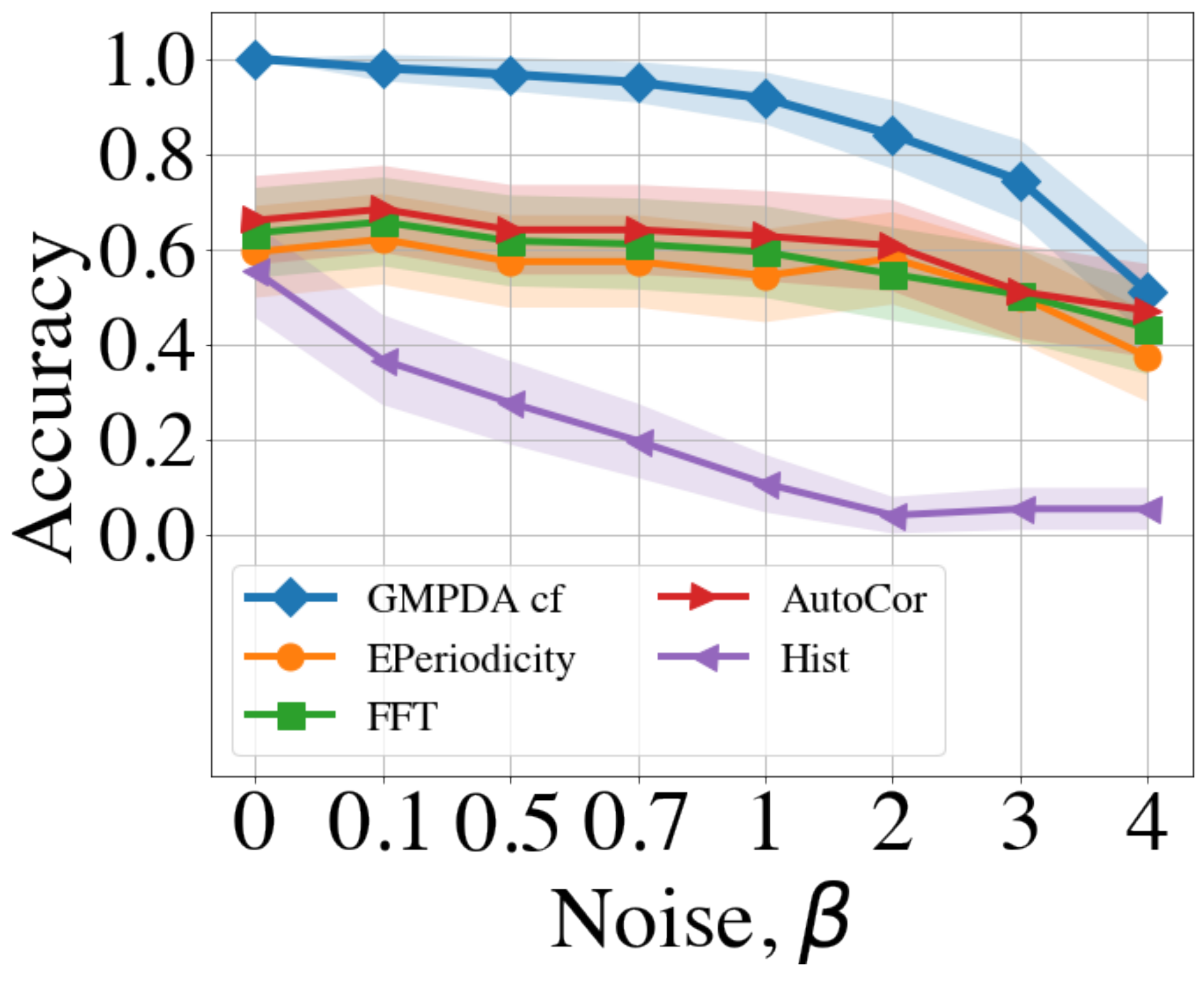}}}
    \caption{Random Walk Model Performance w.r.t. $\beta$, averaged over $\sigma$, $|\mu|=1$.}
    \label{fig:perf_cl_noise_cf_ncf_app_rw}
\end{figure}
\begin{figure}[h]
    \centering
    \subfloat[n=10]{{\includegraphics[width=0.23\textwidth]{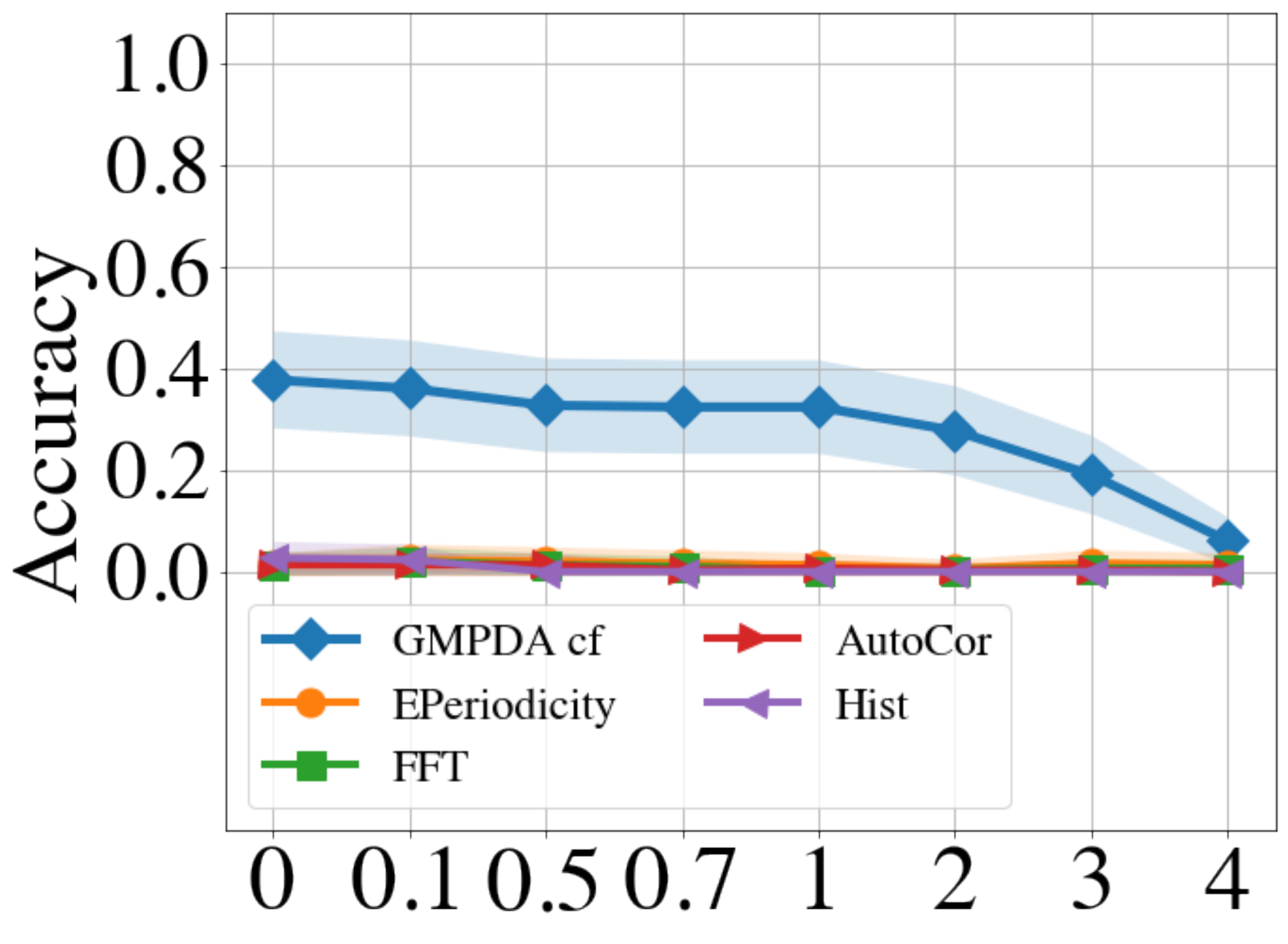}}}
    \subfloat[n=30]{{\includegraphics[width=0.23\textwidth]{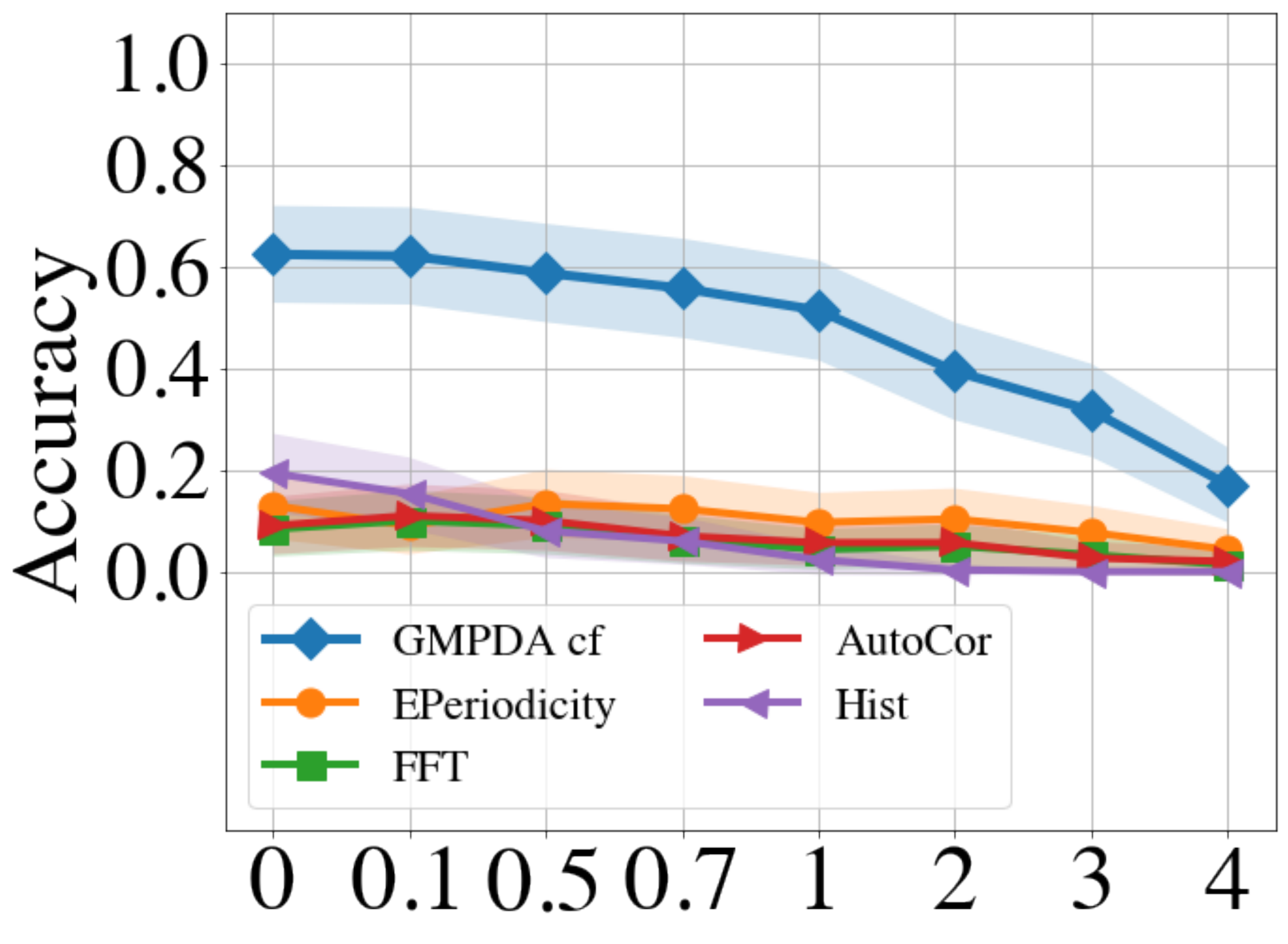}}}
    \newline
    \subfloat[n=50]{{\includegraphics[width=0.23\textwidth]{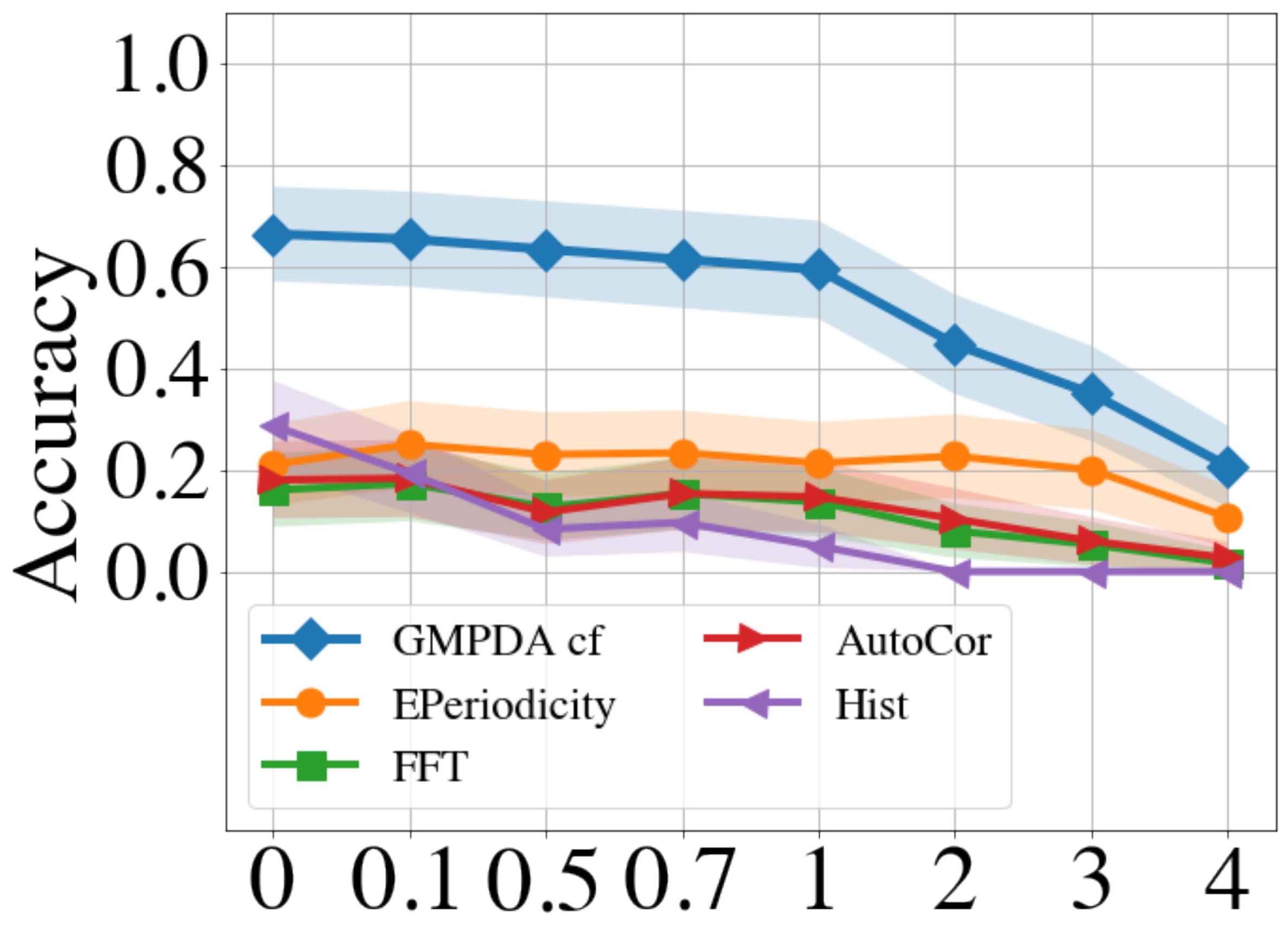}}}
    \subfloat[n=100]{{\includegraphics[width=0.23\textwidth]{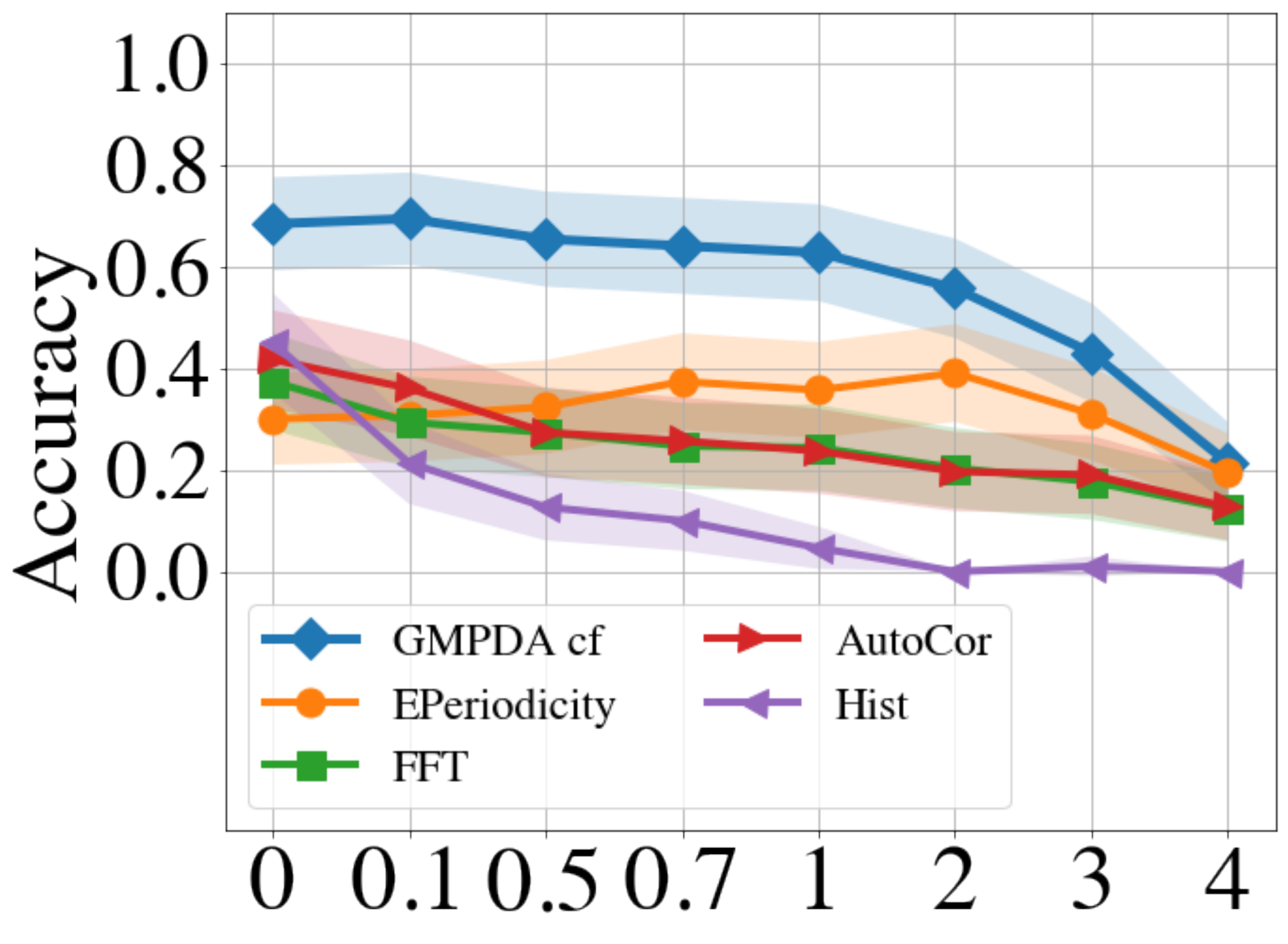}}}
    \newline
    \subfloat[n=300]{{\includegraphics[width=0.23\textwidth]{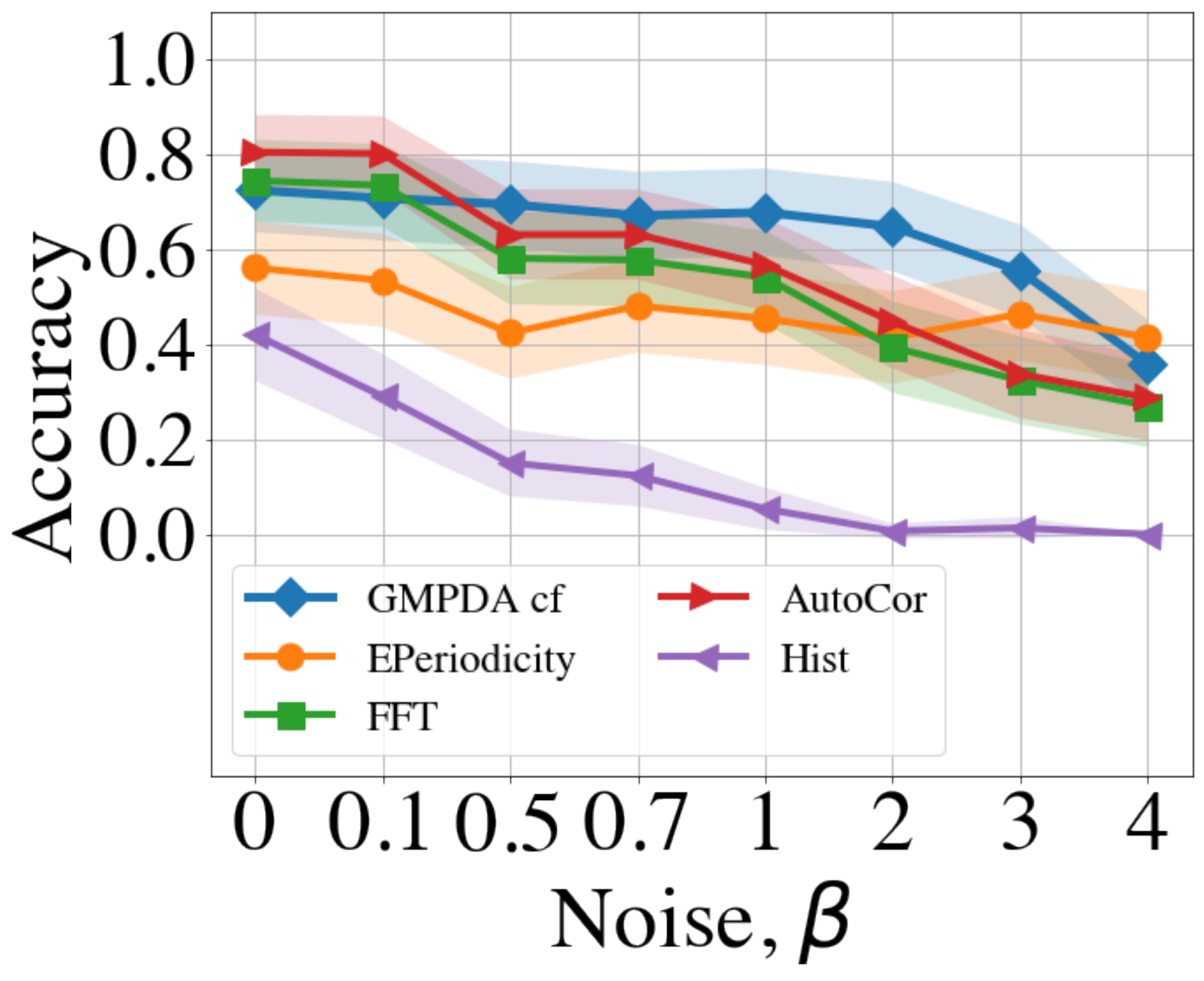}}}
    \subfloat[n=500]{{\includegraphics[width=0.23\textwidth]{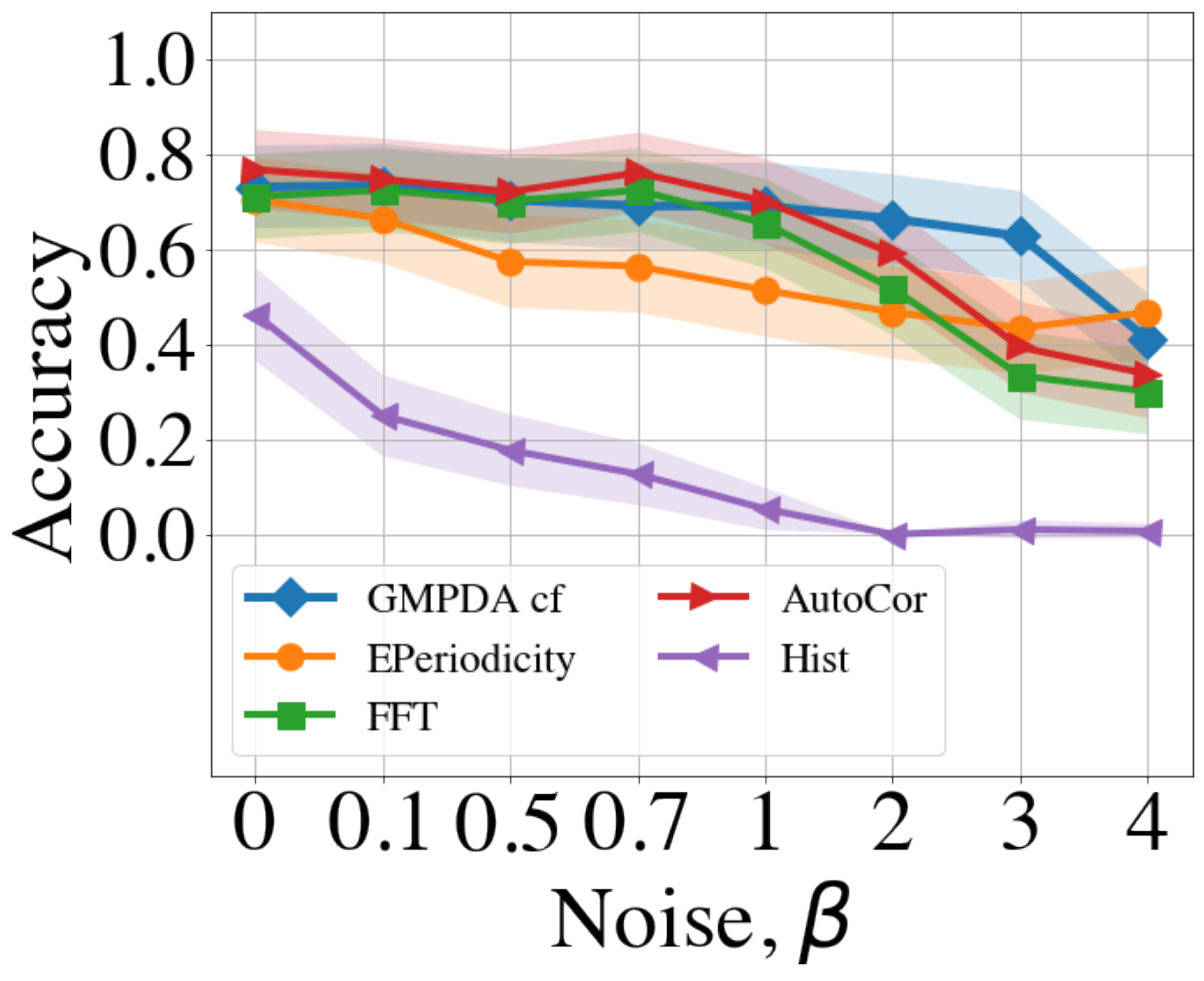}}}
    \caption{Clock Model Performance w.r.t. $\beta$, averaged over $\sigma$, $|\mu|=1$.}
    \label{fig:perf_cl_noise_cf_ncf_app_clock}
\end{figure}
\subsection{Comparison to alternative Methods wrt. Variance $\sigma$}
\label{app:variance_sigma}
Here the performance of GMPDA and of the alternative methods is compared regarding an increase in variance.
In the following figures the accuracy of all the involved methods is plotted for $|\mu=1|$, different number of events $n$, while it was averaged over all considered values of noise $\beta$.
\begin{figure}[h]
    \centering
    \subfloat[n=10]{{\includegraphics[width=0.23\textwidth]{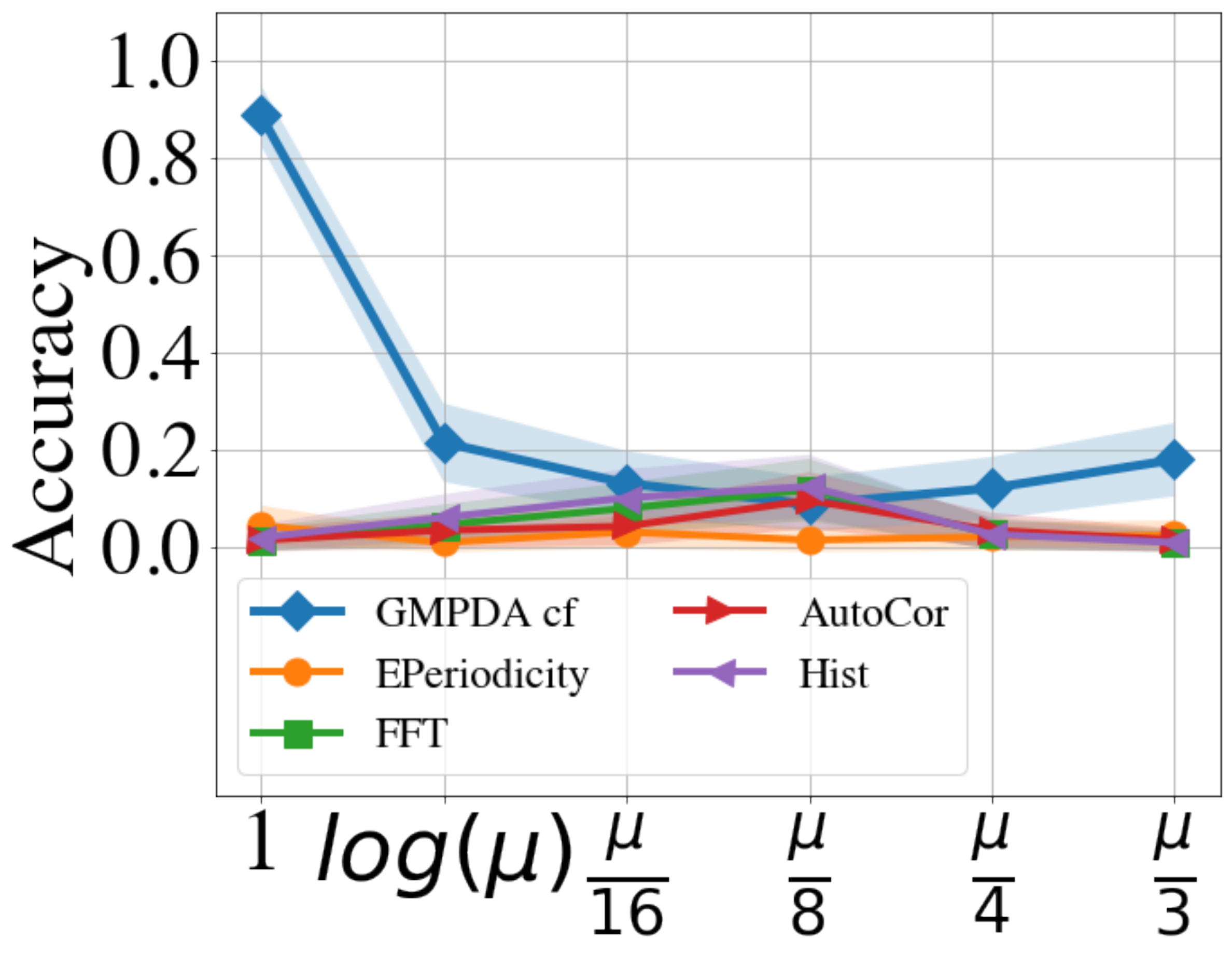}}}
    \subfloat[n=30]{{\includegraphics[width=0.23\textwidth]{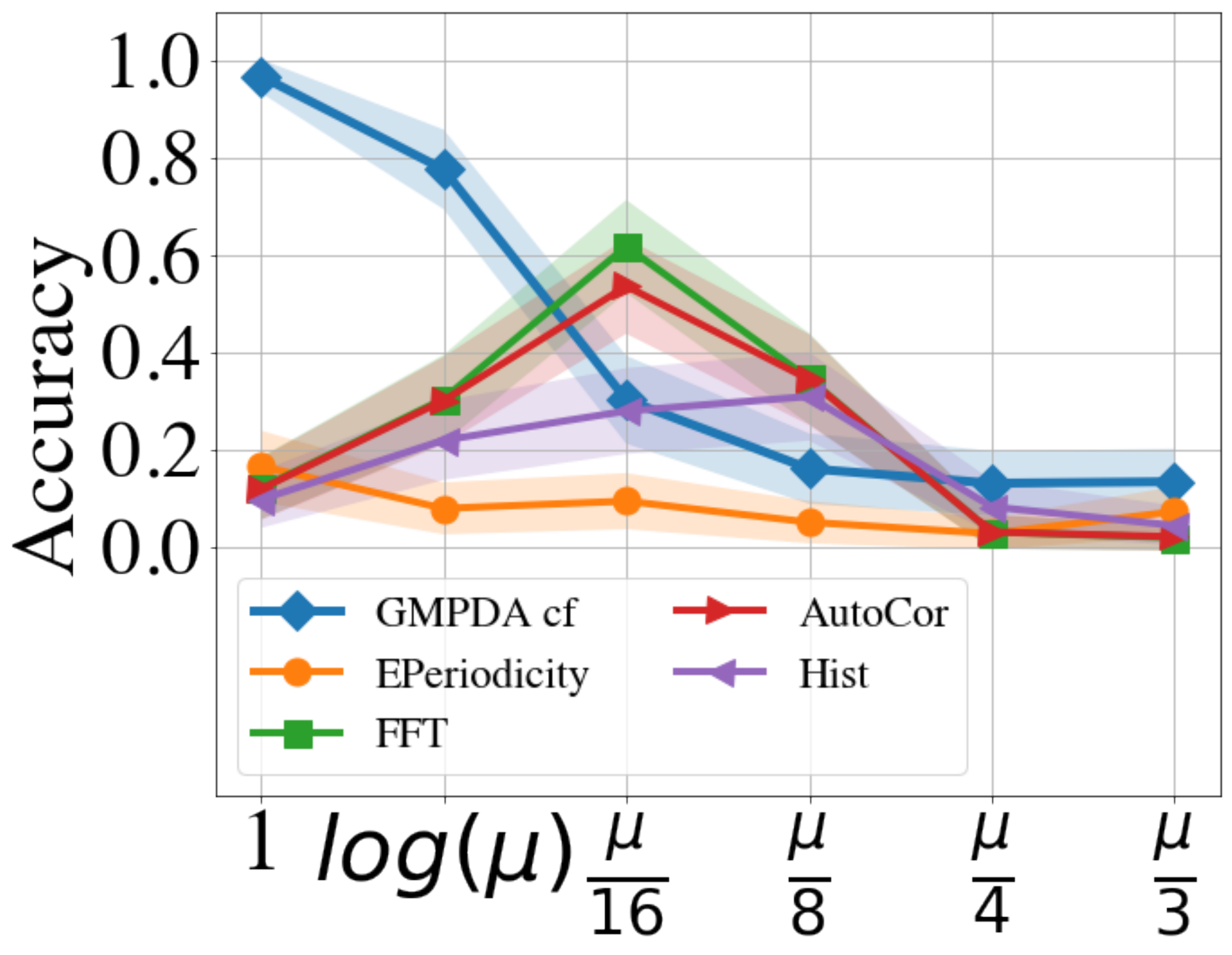}}}
    \newline
    \subfloat[n=50]{{\includegraphics[width=0.23\textwidth]{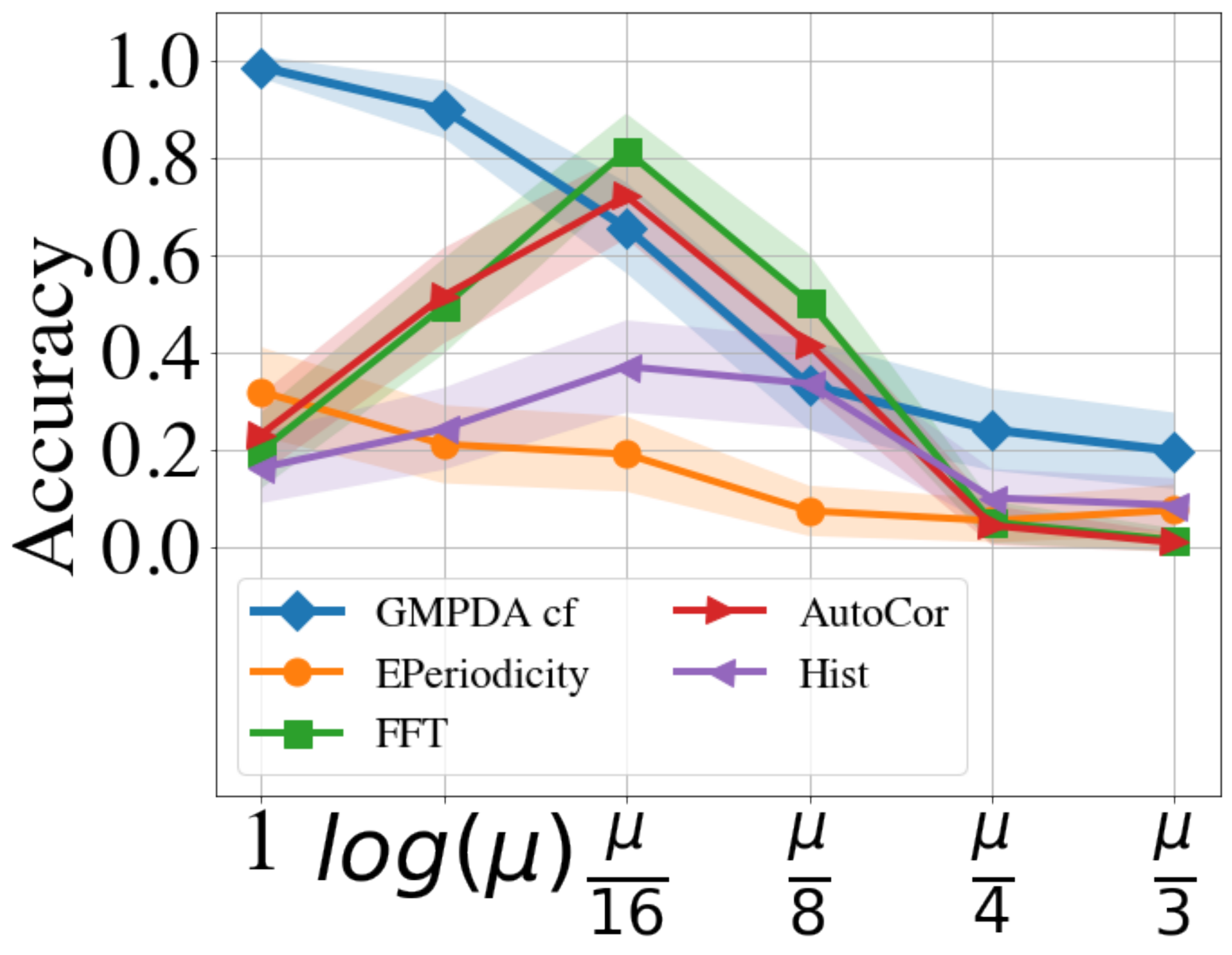}}}
    \subfloat[n=100]{{\includegraphics[width=0.23\textwidth]{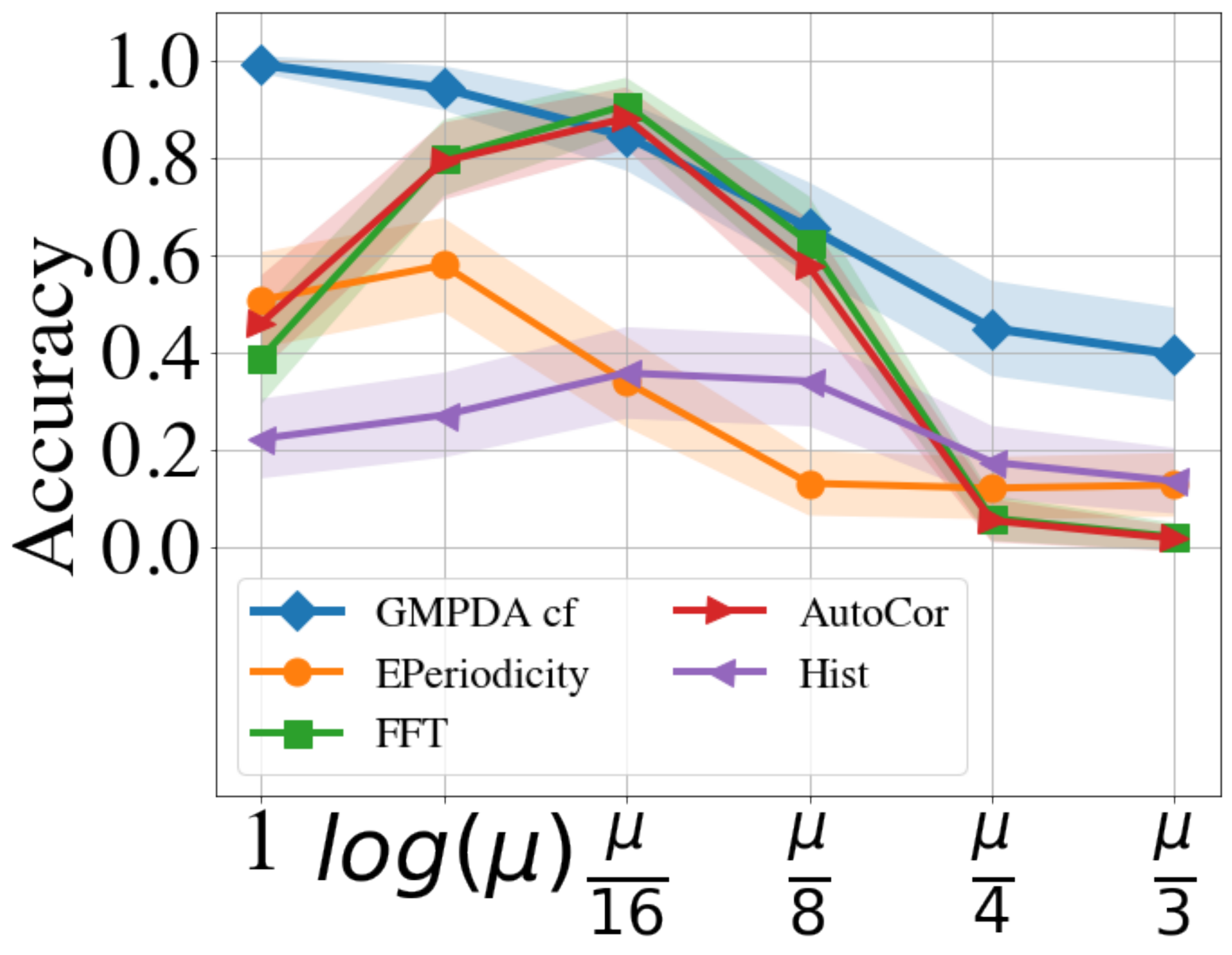}}}
    \newline
    \subfloat[n=300]{{\includegraphics[width=0.23\textwidth]{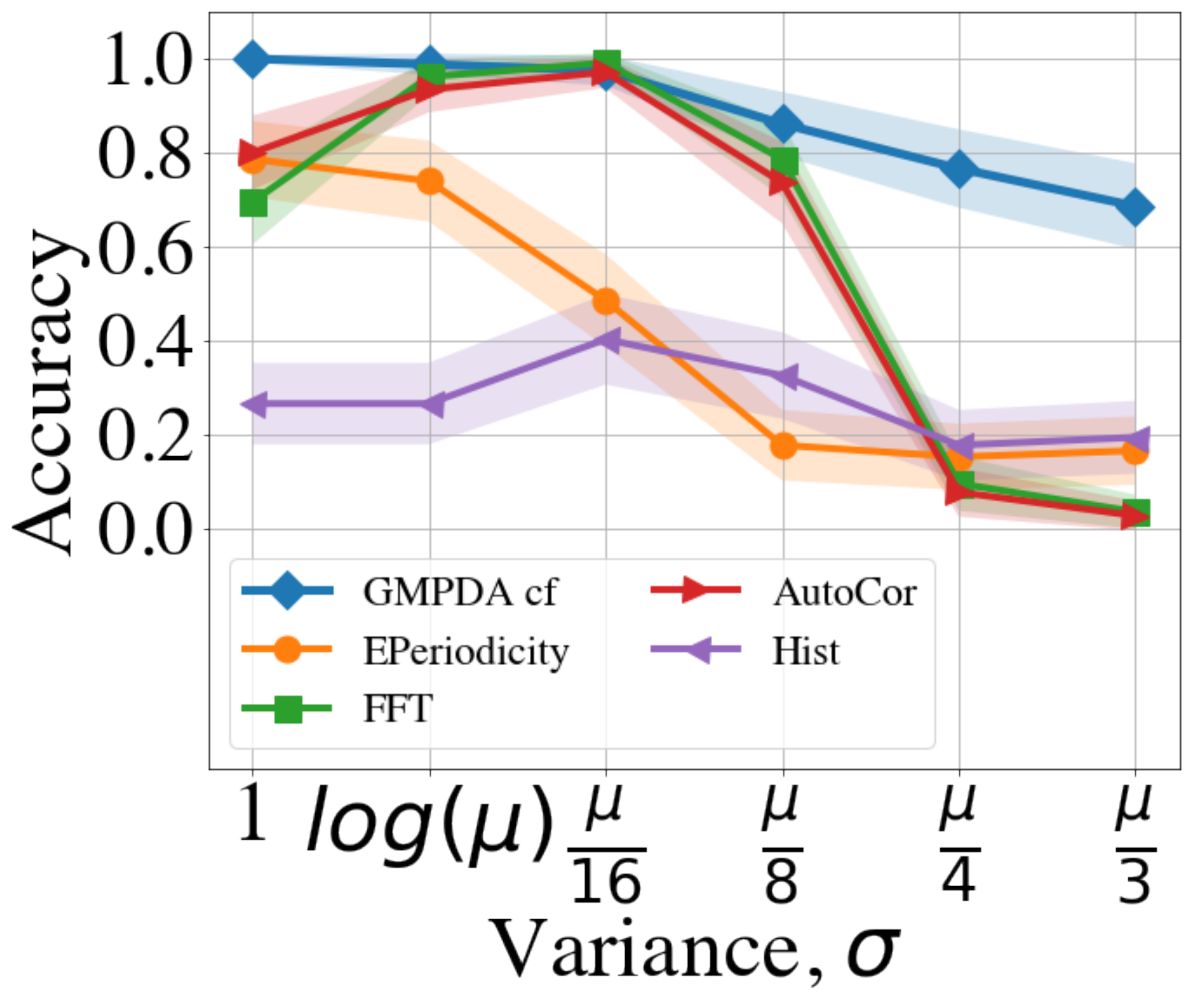}}}
    \subfloat[n=500]{{\includegraphics[width=0.23\textwidth]{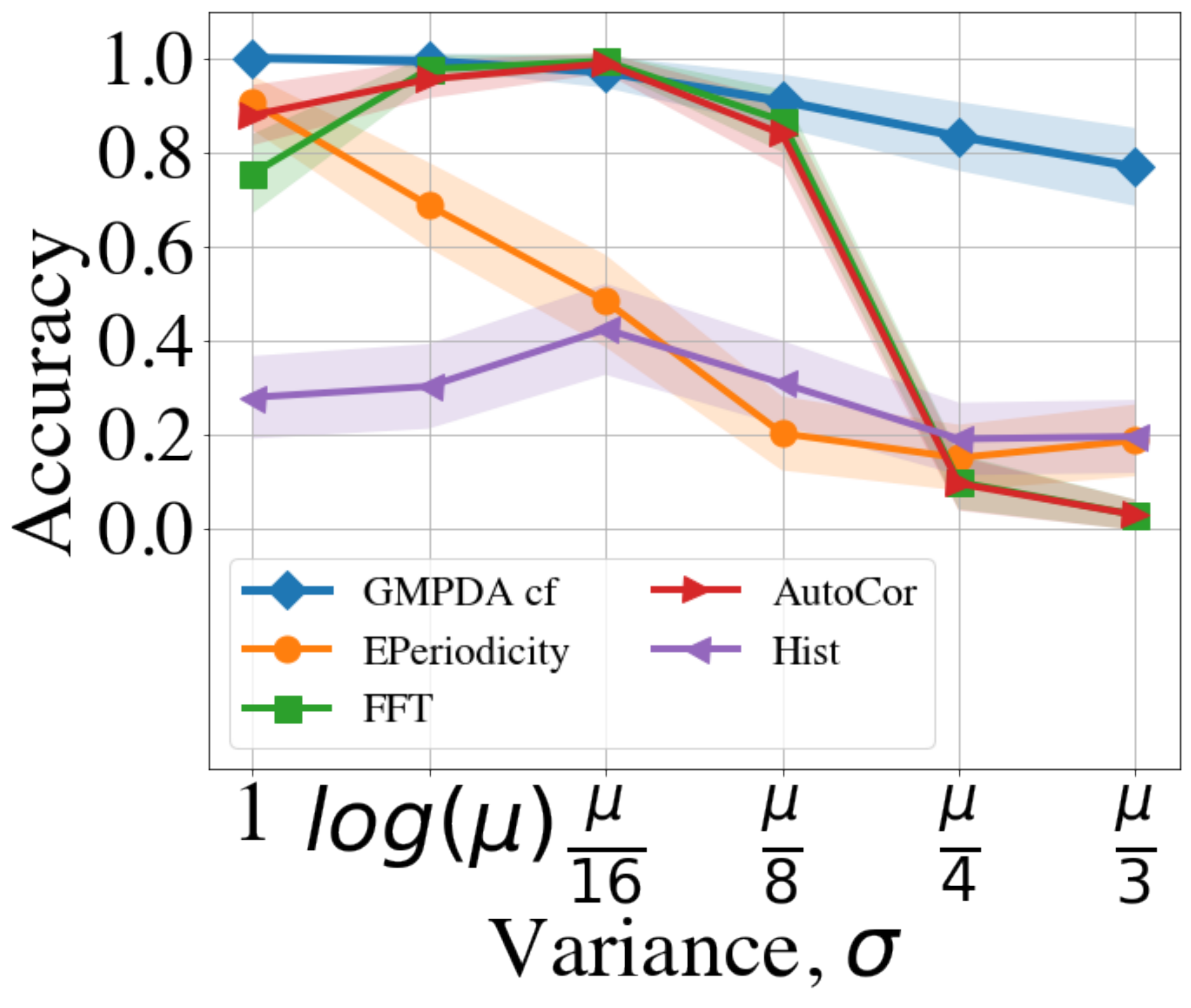}}}
    \caption{Random Walk Model Performance w.r.t. $\sigma$, averaged over noise, $|\mu|=1$.}
    \label{fig:perf_cl_sigma_cf_ncf_app_rw}
\end{figure}
\begin{figure}[ht]
    \centering
  \subfloat[n=10]{{\includegraphics[width=0.23\textwidth]{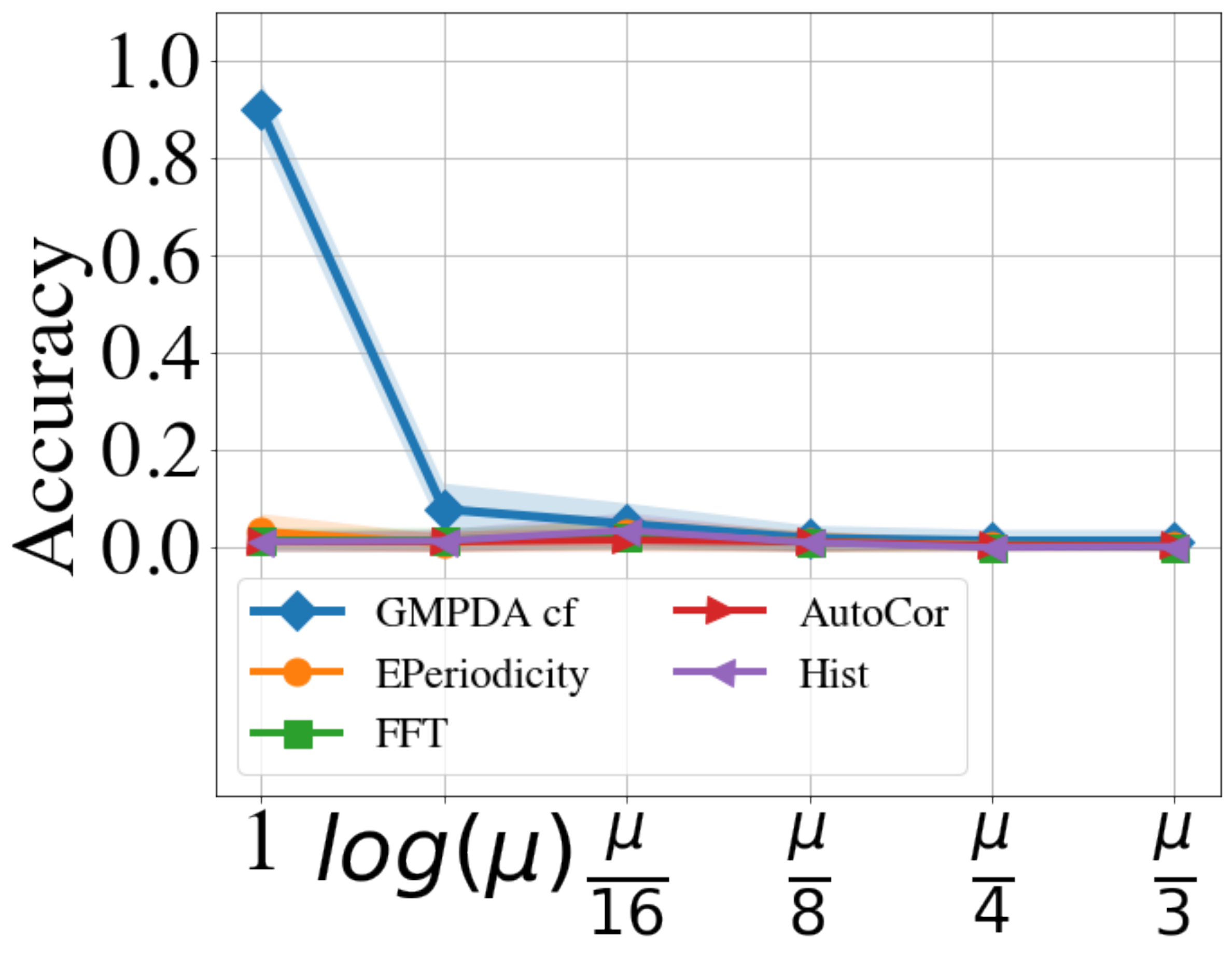}}}
    \subfloat[n=30]{{\includegraphics[width=0.23\textwidth]{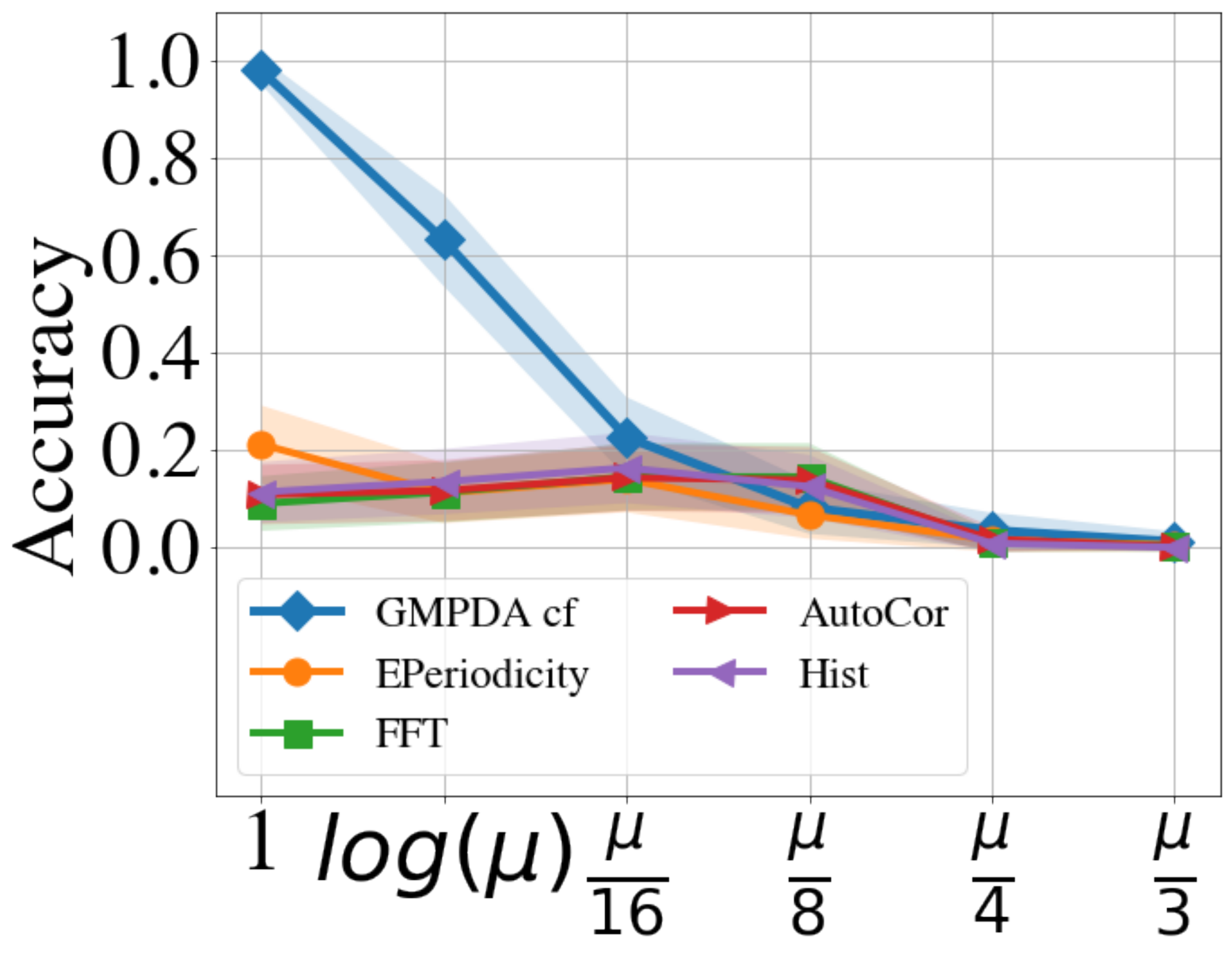}}}
    \newline
    \subfloat[n=50]{{\includegraphics[width=0.23\textwidth]{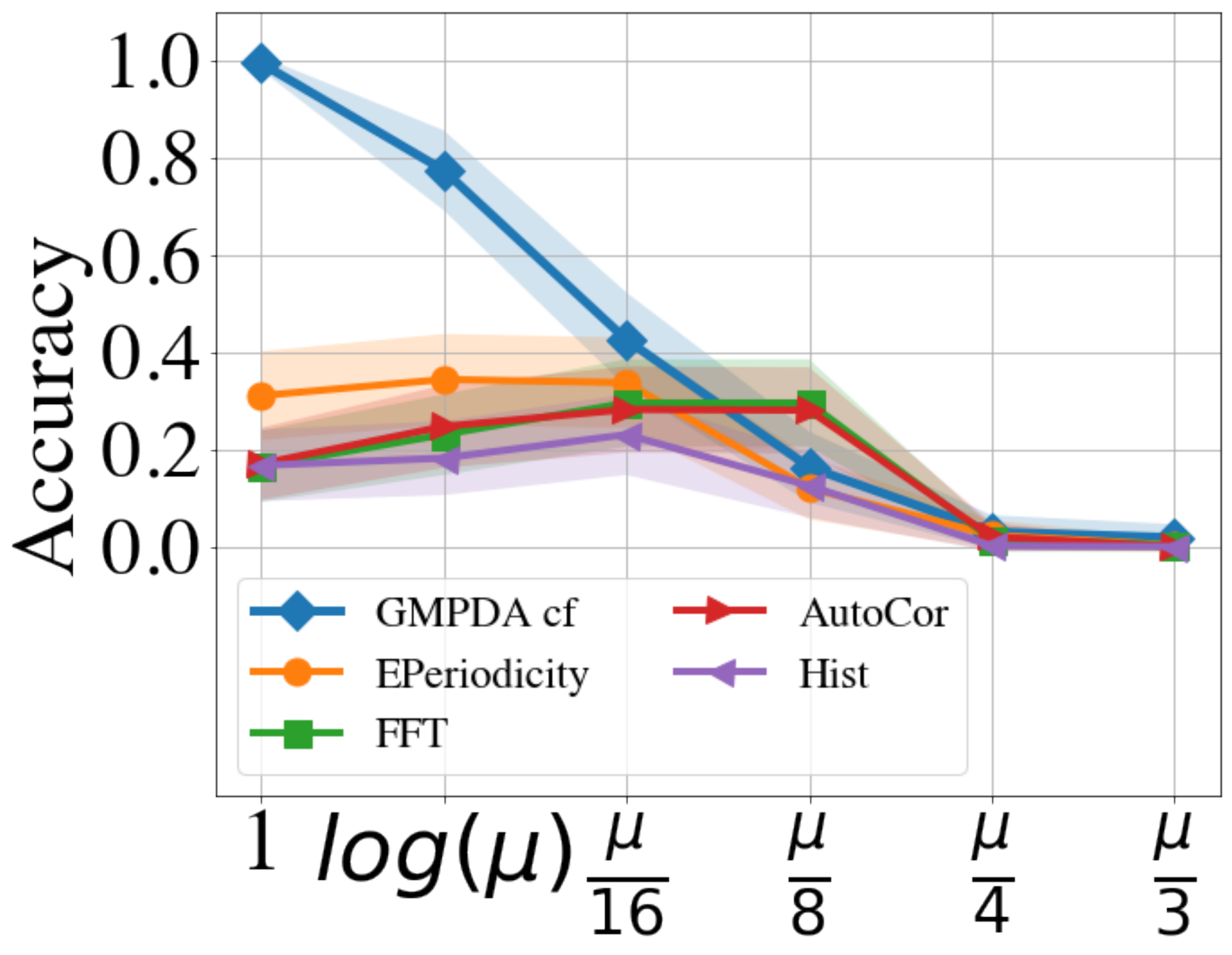}}}
    \subfloat[n=100]{{\includegraphics[width=0.23\textwidth]{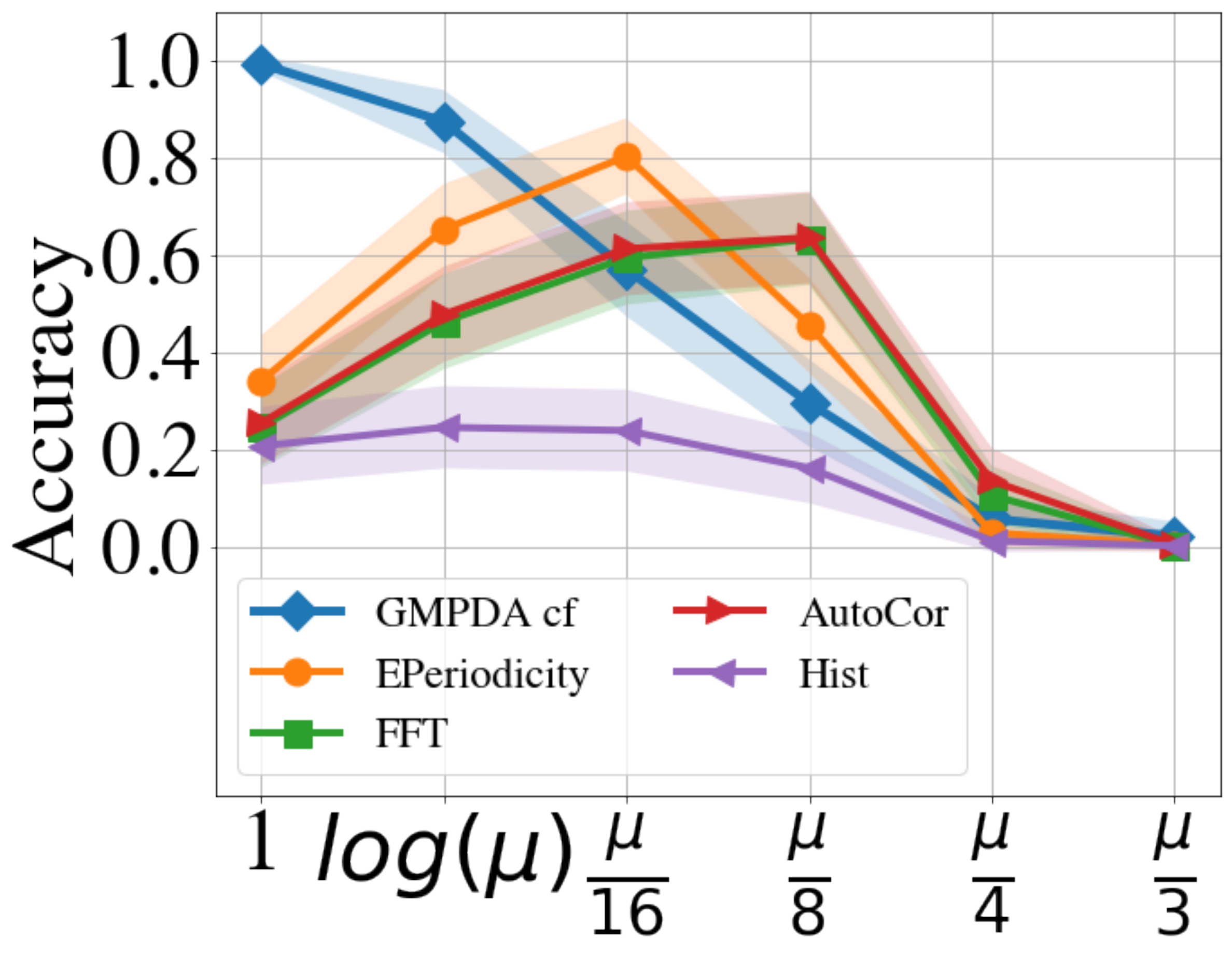}}}
    \newline
    \subfloat[n=300]{{\includegraphics[width=0.23\textwidth]{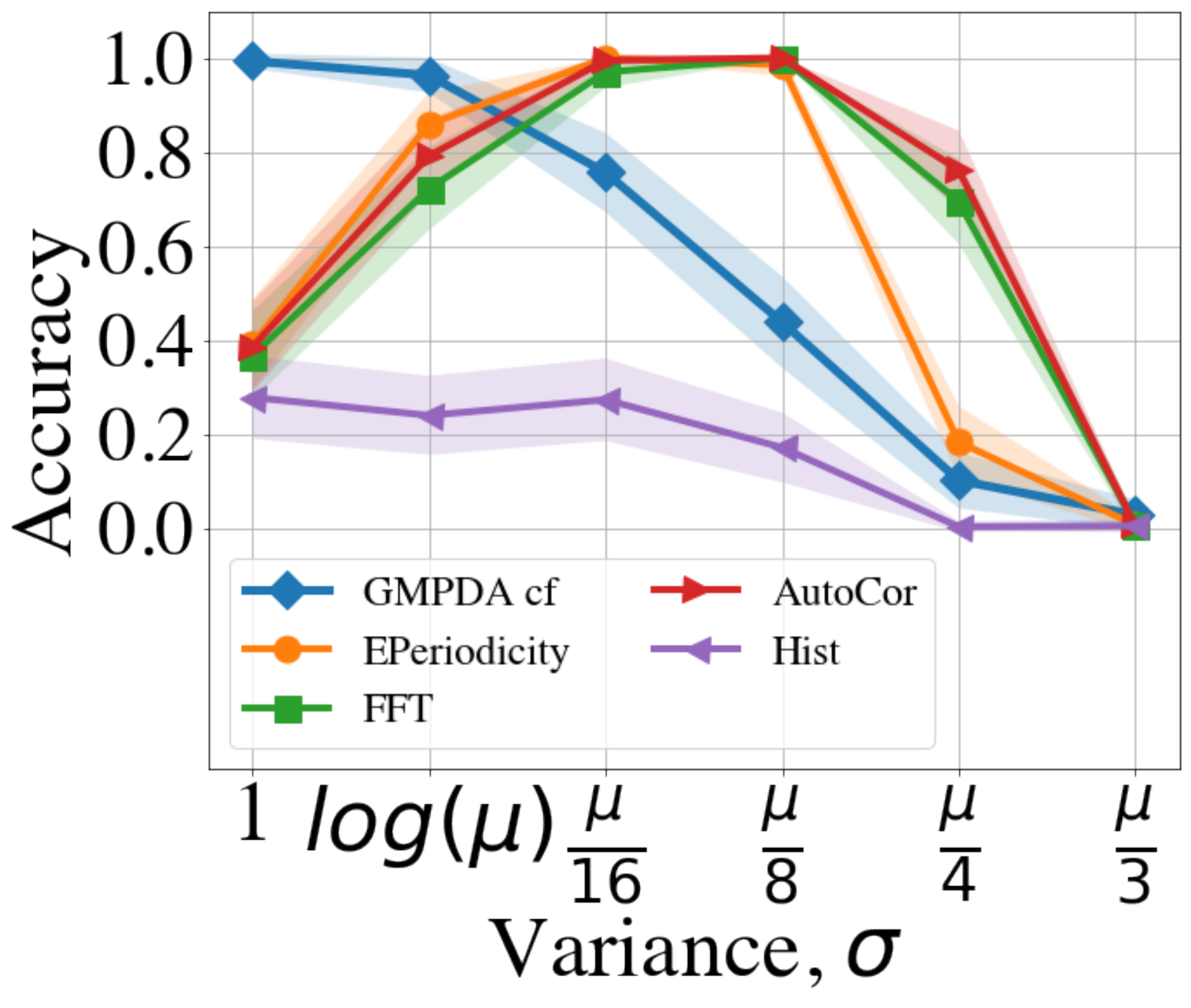}}}
    \subfloat[n=500]{{\includegraphics[width=0.23\textwidth]{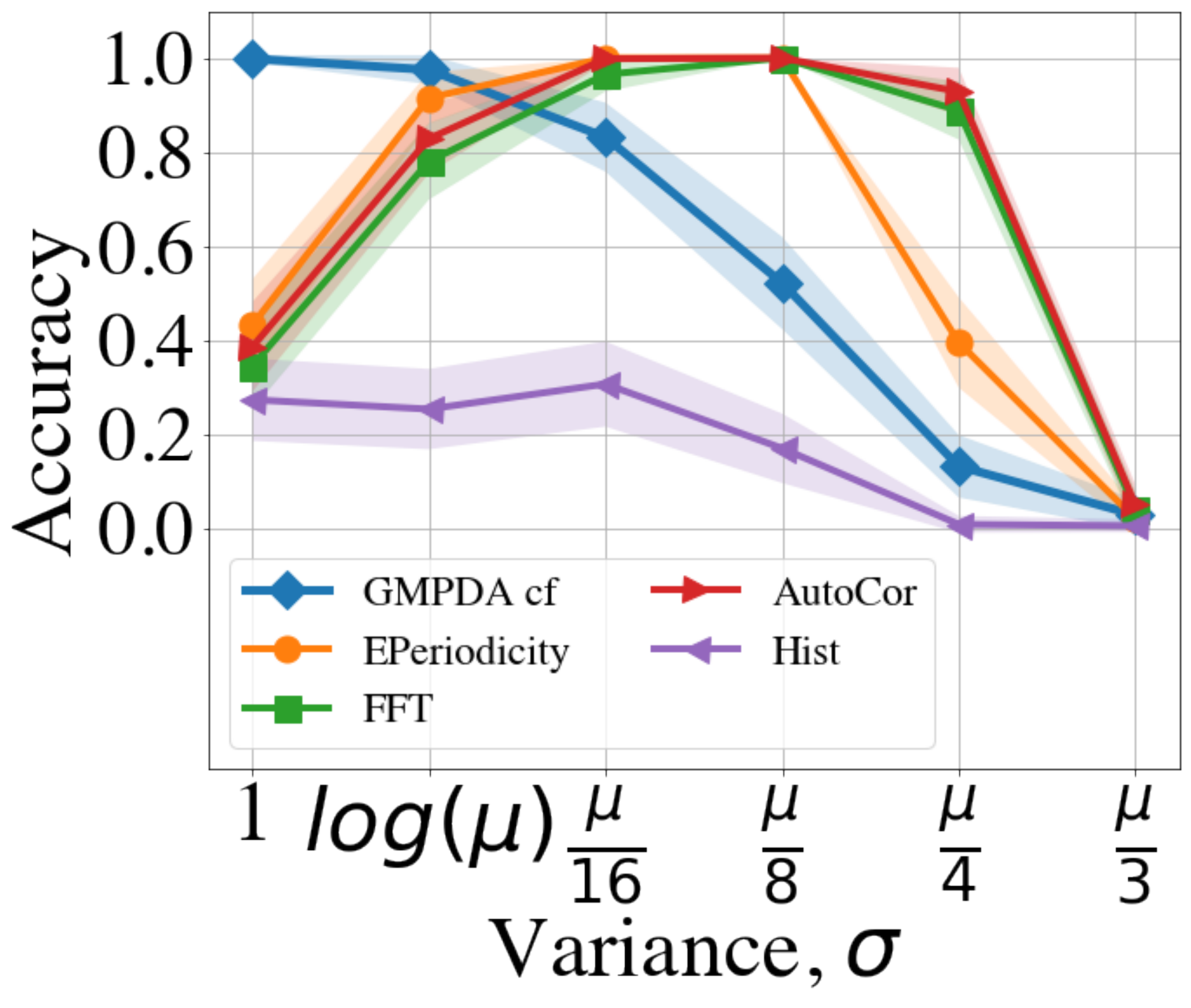}}}
    \caption{Clock Model Performance w.r.t. $\sigma$, averaged over noise, $|\mu|=1$.}
    \label{fig:perf_cl_sigma_cf_ncf_app_noise}
\end{figure}

\subsection{Sensitivity Analysis}
\label{app:sensitivity_ana}
We summarized differences in sensitivity to critical simulation parameters across the different algorithms in Table~\ref{tab:sensAna}. Based on the simulation results obtained in Section~\ref{sec:test_cases}, we used generalized linear mixed models with number of periods $|\mu|=1,2,3$ nested within trials ($n=38400$) with the response being the accurate detection of a single periodicity (coded as 1 if the estimate is within an intervals around the true value $\pm\sigma$) and the independent factors being
\begin{itemize}
    \item number of events $n=10, 30, 50, 100, 300,500$, 
    \item number of periods $|\mu|=1,2,3$,
    \item variance $\sigma=1, \log(\mu), \frac{\mu}{p}, \text{ with } p= 3, 8, 16$,
    \item noise $\beta=0, 0.1, 0.5, 0.7, 1, 2, 4, 8$.
\end{itemize}
Mixed logistic models were computed separately for the clock model and the random walk model and each algorithm. Table~\ref{tab:sensAna} lists the Anova type II sum of squares (SoS), i.e. the SoS of each main effect after the introduction of all other main effects. While the SoS are not directly comparable between models, their relative contribution is and suggests that for the Random Walk Model the number of events had a major effect in all algorithms except the FFT histogram algorithm. The number of periods had a small to moderate effect except for the E-Periodicity where it did not play a role. Both E-Periodicity and GMPDA with curve-fitting were very sensitive to the noise level and across algorithms variations in sigma had one of the strongest effects on accuracy, again with the exception of the E-Periodicity algorithm.

The results for the Clock Model were largely similar with some notable exceptions. Compared to the Random Walk Models, the E-Periodicity algorithm was considerably less sensitive to variations in noise but more sensitive to variations in variance. Overall, the three algorithms E-Periodicity, FFT, and FFT auto-correlation showed a similar patter, with the number of events having the strongest influence, with sigma being the second strongest, and noise and number of periods having only relatively minor effects. For the two GMPDA algorithms the strongest effect was seen for sigma.

Across all models and algorithms, both sigma and the number events emerged as the strongest determinants of periodicity detection accuracy.
\begin{table}[ht]
\resizebox{0.5\textwidth}{!}{
\begin{tabular}{lrrrr}
\hline
                            & \multicolumn{1}{c}{\begin{tabular}[c]{@{}c@{}}Number of\\ Events\end{tabular}} & \multicolumn{1}{c}{\begin{tabular}[c]{@{}c@{}}Number of\\ Periods\end{tabular}} & \multicolumn{1}{c}{Noise} & \multicolumn{1}{c}{\begin{tabular}[c]{@{}c@{}}Sigma\\ ratio\end{tabular}} \\
\multicolumn{1}{r}{Df}      & \multicolumn{1}{c}{5}                                                          & \multicolumn{1}{c}{2}                                                           & \multicolumn{1}{c}{7}     & \multicolumn{1}{c}{5}                                                     \\ \hline
\textbf{Random Walk Models} & \multicolumn{1}{l}{}                                                           & \multicolumn{1}{l}{}                                                            & \multicolumn{1}{l}{}      & \multicolumn{1}{l}{}                                                      \\
GMPDA with curve-fitting     & 5110                                                                           & 3746                                                                            & 5096                      & 4756                                                                      \\
GMPDA w/o curve-fitting      & 3349                                                                           & 1821                                                                            & 1291                      & 7830                                                                      \\
E-periodicity               & 3218                                                                           & 10                                                                              & 5221                      & 438                                                                       \\
FFT                         & 5519                                                                           & 1420                                                                            & 1359                      & 6687                                                                      \\
FFT Autocorrelation         & 5506                                                                           & 1252                                                                            & 3058                      & 6013                                                                      \\
FFT Histogram               & 619                                                                            & 1159                                                                            & 958                       & 2994                                                                      \\
                            & \multicolumn{1}{l}{}                                                           & \multicolumn{1}{l}{}                                                            & \multicolumn{1}{l}{}      & \multicolumn{1}{l}{}                                                      \\
\textbf{Clock Models}       & \multicolumn{1}{l}{}                                                           & \multicolumn{1}{l}{}                                                            & \multicolumn{1}{l}{}      & \multicolumn{1}{l}{}                                                      \\
GMPDA with curve-fitting     & 2714                                                                           & 2377                                                                            & 3186                      & 7399                                                                      \\
GMPDA w/o curve-fitting      & 2089                                                                           & 454                                                                             & 1210                      & 6806                                                                      \\
E-periodicity               & 7659                                                                           & 280                                                                             & 651                       & 4707                                                                      \\
FFT                         & 8520                                                                           & 157                                                                             & 1522                      & 5216                                                                      \\
FFT Autocorrelation         & 8480                                                                           & 87                                                                              & 1283                      & 5689                                                                      \\
FFT Histogram               & 438                                                                            & 211                                                                             & 2829                      & 1035                                                                      \\ \hline
\end{tabular}
}
\caption{Differences in sensitivity to critical simulation parameters across the different algorithms.}
\label{tab:sensAna}
\end{table}
\subsection{Real Application: Loss}
\label{app:MrOS_loss}
This section shows the GMPDA loss obtained with and without curve fitting for MROS dataset, Figure~\ref{fig:MROS_loss_cf_vs_ncf.pdf} shows the loss for 100 recordings, while Figure~\ref{fig:MROS_loss_cf_vs_ncf_bouts.pdf} shows the loss comparison for the bouts.%
\begin{figure}[h]
    \centering
    \includegraphics[width=0.23\textwidth]{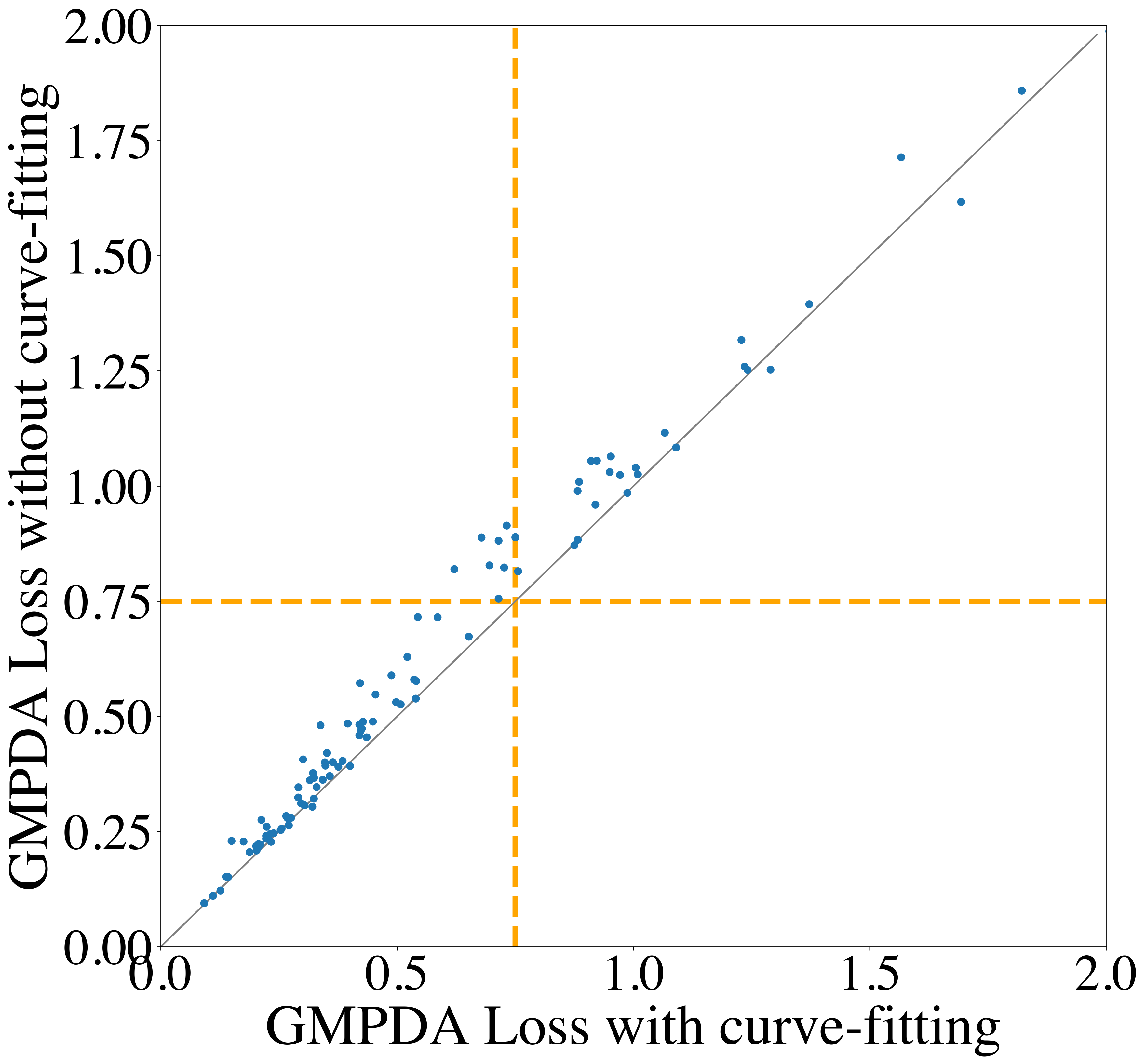}
    \caption{Comparison of the GMPD loss with and without curve fitting for individuals.}
    \label{fig:MROS_loss_cf_vs_ncf.pdf}
\end{figure}
\begin{figure}[h]
    \centering
    \includegraphics[width=0.23\textwidth]{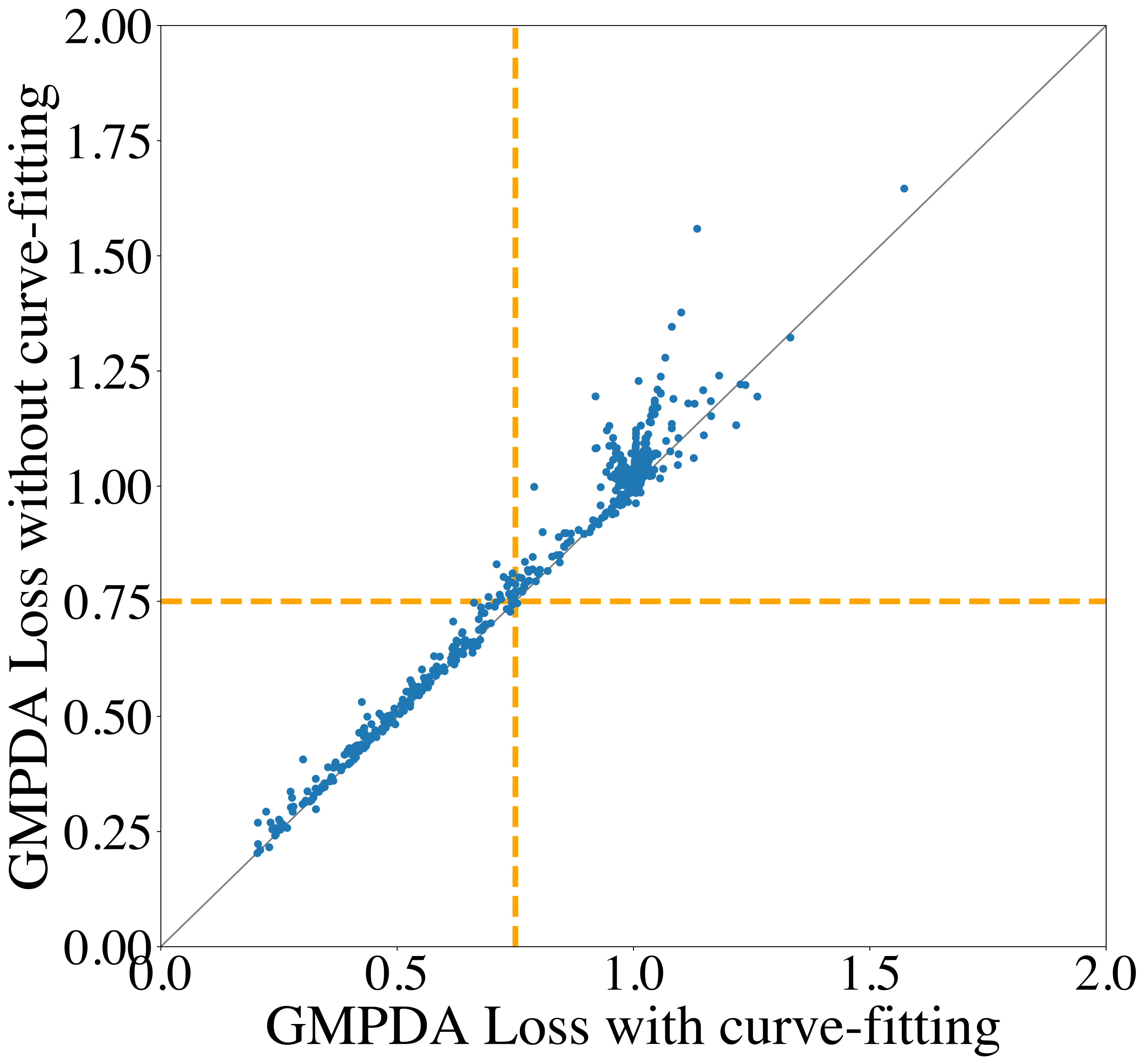}
    \caption{Comparison of the GMPD loss with and without curve fitting for bouts.}
    \label{fig:MROS_loss_cf_vs_ncf_bouts.pdf}
\end{figure}

\ifCLASSOPTIONcompsoc
  \section*{Acknowledgments}
\else
  \section*{Acknowledgment}
\fi

Dr. Fulda is supported by Swiss National Science Foundation (SNSF) grants No. $320030\_160009$ and $320030\_179194$.

\ifCLASSOPTIONcaptionsoff
  \newpage
\fi



\bibliographystyle{IEEEtran}
\newpage
\bibliography{references.bib}

\end{document}